\documentclass[a4paper, 11pt]{article}
\usepackage[usenames]{xcolor}
\usepackage{xifthen}
\usepackage{algorithmic}
\usepackage{comment}
\usepackage{tcolorbox}
\usepackage{relsize}
\usepackage[font=footnotesize]{caption}
\usepackage[normalem]{ulem}
\usepackage{todonotes}

\usepackage{bm}
\usepackage{booktabs}
\usepackage{tabularx}
\usepackage{tikz}
\usetikzlibrary{shapes.geometric, arrows.meta, positioning, decorations.pathreplacing, calc}

\def\fontsettingup{2} 

\usepackage{amsmath, amsthm, amssymb}
\usepackage{bbm}
\usepackage[T1]{fontenc}
\ifthenelse{\fontsettingup = 1}{ \usepackage{eulerpx, eucal, tgpagella}   }{}
\ifthenelse{\fontsettingup = 2}{ \usepackage{mathpazo, eufrak, tgpagella} }{}

\usepackage{hyperref}
\hypersetup{
  colorlinks=true,
  linkcolor=blue,
  filecolor=magenta,
  urlcolor=cyan,
  citecolor=blue
}

\usepackage[linesnumbered,boxed,ruled,vlined]{algorithm2e}
\usepackage{tikz}
\usepackage{cleveref}
\usepackage{makecell}

\usepackage[a4paper]{geometry}
\newgeometry{
textheight=9in,
textwidth=6.5in,
top=1in,
headheight=14pt,
headsep=25pt,
footskip=30pt
}

\newtheorem{theorem}{Theorem}

\newtheorem*{claim*}{Claim}

\newtheorem{lemma}[theorem]{Lemma}
\newtheorem{proposition}[theorem]{Proposition}
\newtheorem{corollary}[theorem]{Corollary}
\theoremstyle{definition}

\newtheorem{definition}[theorem]{Definition}
\newtheorem{remark}[theorem]{Remark}
\newtheorem*{remark*}{Remark}
\newtheorem{assumption}{Assumption}

\ifthenelse{\fontsettingup = 1}{
  \def\*#1{\mathbf{#1}} 
  \def\+#1{\mathcal{#1}} 
  \def\-#1{\mathrm{#1}} 
  \def\^#1{\mathbb{#1}} 
  \def\!#1{\mathfrak{#1}} 
}{}

\ifthenelse{\fontsettingup = 2}{
  \def\*#1{\boldsymbol{#1}} 
  \def\+#1{\mathcal{#1}} 
  \def\-#1{\mathrm{#1}} 
  \def\^#1{\mathbb{#1}} 
  \def\!#1{\mathfrak{#1}} 
}{}

\DeclareMathOperator*{\argmax}{arg\,max}

\def\oPr{\mathbf{Pr}}
\renewcommand{\Pr}[2][]{ \ifthenelse{\isempty{#1}}
  {\oPr\left[#2\right]}
  {\oPr_{#1}\left[#2\right]} } 

\DeclareMathOperator*{\oVar}{\mathbf{Var}}
\newcommand{\Var}[2][]{ \ifthenelse{\isempty{#1}}
  {\oVar\left[#2\right]}
  {\oVar_{#1}\left[#2\right]} }

\def\oEnt{\mathbf{Ent}}
\newcommand{\Ent}[2][]{ \ifthenelse{\isempty{#1}}
  {\oEnt\left[#2\right]}
  {\oEnt_{#1}\left[#2\right]} }

\newcommand{\e}{\mathrm{e}}
\renewcommand{\epsilon}{\varepsilon}
\renewcommand{\emptyset}{\varnothing}
\newcommand{\norm}[1]{\left\Vert#1\right\Vert}
\newcommand{\set}[1]{\left\{#1\right\}}
 \newcommand{\eps}{\varepsilon}

\newcommand{\abs}[1]{\left\vert#1\right\vert}
\newcommand{\ctp}[1]{\left\lceil#1\right\rceil}

\newcommand{\junk}[1]{}

\newcommand{\lap}{\texttt{Lap}}

\newcommand{\Paren}[1]{\left(#1\right)}

\newcommand{\cX}{\mathcal{X}}

\newcommand{\cW}{\mathcal{W}}

\newcommand{\cut}{\operatorname{cut}}
\newcommand{\diag}{\operatorname{diag}}

\newcommand{\tr}{\operatorname{tr}}
\newcommand{\E}{\mathbb{E}}
\newcommand{\normt}[1]{\norm{#1}_2}

\newcommand{\ip}[2]{\left\langle #1,#2\right\rangle}

\let\epsilon=\varepsilon
\newcommand{\del}{\delta}

\newcommand{\one}{\mathbbm{1}}
\newcommand{\opnorm}[1]{\left\lVert #1\right\rVert_2}
\newcommand{\fro}[1]{\left\lVert #1\right\rVert_F}

\newcommand{\codeg}{\operatorname{codeg}}

\newcommand{\poly}{\operatorname{poly}}
\newcommand{\Prb}{\mathbf{Pr}}
\newcommand{\R}{\mathbb{R}}
\newcommand{\1}{\mathbf{1}}
\newcommand{\symdiff}{\triangle}
\newcommand{\F}{\mathcal{F}}

\title{Private Approximation of Graph Spectra and Cuts via Spectral Amplifiers}
\date{}
\author{Chenglin Fan\footnotemark[1] \and Jingcheng Liu\footnotemark[2] \and Pan Peng\footnotemark[3] \and Hangyu Xu\footnotemark[3] \and Zongrui Zou\footnotemark[2]}

\footnotetext[1]{Seoul National University. Email: \texttt{fanchenglin@gmail.com}}
\footnotetext[2]{Nanjing University. Emails: \texttt{liu@nju.edu.cn, zou.zongrui@smail.nju.edu.cn}}
\footnotetext[3]{University of Science and Technology of China. Emails: \texttt{ppeng@ustc.edu.cn, hangyuxu@mail.ustc.edu.cn}}

\begin{document}
\pagenumbering{arabic}
\maketitle

\begin{abstract}
We study the problem of releasing a synthetic graph that approximates the sizes of all cuts of an input graph under edge-level differential privacy.  If one insists on purely additive error, the optimal worst-case error is \(\widetilde\Theta(n^{3/2})\).  If one allows a small multiplicative slack, an information-theoretic exponential-time mechanism achieves nearly linear additive error, but the best known polynomial-time algorithms have substantially larger error. We give a polynomial-time \((\varepsilon,\delta)\)-differentially private algorithm which, for every \(n\)-vertex unweighted graph \(G\), outputs a non-negative weighted synthetic graph \(\widetilde G\) such that, with high probability, every cut \(S\subseteq V(G)\) satisfies
\[
  |w_G(S)-w_{\widetilde G}(S)|
  \le
  \gamma w_G(S)+\widetilde O_{\varepsilon,\delta,\gamma}(n^{13/12+o(1)}).
\]
This improves the previous polynomial-time worst-case bound \(\widetilde O(n^{5/4+o(1)})\) of Aamand et al. (ICML 2025) for mixed multiplicative/additive private cut approximation.

The main technical ingredient is a new set of private spectral primitives for bounded-degree graphs. Our power-based spectral amplifier releases high powers of the adjacency matrix and uses the noisy amplified matrix to identify the large-eigenvalue subspace.  The square amplifier gives spectral error \(\widetilde O_{\delta}((nd)^{1/4}/\sqrt\varepsilon)\) in estimating the graph Laplacian for graphs of maximum degree \(d\), being the \emph{first} to beat the standard \(\min\{2d,\widetilde O_{\delta}(\sqrt{n}/\varepsilon)\}\) baseline in the high-degree regime. We further develop a bootstrapped fourth-power primitive with a sharper error dependence on $n$ and $d$ for the downstream cut approximation. Combined with a stable port gadget, a demand-aware recursive expander decomposition, and a new edge-sensitive terminal cut oracle with additive error \(\widetilde O(n+(n^2M)^{1/3})\) on graphs with \(M\) edges, this yields the final worst-case \(\widetilde O(n^{13/12+o(1)})\) private cut-release error. We also prove an \(\Omega(\sqrt n)\) lower bound on the error of private Laplacian release, which holds under approximate differential privacy even for graphs of  average degree $O(1)$. This shows that the worst-case \(\sqrt n\)-scale spectral barrier is information-theoretic, even for sparse graphs. 

\end{abstract}

 \thispagestyle{empty}

\newpage
\tableofcontents
\thispagestyle{empty}

\newpage
\pagenumbering{arabic}
\section{Introduction}

Graph data provides a powerful and intuitive framework for modeling complex interactions, making it widely used in areas like machine learning, network analysis, social sciences, and biology. As a graph may contain sensitive information such as contact networks, financial transactions and communication logs, ensuring the privacy of these underlying relationships is a fundamental challenge. To mitigate the risk of exposing sensitive data, Differential Privacy (DP) \cite{dwork2006calibrating} has been established as the gold standard for privacy-preserving data analysis. A randomized algorithm $\mathcal{A}$ is $(\epsilon, \delta)$-differentially private if, for every pair of \emph{``neighboring''} datasets $D$ and $D'$ differing by at most one individual record, and for every set of possible outputs $\mathcal{S}$, it holds that:
\[
    \Prb[\mathcal{A}(D) \in \mathcal{S}] \le e^\epsilon \Prb[\mathcal{A}(D') \in \mathcal{S}] + \delta.
\]
Informally, this ensures that the behavior of the algorithm remains indistinguishable (parametrized by $(\eps,\del)$) whether any single individual opts into or out of the dataset, so an adversary observing the final algorithmic output cannot infer the presence or exact data of any specific individual. Differential privacy has established successful applications in \emph{private graph analysis} including cut approximation \cite{gupta2012iterative, eliavs2020differentially, dalirrooyfard2023nearly,liu2024optimal, aamandbreaking, chandra2026differentially}, spectral approximation \cite{blocki2012johnson, arora2019differentially, upadhyay2021differentially, zou2025almost}, correlation or hierarchical clustering (\cite{bun2021differentially, DBLP:conf/nips/Cohen-AddadFLMN22, DBLP:conf/icml/ImolaEMCM23, cohen2022scalable, deng2025price, zoudifferentially}) and numerical statistics release (\cite{DBLP:conf/tcc/KasiviswanathanNRS13, DBLP:conf/nips/UllmanS19, DBLP:conf/focs/BorgsCSZ18, ding2021differentially, imola2022differentially, raskhodnikova2026local}), among others.

A central problem in private graph analysis is to release a synthetic graph preserving the \emph{values} of all cuts.  For an unweighted graph \(G=([n],E)\), write \(w_G(S)\) for the number of edges (or the sum of edge weights in weighted graphs) crossing \((S,[n]\setminus S)\).  We seek to release a (probably weighted) graph \(\widetilde G\) under \emph{edge-level privacy} (firstly proposed in \cite{dwork2006calibrating}) such that, simultaneously for all \(S\subseteq[n]\),
\begin{equation}\label{eq:cut-approximation}
    (1-\gamma)w_G(S)-\alpha
    \le
    w_{\widetilde G}(S)
    \le
    (1+\gamma)w_G(S)+\alpha .
\end{equation}
Here \(\gamma\in[0,1)\) is a relative-error parameter and \(\alpha\) is the additive error. In edge-level privacy, two graphs $G,G'$ are neighboring if they differ by adding or removing at most one edge.

The purely additive case \(\gamma=0\) is now well understood. It is easy to verify that adding independent Gaussian noise with variance ${O}(\log^2(1/\del)/\eps^2)$ on the $0/1$ weight between each pair of vertices is $(\eps,\del)$-edge-level differentially private, and it gives an error \(\alpha = \widetilde O(n^{3/2}/\varepsilon)\). More refined algorithms of Gupta et al.~\cite{gupta2012iterative} and Eli\'{a}\v{s} et al.~\cite{eliavs2020differentially} achieve \(\widetilde O(\sqrt{mn/\varepsilon})\) on \(m\)-edge graphs, and this dependence is essentially tight for unweighted graphs due to lower bounds in \cite{eliavs2020differentially}. Liu et al.~\cite{liu2024optimal} further sharpened the weighted setting.  Thus, in the dense worst case \(m=\Theta(n^2)\), purely additive private cut release has an unavoidable \(\widetilde\Theta(n^{3/2})\)-scale error.

This error scale is too large for many dense graphs.  A dense graph may contain sparse but structurally important cuts, and an additive error of order \(n^{3/2}\) can completely hide them.  This motivates the mixed guarantee (\cref{eq:cut-approximation}) for $\gamma\in (0,1)$: large cuts are approximated relatively, while small cuts are protected by an additive floor.  In this mixed model, Blocki et al.~\cite{blocki2012johnson} gave an information-theoretic, exponential-time mechanism with additive error \(\widetilde O(n)\), using the existence of cut sparsifiers and the exponential mechanism over all sparse graphs. The challenge is algorithmic.  The exponential-time mechanism searches over a huge family of sparse synthetic graphs, and it has not been clear how to reproduce the same behavior in polynomial time.

Aamand et al.~\cite{aamandbreaking} recently made the first polynomial-time progress beyond the purely additive $\widetilde O(n^{3/2})$ barrier in the mixed model, achieving an additive error of $\widetilde O_{\varepsilon,\delta,\gamma}(n^{5/4+o(1)})$. Their approach is based on private expander decompositions: inside an expander piece, the graph is highly connected, which means that any cut contains enough edges to absorb the additive noise of a local private release as a small relative error (i.e. $\gamma w_G(S)$), leaving a \emph{residual} graph of unprocessed inter-component edges to be released separately.

Their work identifies the right high-level strategy, but it relies on a one-time decomposition and leaves two bottlenecks. First, the best-known (before this paper) private spectral primitive used to estimate cuts inside the expanders has an $\widetilde O(\sqrt{n})$ spectral error. This error is oblivious to the graph's density, and forces the algorithm to use a high expansion threshold that consequently leaves a massive residual graph. Second, to process this leftover graph, their algorithm falls back to a standard purely additive release. Specifically, it uses the synthetic graph oracle by Eli\'{a}\v{s} et al.~\cite{eliavs2020differentially}, which inherently incurs an $\widetilde O(\sqrt{nM})$ error on an $M$-edge residual.

This paper addresses both bottlenecks. To overcome the first, we have to develop new private spectral primitives whose errors explicitly depend on the graph's maximum degree $d$ (e.g., achieving an $\widetilde O((nd)^{1/4})$ spectral error). Because these primitives become increasingly accurate as the graph becomes sparser, they make it possible to apply expander decompositions \emph{recursively} on the residual, which is impossible with previous $\widetilde O(\sqrt{n})$-error private spectral estimators. To address the second bottleneck, we introduce a new \emph{terminal} oracle for sparse residual graphs. Here \emph{terminal} means ``at the end of the recursion'': it is a final-stage subroutine invoked only after the recursive decompositions have made the residual sparse enough, releasing this leftover graph directly rather than decomposing it further. By our new degree-sensitive primitives, the recursion drives the residual graph to $n^{5/4+o(1)}$ edges. Our terminal oracle privatizes an $M$-edge residual with an additive error of $\widetilde{O}(n+(n^2M)^{1/3})$ and a small multiplicative relaxation. Combining these two ingredients yields a final additive error of $\widetilde O(n^{13/12+o(1)})$, improving over the previous $\widetilde O(n^{5/4+o(1)})$ polynomial-time bound and moving closer to the $\widetilde O(n)$ information-theoretic benchmark.

\subsection{Our Results}

\paragraph{Private cut approximation.}
Our main theorem is the following.

\begin{theorem}[Private Cut Approximation, restatement of Corollary~\ref{cor:cut_13_12}, Section~\ref{sec:cut}]
\label{thm:result-main-cut}
Let \(G=(V,E)\) be an \(n\)-vertex unweighted graph. For every \(\varepsilon\in(0,1)\), $\del\in (0,1/2)$ and \(\gamma\in(0,1/4)\), there is a polynomial-time \((\varepsilon,\delta)\)-edge-DP algorithm that outputs a non-negative weighted synthetic graph \(\widetilde G\) such that, with high probability, simultaneously for all cuts \(C\subseteq V\),
\[
    |w_G(C)-w_{\widetilde G}(C)|
    \le
    \gamma w_G(C)
    +
    \widetilde O_{\delta}\left(
      \frac{n^{13/12+o(1)}}{\varepsilon\gamma^{7/6}}
    \right).
\]
\end{theorem}

The theorem should be compared with three benchmarks.  The first is the optimal purely additive bound \(\widetilde\Theta(n^{3/2})\) in the dense worst case.  The second is the exponential-time \(\widetilde O(n)\) mixed-error mechanism of Blocki et al.~\cite{blocki2012johnson}.  The third is a recent polynomial-time mixed-error bound \(\widetilde O(n^{5/4+o(1)})\) of Aamand et al.~\cite{aamandbreaking}.  Our exponent \(13/12\) lies strictly between the efficient \(5/4\) barrier and the nearly linear information-theoretic target.  More importantly, the proof introduces a modular way to combine degree-sensitive private spectral release with an edge-sensitive terminal oracle; this modular interface may be useful for further improvements.

\paragraph{Private spectral primitives.}
The first main technical component is a family of private Laplacian estimators.  For a graph \(G\) with Laplacian \(L_G :=D_G-A_G\), we measure error in spectral norm $\|L_G - L_{\widetilde{G}}\|_2$. This is stronger than cut approximation in the sense that
\begin{equation}\label{eq:laplacian-implies-cut}
        |w_G(S)-w_{\widetilde G}(S)|
    =
    |\mathbf 1_S^\top(L_G-L_{\widetilde G})\mathbf 1_S|
    \le
    \|L_G-L_{\widetilde G}\|_2\cdot |S| .
\end{equation}
Spectral approximation is also useful beyond cuts, since the Laplacian controls random walks, effective resistances, and many spectral graph quantities~\cite{lovasz1993random, spielman2011graph}.

For an \(n\)-vertex graph of maximum degree \(d\), two simple baselines are immediate.  The zero estimator has error at most \(2d\) as $\|L_G\|_2 \leq 2d$ if $G$ is unweighted, and the dense Gaussian mechanism has error \(\widetilde O(\sqrt n/\varepsilon)\) (see e.g. \cite{blocki2012johnson}).  Previous private spectral-approximation results, including the topology-sampler approach of Liu, Upadhyay and Zou~\cite{liu2024optimal}, essentially achieve a \(\widetilde O(d)\)-type worst-case dependence on the maximum degree for \textit{weighted} graphs, together with a \(\widetilde O(\sqrt n)\)-type fallback. Both bounds are insufficient for our recursive cut algorithm: an additive error proportional to $d$ is trivial for unweighted graphs, and the $\widetilde O(\sqrt{n})$ error from adding Gaussian noise is oblivious to graph sparsity and therefore does not yield better accuracy as the residual becomes sparser, preventing us from applying recursion to the residual graph. We therefore need a primitive that is genuinely sublinear in \(d\) in the regimes encountered by the recursion.

Our new spectral primitives are based on a new technique, the \emph{power-based spectral amplifier} (Theorem~\ref{thm:result_generic_framework}, Section~\ref{sec:power_amplifier}), which we develop as a generic mechanism for the private release of symmetric matrices. Specialized to graph Laplacian matrices, the square amplifier already gives a clean improvement upon the longstanding baseline $\min\{O(d), \widetilde{O}(\sqrt{n}/\eps)\}$, and the fourth-power primitive is designed for the recursive cut framework.

\begin{theorem}[Private Laplacian Release, restatement of Theorems~\ref{thm:laplacian-square} and~\ref{thm:improved-spectral}, Section~\ref{sec:laplacain_analysis}]
\label{thm:intro_laplacian}
Let \(G\) be an \(n\)-vertex graph with maximum degree at most \(d\).  For every tunable parameter \(q\ge \operatorname{polylog}(n,1/\delta)/\varepsilon\), there is a polynomial-time \((\varepsilon,\delta)\)-edge-DP algorithm that outputs a non-negative weighted graph \(\widetilde G\) whose Laplacian satisfies, with high probability,
\begin{equation}\label{eq:intro-a4-amplifier}
    \|L_G-L_{\widetilde G}\|_2
    \le
    \widetilde O\left(\frac1\varepsilon\right)
    +
    \min\left\{
      \widetilde O\left(\frac{(nd)^{1/4}}{\sqrt\varepsilon}\right),
      \widetilde O\left(
        \frac{d}{\sqrt{\varepsilon q}}
        +
        \frac{n^{1/8}d^{1/2}q^{1/8}}{\varepsilon^{3/8}}
        +
        \frac{n^{1/4}q^{1/4}}{\varepsilon^{3/4}}
      \right),
      2d
    \right\}.
\end{equation}
\end{theorem}

The term \(d/\sqrt q\) in \eqref{eq:intro-a4-amplifier} is important.  It should not be confused with the known \(\widetilde O(d)\) primitive.  In the recursive cut framework, a local spectral error \(\Lambda(d)\) leads to a residual-degree recurrence of the form \(d_{t+1}\approx \Lambda(d_t)\), up to privacy, \(\gamma\), and polylogarithmic factors.  A \(\widetilde O(d)\) error therefore does not reduce the residual degree.  By contrast, \(d/\sqrt q\), with \(q\) a sufficiently large polylogarithm, is a tunably contracting linear term.  The polynomial fixed point is then governed by the sublinear terms in \eqref{eq:intro-a4-amplifier}.

We also show that the maximum degree $d$ cannot be replaced by average degree, as the worst-case private spectral release has an unavoidable \(\sqrt n\) lower bound even if average degree is constant.

\begin{theorem}[Spectral Lower Bound in Sparse Graphs, restatement of Theorem~\ref{thm:spectral-lower-bound}, Section~\ref{sec:lowerbound}]
\label{thm:result-lower-bound}
Fix a constant \(0<\varepsilon<1\) and let \(\delta=n^{-c}\).  Any \((\varepsilon,\delta)\)-edge-DP algorithm that releases a synthetic graph Laplacian $\widehat{L}$ must incur worst-case expected spectral error
\[
    \mathbb E\bigl[\|L_G-\widehat L\|_2\bigr]
    =\Omega(\sqrt n).
\]
The hard instances are connected graphs with maximum degree \(\Theta(n)\) and average degree $O(1)$.
\end{theorem}

It should be noted that an \(\Omega(\sqrt{n})\) worst-case spectral lower bound in \textbf{dense} graphs already implicitly follows from the \(\Omega(n^{3/2})\) lower bound for cut error by Eli{\'a}{\v{s}} et al.~\cite{eliavs2020differentially} (since a spectral error of \(o(\sqrt n)\) would immediately yield a cut error of \(o(n^{3/2})\) by \cref{eq:laplacian-implies-cut}). However, their hard instances require \(d_{\text{avg}}=\Theta(n)\) so {cannot} rule out possibility of a spectral primitive whose error scales with the {average} degree. Our lower bound is based on connected hidden-star instances. Spectral error \(o(\sqrt n)\) would reconstruct most private edges incident to a high-degree center; approximate differential privacy forbids such reconstruction.  This shows that the role of our degree-sensitive spectral primitives is not to beat this barrier on every graph, but to exploit the decreasing residual degree inside the cut algorithm.

\paragraph{Edge-sensitive terminal cut release.}
The second ingredient is a terminal oracle for sparse residual graphs.  We use the phrase \emph{terminal cut release} for this final direct release step: after the recursive expander routine has reduced the hard leftover edges to a graph with only \(M\) edges, we stop recursing and privately release that leftover graph's cut function in one shot.  This terminal step is different from the local spectral releases used earlier in the recursion; it is designed to exploit sparsity of the residual edge set rather than expansion of individual pieces.
 

\begin{theorem}[Edge-Sensitive Cut Oracle, restatement of Theorem~\ref{thm:cut-oracle-main}, Section~\ref{sec:edge-sensitive-cut-oracle}]
\label{thm:intro-edge-sensitive-cut-oracle}
Let \(G=(V,E)\) be an \(n\)-vertex unweighted graph, and let \(M\ge |E|\) be a public upper bound.  For every \(\varepsilon,
\delta\in(0,1)\) and \(\gamma\in(0,1/4)\), there is a polynomial-time \((\varepsilon,\delta)\)-edge-DP algorithm that outputs a non-negative weighted graph \(\widetilde G\) such that, with high probability, simultaneously for all cuts \(C\subseteq V\),
\[
    |w_G(C)-w_{\widetilde G}(C)|
    \le
    \gamma w_G(C)
    +
    \widetilde O\left(
       \frac n\varepsilon
       +
       \left(\frac{n^2M}{\varepsilon^2\gamma}\right)^{1/3}
    \right).
\]
\end{theorem}

The above theorem should be compared with the purely additive error bound \(\widetilde O(\sqrt{nM/\eps})\) in \cite{eliavs2020differentially}.  When \(M\) is moderately sparse and a relative slack is allowed, our terminal oracle is better.  The key new feature is that its error depends on the number of residual edges through the cubic expression \((n^2M)^{1/3}\), rather than through \(\sqrt{nM}\).  This improved edge sensitivity is the reason reducing the residual to \(n^{5/4+o(1)}\) edges suffices to obtain the final exponent \(13/12\).

\paragraph{A modular composition theorem.}To compose preceding primitives, we must bridge the arbitrary degrees of intermediate residuals with the bounded-degree requirements of our spectral oracles. We achieve this via two combinatorial tools: a \emph{stable port gadget} (Section~\ref{sec:stable_port_gadget}) privately flattens high-degree vertices to match the graph's average degree. It routes incident edges to public ``ports'' via randomized stable truncation, preserving cuts with $O(1)$ sensitivity. The \emph{demand-aware expander decomposition} (Section~\ref{sec:expander-decomposition}) is the matching decomposition step: it measures the size of a side by its public port demand, not by the raw number of original vertices.  This ensures that expansion, spectral error, and residual charging are all measured in the same units.

\begin{theorem}[Modular spectral-to-cut amplification, informal]
\label{thm:intro-modular-composition}
Suppose there is a polynomial-time \((\varepsilon,\delta)\)-DP spectral primitive which, on every \(N\)-vertex graph of maximum degree at most \(d\), releases a Laplacian with error at most \(\Pi(N,d)\).  Suppose further that there is a terminal cut oracle with mixed error \(\alpha_{\rm term}(n,M)\) on residual graphs with at most \(M\) edges.  Then the recursive port-and-expander framework gives a polynomial-time private cut release algorithm whose residual edge counts satisfy
\[
    M_{t+1}
    \le
    \widetilde O_{\varepsilon,\delta,\gamma}
    \left(
      n\,\Pi\left(\widetilde O(n), M_t/n\right)
    \right),
\]
up to polylogarithmic and \(n^{o(1)}\) factors.  If this recurrence reaches \(M_T\), then the final cut error is
\[
    \gamma w_G(S)+\alpha_{\rm term}(n,M_T)
\]
simultaneously for all cuts \(S\subseteq V(G)\).
\end{theorem}

In our instantiation, \(\Pi\) is the bootstrapped fourth-power primitive (the second term in \cref{eq:intro-a4-amplifier}, \Cref{thm:intro_laplacian}) and \(\alpha_{\rm term}(n,M)=\widetilde O(n+(n^2M)^{1/3})\) due to Theorem~\ref{thm:intro-edge-sensitive-cut-oracle}. The recurrence reaches \(M_T=n^{5/4+o(1)}\), and the terminal oracle then gives \(n^{13/12+o(1)}\).  This theorem is included to emphasize that the proof is \emph{not} a black-box substitution into the one-shot expander-decomposition algorithm in Aamand et al.~\cite{aamandbreaking}.  The spectral primitive, the stable port gadget, the recursive demand-aware decomposition, and the terminal oracle must be matched so that each edge is charged either multiplicatively inside an expander or additively only at the sparse terminal stage.

\section{Technical Overview}\label{sec:tech_overview}

We now give a more detailed overview of the proof.  The algorithm has three layers.  The first layer is spectral: we build private Laplacian estimators whose accuracy improves as the maximum degree decreases.  The second layer is combinatorial: we use a stable port gadget and a demand-aware expander decomposition so that a residual with \(m\) edges can be treated as a graph of degree scale \(m/n\).  The third layer is terminal: once the residual is sparse enough, the edge-sensitive cut oracle finishes the release.

The layers are complementary. A better spectral primitive alone only yields an additive cut error proportional to the size of the queried side $|S|$ (as illustrated in \cref{eq:laplacian-implies-cut}). An expander decomposition alone reduces the problem size only if the local spectral oracle is accurate enough relative to the expansion threshold. A terminal oracle alone helps only after the graph is already sparse. The proof is therefore a unified pipeline: the spectral primitive triggers the private expander decomposition to shrink the residual recursively, and the terminal oracle converts the final residual size into the global all-cuts error bound.


\subsection{Private Laplacian Analysis via a Spectral Amplifier}
\label{sec:tech_overview-laplacian-oracle}

Given a simple graph $G = ([n],E)$ with bounded degree $d$ and adjacency matrix $A\in \mathbb{R}^{n\times n}$, our first step is to privately release a synthetic matrix $\widetilde{A}$ to minimize $\|\widetilde{A} - A\|_2$, and then combine $\widetilde{A}$ with privately released degrees to form an estimator of Laplacian $L_G$. A direct application of the Gaussian mechanism to \(A\) adds  independent noise $\mathcal{N}(0, \sigma^2)$ with $\sigma = \Theta(\log(1/\del)/\eps)$ to each entry of $A$ as $\|A - A'\|_F\leq O(1)$ under edge-level privacy. Therefore, the resulting noise matrix has spectral norm \(\widetilde O(\sqrt n/\eps)\) due to Lemma~\ref{lem:gaussian_norm}. This gives the standard oblivious baseline, but it does not exploit the bounded-degree structure of the graph.

\paragraph{The square amplifier.} 
Our first spectral amplifier improves on this baseline by releasing a noisy version of \(A^2\) rather than \(A\) itself. The entries of \(A^2\) count length-two walks, or equivalently common neighbors. In a \(d\)-bounded graph, changing one edge only changes a controlled local collection of such two-walk counts. Formally, this ensures that
$
    \|A^2-(A')^2\|_F=O(\sqrt d)
$ 
for neighboring adjacency matrices \(A,A'\). Hence the noise matrix of a Gaussian release $
    \widetilde Y=A^2+W
$ 
has spectral norm
$
    \|W\|_2=\widetilde O(\sqrt{nd}/\eps).
$

As illustrated in Figure~\ref{fig:amplifier_intuition}, the advantage is that the signal is amplified quadratically. If
$A=\sum_i \lambda_i u_i u_i^\top$, 
then $A^2=\sum_i \lambda_i^2 u_i u_i^\top$. 
Thus an eigenvector with \(|\lambda_i|\) moderately large in \(A\) becomes much easier to detect in \(A^2\). In particular, once $\lambda_i^2 \gtrsim \widetilde O(\sqrt{nd}/\eps),$
the direction \(u_i\) can be separated from the noise in the amplified matrix. Equivalently, the amplifier identifies all directions with
\[
    |\lambda_i|
    \gtrsim
    \widetilde O\left((nd)^{1/4}/\sqrt{\eps}\right).
\]
The remaining directions can be safely discarded, as their contribution to the final spectral error $\|A-\widetilde A\|_2$ is exactly their maximum absolute eigenvalue in $A$, which is at most $\widetilde O((nd)^{1/4}/\sqrt{\eps})$.

\begin{figure}[t]
\centering
\begin{tikzpicture}[
    >=Latex,
    font=\sffamily\small,
    eig/.style={circle, fill=blue!60, inner sep=1.8pt, draw=blue!80!black},
    sig/.style={circle, fill=red!60, inner sep=2pt, draw=red!80!black},
    thresh/.style={dashed, thick, red!80},
    naive/.style={dashed, thick, gray!80}
]

\node[align=center, font=\bfseries] at (2.5, 2.35) {Original adjacency matrix \(A\)};

\draw[->, thick] (0,0) -- (5.5,0) node[right] {\(|\lambda_i|\)};

\foreach \x in {0.2, 0.4, 0.6, 0.9, 1.2} {
    \node[eig] at (\x, 0) {};
}

\node[eig] (eigL0) at (1.4, 0) {};
\node[sig] (sigL1) at (2.6, 0) {};
\node[sig] (sigL2) at (3.8, 0) {};

\draw[thresh] (1.8, -0.8) -- (1.8, 1.0)
    node[above, text=red!80, align=center]
    {\footnotesize direct-release noise \\ \(\widetilde O(\sqrt n)\)};

\draw[decorate, decoration={brace, amplitude=5pt, mirror}, thick, blue!80]
    (0.1, -0.2) -- (1.7, -0.2)
    node[midway, below=5pt, align=center]
    {\footnotesize hidden by \\ noise};

\draw[decorate, decoration={brace, amplitude=5pt, mirror}, thick, red!80]
    (2.1, -0.2) -- (5.4, -0.2)
    node[midway, below=5pt, align=center]
    {\footnotesize directions visible \\ after direct release};

\draw[->, ultra thick, blue!50, bend right=15] (2.7, -1.8)
    to node[above=4pt, font=\bfseries, text=black] {\(\Phi_2(X)=X^2\)}
       node[below=4pt, align=center, font=\footnotesize, text=gray!80!black]
       {Spectral Amplification}
    (8.5, -1.8);

\node[align=center, font=\bfseries] at (11, 2.35) {Squared amplifier \(A^2\)};

\draw[->, thick] (7,0) -- (15.5,0) node[right] {\(\lambda_i^2\)};

\foreach \x in {7.02, 7.06, 7.14, 7.32, 7.57, 7.78} {
    \node[eig] at (\x, 0) {};
}

\node[sig] (sigR0) at (10.0, 0) {};
\node[sig] (sigR1) at (12.5, 0) {};
\node[sig] (sigR2) at (14.0, 0) {};

\draw[thresh] (8.6, -0.8) -- (8.6, 1.0)
    node[above, text=red!80, align=center]
    {\footnotesize amplified noise \\ \(\widetilde O(\sqrt{nd})\)};

\draw[naive] (11.5, -0.8) -- (11.5, 1.0)
    node[above, text=gray!80, align=center]
    {\footnotesize dense-graph scale \\ \(\widetilde O(n)\)};

\draw[decorate, decoration={brace, amplitude=5pt, mirror}, thick, blue!80]
    (7.0, -0.2) -- (8.5, -0.2)
    node[midway, below=5pt, align=center]
    {\footnotesize discarded \\ directions};

\draw[decorate, decoration={brace, amplitude=5pt, mirror}, thick, red!80]
    (8.7, -0.2) -- (13.5, -0.2)
    node[midway, below=5pt, align=center]
    {\footnotesize directions recovered from \(A^2\)};

\draw[->, dashed, red!40, bend left=15] (sigL1) to (sigR1);
\draw[->, dashed, red!40, bend left=15] (sigL2) to (sigR2);
\draw[->, dashed, red!40, bend left=15] (eigL0) to (sigR0);

\end{tikzpicture}
\caption{
Spectral amplification by the square map. Directly releasing \(A\) incurs spectral noise \(\widetilde O(\sqrt n)\). Releasing \(A^2\) instead amplifies eigenvalues from \(\lambda_i\) to \(\lambda_i^2\), while bounded degree keeps the sensitivity of \(A^2\) at \(O(\sqrt d)\), giving amplified spectral noise \(\widetilde O(\sqrt{nd})\). Thus directions with \(\lambda_i^2\gtrsim \widetilde O(\sqrt{nd})\), equivalently \(|\lambda_i|\gtrsim \widetilde O((nd)^{1/4})\), can be recovered. This explains the \(\widetilde O((nd)^{1/4})\) adjacency spectral error.
}
\label{fig:amplifier_intuition}
\end{figure}
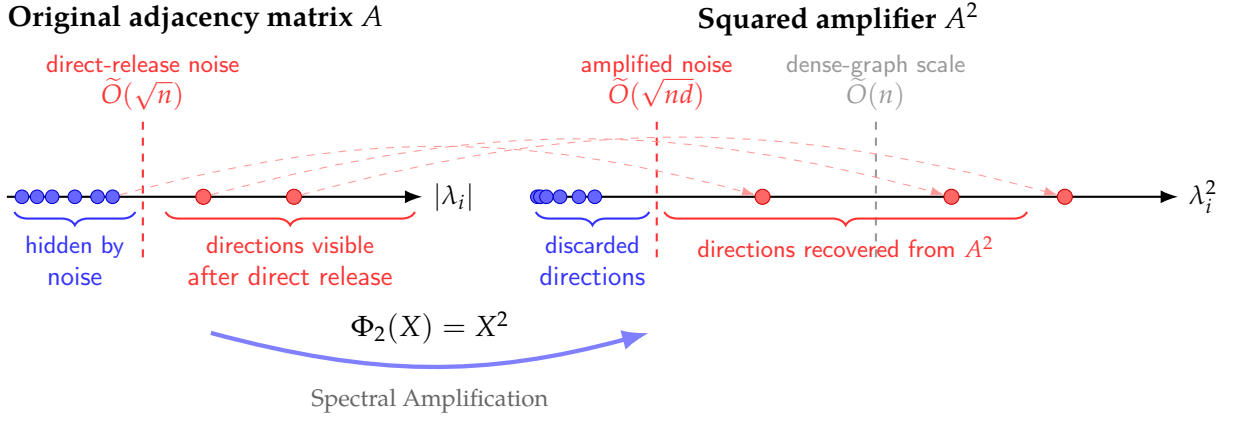


Driven by this intuition, the adjacency release proceeds in three steps: 
\begin{itemize}
    \item First, we release an amplified matrix $\widetilde{Y} = A^2 + W$, where $W$ is Gaussian noise calibrated to the $O(\sqrt{d})$ sensitivity of $A^2$.
    \item Second, we truncate all eigenvectors of $\widetilde{Y}$ with eigenvalues falling below the noise scale $\eta = \widetilde{O}(\sqrt{nd}/\epsilon)$, leaving an orthonormal basis $U \in \mathbb{R}^{n \times r}$ for the remaining $r$ eigenvectors. 
    \item Third, privately release only the compressed matrix \(U^\top A U + E\) where $E \in \mathbb{R}^{r \times r}$ is the Gaussian noise calibrated by $\|U^\top AU - U^\top A'U\|_F= O(1)$.
\end{itemize}

The total spectral error balances the \emph{truncation bias} and the \emph{subspace noise}. For the truncation bias, eigenvalues of $A^2$ in the discarded subspace are bounded by $\eta$, so their spectral contribution to $A$ is bounded by $\eta^{1/2} = \widetilde{O}((nd)^{1/4}/\sqrt{\epsilon})$. For the subspace noise, the diagonal of $A^2$ contains vertex degrees, giving $\operatorname{tr}(A^2) \le nd$. Since each of the $r$ surviving eigenvectors has an eigenvalue of at least $\eta$, we establish that $r \cdot \eta \le \operatorname{tr}(A^2) \le nd$, yielding $r \le \widetilde{O}(\sqrt{nd})$. Consequently, the spectral norm of the noise matrix $E$ is bounded by $\widetilde{O}(\sqrt{r}) = \widetilde{O}((nd)^{1/4})$, matching the truncation error.


Finally, combining $\widetilde{A}$ with a privately released degree matrix $\widetilde D$ yields a pseudo-Laplacian $\widetilde L=\widetilde D-\widetilde A$. To guarantee non-negative edge weights for downstream cut algorithms, we apply an SDP post-processing step from \cite{upadhyay2021differentially}: solving $L_{\widetilde{G}} = \arg\min_{X \in \mathbb{L}_n} \|\widetilde{L} - X\|_2$ to project $\widetilde{L}$ onto $$ \mathbb{L}_n = \left\{ X \in \mathbb{R}^{n \times n} \mid X = X^\top, X_{ij} \le 0 \text{ for } i \ne j, X \mathbf{1} = \mathbf{0} \right\},$$ the convex cone of all valid non-negative weighted graph Laplacians.


This square amplifier is useful as a standalone primitive. It improves over the naive \(\widetilde O(\sqrt n/\eps)\) adjacency release whenever \(d=o(n)\), and it improves over the degree-only \(\widetilde O(d)\) bound when \(d\gg n^{1/3}\). Further, it is also strong for the cut-approximation in our recursive framework. If the spectral primitive has error scale $n^{1/4}d^{1/4}$, 
then the recurrence in the cut framework drives the edge count of residual graph to only $M=n^{4/3}$. Specifically, applying our new terminal oracle (Theorem~\ref{thm:intro-edge-sensitive-cut-oracle}) to this residual already yields an $\widetilde O(n^{10/9})$ additive cut error, which breaks the previous $\widetilde O(n^{5/4})$ polynomial-time barrier. However, to achieve our final $\widetilde O(n^{13/12})$ error, we must force the residual graph even sparser. Next we show how to develop a more advanced fourth-power primitive that achieves a sharper error contraction.


\paragraph{Why fourth powers are delicate.}
Motivated by the $A^2$ amplifier, it is natural to ask whether we can evaluate higher powers, as evaluating $A^q$ for larger even $q$ amplifies the spectral gap more aggressively. Unfortunately, the delicate trade-off between truncation error and subspace noise in $A^2$ breaks down for higher powers, such as $A^4$.
The obstacle is local two-walk congestion. The sensitivity of \(A^4\) is governed by \textit{two-walk energy} defined as 
\[
    \kappa_2(G)
    :=
    \max_{v\in V}\|A^2\mathbf{e}_v\|_2^2
    =
    \max_{v\in V}\sum_{u\in V}\codeg(u,v)^2 .
\]
This quantity can be as large as \(d^3\), for example in a \(d\)-clique or in a complete bipartite graph \(K_{d,d}\). Thus the hard examples are not only triangle-rich graphs; biclique-like regions with many repeated two-walks are equally problematic. A naive \(A^4\)-amplifier calibrated to the worst-case $d^3$ sensitivity yields a spectral error of $ \widetilde O\left(n^{1/8}d^{5/8} + (nd)^{1/4}\right) $, therefore does not improve over the \(A^2\) primitive. For further improvements, our solution is to first identify, in a private aggregate way, the edges responsible for large two-walk energy $\kappa_2$, and only then apply the \(A^4\) amplifier to the remaining graph with smaller sensitivity.

\paragraph{Iterated cubic-score cover.}
The key preprocessing step is an \emph{iterated cubic-score cover}. The score
$
    (A^3)_{uv}
$ 
counts length-three walks from \(u\) to \(v\). It is useful because, for any remaining graph \(R\) and any $x\in[n]$:
\begin{equation}\label{eq:local-walk-energy-and-cubic}
    \|A_R^2\mathbf{e}_x\|_2^2
    =
    (A_R^4)_{xx}
    =
    \sum_{y\in N_R(x)}(A_R^3)_{xy},
\end{equation}
where $N_R(x)$ is the set of neighbors of $x$ in $R$. Hence, if every edge \(xy\) in $R$ has small cubic score \((A_R^3)_{xy}\), then every vertex has small two-walk energy. The cover works as follows. Starting with \(R_0=G\), we privately release a noisy cubic matrix by Gaussian mechanism $\widetilde B_i=A_{R_i}^3+Z_i$
for the current residual \(R_i\). We then add to a public pair support \(P_i\) all unordered pairs \(uv\) whose noisy cubic score is above the current threshold, and we remove from \(R_i\) the actual residual edges whose unordered pairs lie in \(P_i\). In other words,
\[
    R_{i+1}
    :=
    R_i\setminus \bigl(E(R_i)\cap P_i\bigr),
    \qquad
    P:=\bigcup_i P_i .
\]
The support \(P_i\) contains pairs rather than only true edges; this aggregate support is important for privacy, because it avoids revealing the exact edge set.

The reason for iterating is that the sensitivity decreases as the residual becomes less congested. Initially every \(d\)-bounded graph satisfies
$
    \kappa_2(R_0)\le d^3.
$ 
More generally, we show that if the current residual satisfies
$
    \kappa_2(R_i)\le K_i,
$
then the Frobenius sensitivity of the cubic query \(A_{R_i}^3\) is only
$
    \widetilde O(\sqrt{K_i}).
$
Thus the next noisy cubic-score cover can use threshold \(\widetilde O(\sqrt{K_i})\). Since every remaining residual edge has a cubic score below this threshold, and any vertex has at most $d$ neighbors, the new local two-walk energy (which is the sum of cubic scores to adjacent neighbors, as shown in \cref{eq:local-walk-energy-and-cubic}) is bounded by $d$ times the threshold. Thus, the next residual satisfies $K_{i+1} = \widetilde O(d\sqrt{K_i}).$
Starting from \(K_0=d^3\), this recurrence rapidly approaches to a fixed point of \(d^2\). We stop when
$K_i\le \widetilde O(d^2q),$
where \(q\) is a polylogarithmic slack parameter. At this point,  
\[
    \kappa_2(R)\le \widetilde O(d^2q),\quad \Delta(P)\le \widetilde O(d^2/q).
\]
where $\Delta(P)$ is the maximum degree in the pair support $P$. We then release the actual edges of \(G\) inside the support \(P\) by adding Gaussian noise only to every existing pair in \(P\), as $P$ is now \textit{public}. Since the support degree is \(\widetilde O(d^2/q)\), Gaussian mechanism gives spectral error
\[
    \widetilde O\left(\sqrt{d^2/q}\right)
    =
    \widetilde O(d/\sqrt q).
\]
This term is linear in \(d\), but it is divided by a tunable polylogarithmic factor. In the recursive cut framework, this makes it a genuinely contracting linear term rather than the non-contracting \(\widetilde O(d)\) term from degree-only spectral release.

\paragraph{Fourth-power release on the residual.}
After the iterated cover, the residual graph \(R\) satisfies
$
    \kappa_2(R)\le \widetilde O(d^2q).
$ 
This low two-walk-energy condition makes the fourth-power query stable:
\[
    \|A_R^4-A_{R'}^4\|_F \leq O\left(d\sqrt{\kappa_2(R)}\right)
    \le
    \widetilde O(d^2\sqrt q)
\]
for neighboring graphs \(R,R'\). Therefore a Gaussian release of \(A_R^4\) has spectral noise
$
    \widetilde O\left(\sqrt n\,d^2\sqrt q\right).
$
Taking fourth roots, the \textbf{truncation bias} is 
$\widetilde O\left(n^{1/8}d^{1/2}q^{1/8}\right).$

Following the identical logic as the analysis of the square amplifier, the \textbf{subspace noise} added into the retained low-dimensional space is 
$\widetilde{O}\left(n^{1/4}q^{1/4}\right).$

Combining the private release of the covered part, the application of $A^4$-amplifier on the residual graph, and a standard private degree estimation yields our improved spectral primitive, which achieves an error of
\[
    \|L_G-\widetilde L\|_2
    \le
    \widetilde O_{\eps,\delta}
    \left(
        \frac d{\sqrt q}
        +
        n^{1/8}d^{1/2}q^{1/8}
        +
        n^{1/4}q^{1/4}
    \right),
\]
which is the second bound stated in \eqref{eq:intro-a4-amplifier}. As before, we conclude with a post-processing SDP projection to ensure the final output is a valid Laplacian of a non-negative weighted graph.

The important feature is the combination of a tunable contracting term \(d/\sqrt q\) with two genuinely sublinear amplifier terms. Taking \(q=\operatorname{polylog} n\), the polynomial fixed point of the residual-degree recurrence is \(d=n^{1/4}\), which is the source of the improved cut approximation bound.

\subsection{Building an Edge-Sensitive Cut Oracle}
\label{sec:tech_overview-cut-oracle}

We next describe the terminal cut oracle used at the end of the recursion. Its role is to release the cut function of a sparse residual graph more accurately than what is possible from a purely additive release (e.g., the mechanisms in Eli{\'a}\v{s} et al.~\cite{eliavs2020differentially} or Gupta et al.~\cite{gupta2012iterative}) alone.

Let \(G=([n],E)\) be the current residual graph, and suppose that we have a \textit{public} upper bound \(M\ge |E|\). A standard differentially private cut-release mechanism gives additive error
$
    \widetilde O(\sqrt{Mn})
$ 
simultaneously for all cuts~\cite{eliavs2020differentially}. This bound is essentially optimal for purely additive cut release. However, our overarching goal is to achieve a mixed multiplicative and additive error. Since the edges processed earlier in the recursive expander decomposition are already approximated with a $\gamma$ relative error, it is sufficient to evaluate the sparse residual graph under the same multiplicative slack. We exploit this relaxation to bypass the purely additive barrier, designing a mixed-error terminal oracle: for every cut \(S\subseteq[n]\), it outputs a synthetic graph \(\widehat G\) satisfying
\[
    |w_{\widehat G}(S)-w_G(S)|
    \le
    \gamma w_G(S)
    +
    \widetilde O_{\eps,\delta,\gamma}
    \left(
        n+\left(n^2M\right)^{1/3}
    \right).
\]

Our terminal oracle is best viewed as a private multiplicative-weights update over edge histograms. We represent the residual graph by a histogram \(h\) over the edge universe \(\binom{[n]}2\). To make the total mass public and stable under edge replacement, we add a dummy coordinate \(\bot\) and set
$
    h(\bot)=M-|E|.
$ 
Thus the padded histogram has total mass \(M\). A cut \(S\subseteq[n]\) corresponds to a query vector
$
    c_S\in\{0,1\}^{\binom{[n]}2\cup\{\bot\}}$, 
with $c_S(\bot)=0$, 
and
$
    \langle h,c_S\rangle=\cut_G(S).
$ 
The goal is to construct a synthetic histogram \(\widehat h\) such that, for all cuts \(S\),
\[
    |\langle \widehat h-h,c_S\rangle|
    \le
    \gamma\langle h,c_S\rangle+\alpha,
\]
where \(\alpha\) is the target error. In multiplicative weights update, we iteratively find a cut that violates the above inequality, computing the gradients and updating the edge weights in $\widehat{h}$ accordingly.

\paragraph{Breaking the \(\sqrt{Mn}\) barrier: Absorbing variance via multiplicative slack.}
Inspired by private iterative methods~\cite{gupta2012iterative,eliavs2020differentially,peng2025differentially}, we replace the computationally inefficient discrete search for violated cuts with a continuous regularized semidefinite programming relaxation. However, the fundamental departure from prior work lies in our potential function analysis. 

It is important to note that the improved $(n^2M)^{1/3}$ bound cannot be achieved by just modifying the objective in prior works, that is, applying the standard private mirror descent analysis from Eli{\'a}{\v{s}} et al.~\cite{eliavs2020differentially} to our relative-error objective fails to beat the $\widetilde O(\sqrt{Mn})$ barrier. Symmetrically, employing our analysis on a purely additive objective is also restricted to the same $\widetilde O(\sqrt{Mn})$ bottleneck. To be more specifically, in standard private mirror descent using KL divergence as the potential\footnote{In this framework, the potential function (specifically the KL divergence) acts as a metric of the "distance" between our current synthetic edge distribution and the true target graph. The algorithm's convergence is proven by showing that every valid update strictly decreases this potential, thereby bounding the total number of required steps. See \cite{bansal2017potential} for a detailed introduction to the potential method.} (Bregman Divergence), a single update step with learning rate \(\eta\) changes the expected potential by roughly \(-\eta \cdot (\text{first-order gain}) + \eta^2 U\), where the second-order term \(\eta^2 U\) acts as a variance penalty proportional to the expected gradient magnitude \(U\). For a purely additive error guarantee (\(\gamma=0\)), one is forced to bound this \(\eta^2 U\) term generically (e.g., via Cauchy-Schwarz), which restricts the step size \(\eta\) to be very small and limits the potential drop per round, this results in the final \(\sqrt{Mn}\) error.

Our core insight is that with relative error $\gamma$, the first-order progress term naturally contains a \(-\eta \gamma U\) component. By carefully tuning the learning rate \(\eta = \Theta(\gamma/B)\) (where \(B\) is a hard clipping threshold of gradients), this \(-\eta \gamma U\) term actually \emph{absorbs} the positive variance penalty \(+\eta^2 U\). As detailed in our proofs of Lemma~\ref{lem:potential progress}, this absorption allows us to extract a much larger drop in the potential function per iteration (scaling with \(\gamma\)) without relying on generic Cauchy-Schwarz bounds. This rapid potential decay reduces the required number of iterations, lowering the privacy composition cost and yielding the final \(\widetilde O((n^2M)^{1/3})\) additive error.

Having established the intuition behind our improved error bound, we now detail the specific algorithmic components that gives our new edge-sensitive cut oracle.

\paragraph{Multiplicative Gaussian mechanism and log-det stability.} To privately generate the gradients for this update step, recall that we formulate the cut separation problem as a continuous SDP relaxation over a covariance matrix \(\Sigma\). Conceptually, \(\Sigma\) is a fractional representation of a violating cut found by SDP. However, the exact optimum of an unregularized SDP can be very unstable under edge changes, making privatization difficult. To overcome this, we show that the log-det regularizer,
$    \lambda\log\det(\Sigma),
$ 
gives the stability needed for privacy. Let \(\Sigma\) and \(\Sigma'\) be the optimal covariance matrices (i.e., fractional cuts) for two neighboring input graphs, we prove that
\[
    \left\|
    \Sigma^{-1/2}(\Sigma'-\Sigma)\Sigma^{-1/2}
    \right\|_F
    \le
    O(1/\lambda).
\]
 We then use a multiplicative Gaussian mechanism, adapted from Eli{\'a}{\v{s}} et al.~\cite{eliavs2020differentially}, to privately sample a Gaussian direction from the optimized covariance. By advanced composition over \(T\) iterations, it is enough to choose
$
    \lambda=\widetilde\Theta_{\del}\left(\frac{\sqrt T}{\eps}\right)
$
to obtain \((\eps,\delta)\)-differential privacy.

\paragraph{Relative multiplicative weights update.}
Once a private fractional violating cut is obtained, we update the synthetic histogram using a relative multiplicative-weights step. 
To bound the maximum gradient \(B\) and satisfy the learning rate requirement \(\eta = \Theta(\gamma/B)\), the SDP feedback is clipped before being used in the update. Clipping controls the tails of the Gaussian sample and ensures that the expected KL potential decreases predictably by the \(\gamma\)-absorbed amount described above. Consequently, the capped run terminates after
$
    T = \widetilde O\left(\frac{M}{\gamma\alpha}\right)
$ 
iterations with high probability.

\paragraph{Balancing the parameters.}
The final error is obtained by balancing three effects. The log-det regularizer contributes an additive bias proportional to
$
    \widetilde O(\lambda n).
$
Privacy requires \(\lambda\) to grow like \(\widetilde\Theta_{\del}(\sqrt T/\eps)\), where the number of adaptive calls is roughly
$
    T=\widetilde O\left(\frac{M}{\gamma\alpha}\right).
$ 
Larger \(\lambda\) improves stability but increases the regularization bias. Solving the resulting balance gives
\[
    \alpha
    =
    \widetilde O_{\del}
    \left(
        \frac n\eps
        +
        \left(\frac{n^2M}{\eps^2\gamma}\right)^{1/3}
    \right).
\]
This proves the edge-sensitive terminal cut oracle stated in
Theorem~\ref{thm:intro-edge-sensitive-cut-oracle}. Conceptually, the terminal oracle is a separate, self-contained module from our recursive spectral reduction for the worst-case cut approximation. The recursion is used to reduce the residual edge bound \(M\), while this terminal oracle converts the smaller \(M\) into a final mixed multiplicative/additive cut guarantee. Any future improvement either to the residual size produced by the recursion or to the terminal oracle itself can be plugged into the same modular framework.

\subsection{Private Cut Approximation with \texorpdfstring{\(\widetilde O(n^{13/12+o(1)})\)}{O(n^{13/12+o(1)})} Additive Error}
\label{sec:tech_overview_cut}

We now explain how the spectral primitive and the edge-sensitive terminal oracle combine to yield our worst-case private cut approximation guarantee. We first recall the bottleneck in the previous \(\widetilde O(n^{5/4})\)-additive-error framework, and then describe how our recursive degree-sensitive approach breaks this barrier.

\subsubsection{The bottleneck behind the \texorpdfstring{\(\widetilde O(n^{5/4})\)}{O(n^{5/4})} bound}

Aamand et al.~\cite{aamandbreaking} obtains a private cut approximation with additive error \(\widetilde O(n^{5/4})\), while allowing a small multiplicative error. Their approach can be viewed as having two phases. First, one can just use a private Laplacian for estimating cut sizes as illustrated in \cref{eq:laplacian-implies-cut}. The Gaussian mechanism adds noise to all ${n\choose 2}$ potential edges and produces a synthetic graph \(\widetilde G\) satisfying
\[
    \|L_G-L_{\widetilde G}\|_2
    \le
    \widetilde O(\sqrt n/\eps).
\]
Consequently, for every cut \(S\), the error is 
\[
    \left|
    \mathbf 1_S^\top(L_G-L_{\widetilde G})\mathbf 1_S
    \right|
    \le
    \widetilde O(|S|\sqrt n/\eps).
\]
Their key observation is that this error can be absorbed into a \((1+\gamma)\)-multiplicative error only when the true value of every cut $S$ is at least
$\widetilde \Omega(|S|\sqrt n/(\eps\gamma)).$
Equivalently, the graph must have expansion parameter roughly
$\psi=\widetilde \Theta(\sqrt n/(\eps\gamma))$
with respect to the usual cardinality demand \(\rho(S)=|S|\) (see \cref{eq:expansion_definition} for the definition of expansion). 

Given this observation, the graph is privately decomposed into such expanders by iteratively using the private Laplacian estimator $L_{\widetilde{G}}$ to find sparsest cuts. In expanders, cuts are multiplicatively estimated. A \(\psi\)-expander decomposition leaves
$
    \widetilde O(n\psi)
$ 
residual edges~\cite{nanongkai2017dynamic}. With \(\psi\approx \sqrt n\), the residual therefore has
$
    m_{\mathrm{res}}=\widetilde O(n^{3/2})
$ 
edges. Applying the standard purely additive private cut release of Eli{\'a}{\v{s}} et al.~\cite{eliavs2020differentially} to this residual gives error
$
    \widetilde O(\sqrt{n m_{\mathrm{res}}})
    =
    \widetilde O(n^{5/4}).
$

The bottleneck is thus twofold. The local spectral primitive has global error scale \(\widetilde O(\sqrt n)\), forcing the expander threshold to remain at \(\sqrt n\) and leaving \(n^{3/2}\) residual edges. Moreover, the final release is purely additive, so even if the residual were somewhat smaller, the terminal cost would still be governed by the square-root expression \(\widetilde O(\sqrt{nm_{\mathrm{res}}})\).

\subsubsection{Our approach: recursive sparsification plus a relative terminal oracle}

Our framework improves both parts of the above pipeline.

\begin{enumerate}
    \item \textbf{Degree-sensitive recursive sparsification.}
    Instead of using a worst-case \(\widetilde O(\sqrt n)\)-spectral primitive, we use our new degree-sensitive private Laplacian primitive in Theorem~\ref{thm:intro_laplacian}. Its error decreases as the residual graph becomes sparser. This allows us to apply expander decomposition recursively: after each round, the residual becomes sparser, which in turn makes the next spectral call more accurate.

    \item \textbf{Edge-sensitive terminal release.}
    Once the residual is sufficiently sparse, we do not use the standard purely additive \(\widetilde O(\sqrt{mn})\) release. Instead, we use the relative-error terminal oracle from Theorem~\ref{thm:intro-edge-sensitive-cut-oracle}, which gives error
    $
        \widetilde O_{\eps,\delta,\gamma}
        \left(
            n+\left(n^2M\right)^{1/3}
        \right)
    $
    on a residual with at most \(M\) edges.
\end{enumerate}

A technical issue remains before a degree-sensitive primitive can be really used. A sparse graph may still have very large maximum degree; for example, a star has only \(O(n)\) edges but maximum degree \(\Theta(n)\). Since our spectral primitive depends on maximum degree instead of average degree (which is unavoidable due to Theorem~\ref{thm:result-lower-bound}), we first flatten degrees using a stable port gadget and then run a demand-aware expander decomposition. This yields the following four-step pipeline.

\paragraph{Step 1: Degree flattening via the stable port gadget.}
Let \(R\) be the current residual graph with \(m\) edges. Set
$
    d_{\mathrm{avg}}:=2m/n.
$ 
The stable port gadget replaces each vertex \(v\) by roughly
\[
    b_v\approx \max\left\{1,\frac{\deg_R(v)}{d_{\mathrm{avg}}}\right\}
\]
ports. Incident edges are routed to ports using public hashing and stable truncation. The resulting proxy graph, denoted \(^{\sharp} R\), has
\begin{equation}\label{eq:total-demand}
    \sum_{v\in[n]} b_v =\sum_{v\in[n]} \max\{1, deg(v)/d_{\text{avg}}\} =  O(n)
\end{equation}
vertices and maximum degree
$
    \Delta(R^\sharp)= O(d_{\mathrm{avg}})= O(m/n).
$
Moreover, changing one original edge changes only \(O(1)\) edges in the proxy graph, so the transformation preserves edge-level privacy up to constant factors. Cuts in the original graph are represented by clone-consistent cuts in the proxy graph.

\paragraph{Step 2: Demand-aware expansion.}
Vertex splitting changes the correct notion of cut size. A set containing a high-degree original vertex may contain many ports, even if it contains only one original vertex. Thus cardinality is no longer the right scale.

We assign each original vertex a demand
$
    \rho(v):=b_v,
$ 
so that the number of ports representing a set \(S\) is exactly \(\rho(S) = \sum_{v\in S}\rho(v)\). A demand-aware \(\psi\)-expander guarantees that every internal cut has size at least
$
    \psi\cdot \min\{\rho(S),\rho(V\setminus S)\}.
$ 
This matches the way the spectral error scales after the port transformation. Let \(m_t\) be the public upper bound of residual edge count at round \(t\), and write
$
    d_t:=2m_t/n.
$ 
Suppose the private spectral primitive on a graph of degree \(d\) has error
$
    \Pi(n,d,\eps,\delta).
$ 
For further decomposing the residual with at most $m_t$ edges, after flattening the degrees to make sure that the maximum degree is around $O(d_t)$, we choose the expansion parameter:
\[
    \psi_t
    \asymp
    \frac{\Pi(n,d_t,\eps,\delta)}{\gamma}.
\]
This ensures not only that such $\psi_t$-expanders can be found privately by our spectral primitives, and that the spectral additive error is absorbed into a \(\gamma\)-fraction of the true cut value inside the expander piece.

\paragraph{Step 3: Recursive sparsification.}
Due to \cite{nanongkai2017dynamic} and \cref{eq:total-demand}, the demand-aware expander decomposition leaves only
$
    \widetilde O(\psi_t \rho(V)^{1+o(1)}) = \widetilde O(\psi_t n^{1+o(1)})
$ 
edges for the next residual. Hence, we start from the trivial upper bound $m_0 = n^2$, and the recursion satisfies:
\[
    m_{t+1}
    \le
    \widetilde O_{\eps,\delta,\gamma}
    \left(
        n^{1+o(1)}
        \Pi\left(n,\frac{m_t}{n},\eps,\delta\right)
    \right).
\]
Equivalently,
\[
    d_{t+1}
    \le
    \widetilde O_{\eps,\delta,\gamma}
    \left(
        \Pi(n,d_t,\eps,\delta)
    \right).
\]
Our bootstrapped fourth-power spectral primitive gives
\[
    \Pi(n,d,\eps,\delta)
    =
    \widetilde O_{\eps,\delta}
    \left(
        \frac d{\sqrt q}
        +
        n^{1/8}d^{1/2}q^{1/8}
        +
        n^{1/4}q^{1/4}
    \right),
\]
where the tunable \(q \geq\operatorname{polylog} (n,1/\epsilon,1/\gamma,\log(1/\del))\) is chosen large enough to absorb privacy, failure-probability, and decomposition losses. The first term \(d/\sqrt q\) is a genuine polylogarithmic contraction. It is crucially different from the known \(\widetilde O(d)\) degree-only spectral bound, which would not contract the residual. The polynomial fixed point is determined by the two sublinear terms. Ignoring polylogarithmic factors, the recurrence is governed by
\[
    d_{t+1}
    \lesssim
    n^{1/8}d_t^{1/2}
    +
    n^{1/4}.
\]
Both terms have the same fixed point:
\[
    d_{t+1} \gets n^{1/8}d_t^{1/2}
    \quad\Longrightarrow\quad
    d_*=n^{1/4},
\]
and trivially
\[
    d_{t+1} \gets n^{1/4}
    \quad\Longrightarrow\quad
    d_*=n^{1/4}.
\]
Thus after \(T = O(\log n)\) recursive rounds, $d_T=\widetilde O(n^{1/4+o(1)}),$
and the final residual has
$
    m_T=nd_T/2=\widetilde O(n^{5/4+o(1)})
$
edges.

\paragraph{Step 4: The terminal end-game.}
If we applied the standard purely additive private cut release to the terminal residual, the error would be
$
    \widetilde O(\sqrt{nm_*})
    =
    \widetilde O(n^{9/8+o(1)}).
$ 
Instead, we use the edge-sensitive terminal oracle 
(Theorem~\ref{thm:intro-edge-sensitive-cut-oracle}). On a residual with at most \(M\) edges, it gives a mixed relative/additive guarantee with additive term
$
    \widetilde O_{\eps,\delta,\gamma}
    \left(
        n+\left(n^2M\right)^{1/3}
    \right).
$ 
Substituting $M=m_T=\widetilde O(n^{5/4+o(1)})$
yields
$\left(n^2\cdot n^{5/4+o(1)}\right)^{1/3}
    =
    n^{13/12+o(1)}$.
    
Therefore the final synthetic graph satisfies, simultaneously for all cuts \(S\subseteq V\),
\[
    |\cut_{\widehat G}(S)-\cut_G(S)|
    \le
    \gamma\cut_G(S)
    +
    \widetilde O_{\eps,\delta,\gamma}(n^{13/12+o(1)}).
\]
This proves the claimed worst-case private cut approximation bound in Theorem~\ref{thm:result-main-cut}. Conceptually, the recursion and the terminal oracle play complementary roles: the recursion reduces the residual edge count from \(n^2\) to \(n^{5/4+o(1)}\), while the terminal oracle converts this sparsity into a final additive error better than the purely additive square-root release.

\section{Preliminaries}\label{sec:preliminary}

\subsection{Basics in Differential Privacy}\label{def:dp}
For two datasets $D,D'$ in some data domain $\mathcal{D}^*$, we denoted as $D \sim D'$ if they are \emph{neighboring}. Fix any $\eps \geq 0, 0\leq \delta \leq 1$. A randomized algorithm $\+A$ is $(\epsilon, \delta)$-differentially private if for any pair of neighboring datasets $D \sim D'$ and any set of outcomes $\+O$, $$\Prb[\+A(D) \in \+O] \le \e^\epsilon \Prb[\+A(D') \in \+O] + \delta.$$


\begin{lemma}[Post processing \cite{dwork2006calibrating}]
\label{lem:postprocess}
    Let $\mathcal{M}:\mathcal{D}^*\rightarrow \mathcal{O}$ be an $(\epsilon,\delta)$-differentially private algorithm. Let $f:\mathcal{O}\rightarrow \mathcal{O}'$ be an arbitrary randomized mapping, then $f\circ \mathcal{M}: \mathcal{D}^* \rightarrow \mathcal{O'}$ is also $(\epsilon,\delta)$-differentially private.
\end{lemma}

\begin{lemma}[Sequential Composition \cite{dwork2014algorithmic}]
\label{lem:composition}
Let $\mathcal{M}_1, \dots, \mathcal{M}_k$ be a sequence of randomized algorithms, where each $\mathcal{M}_i$ takes as input a dataset $D \in \mathcal{D}^*$ and a transcript of all previous outputs $y_{<i} = (y_1, \dots, y_{i-1})$. If for every $i \in [k]$ and every fixed transcript $y_{<i}$, the algorithm $\mathcal{M}_i(\cdot, y_{<i})$ satisfies $(\eps_i, \del_i)$-differential privacy, then the joint mechanism $\mathcal{M}(D) = (\mathcal{M}_1(D), \dots, \mathcal{M}_k(D))$ satisfies $\left(\sum_{i=1}^k \eps_i, \sum_{i=1}^k \del_i\right)$-differential privacy.
\end{lemma}

\begin{lemma}[Advanced Composition~\cite{dwork2010boosting, dwork2014algorithmic}]\label{lem:advanced_composition}
Let $\eps \in (0, 1)$ and $\delta \in (0, 1)$. Suppose an algorithm adaptively applies a sequence of $k$ randomized mechanisms, where each mechanism satisfies $(\eps_0, \delta_0)$-differential privacy. There exists an absolute constant $c > 0$ such that if
\[
    \eps_0 \le c \frac{\eps}{\sqrt{k \log(2/\delta)}} \qquad \text{and} \qquad \delta_0 \le \frac{\delta}{2k},
\]
then the overall composite algorithm satisfies $(\eps, \delta)$-differential privacy.
\end{lemma}

The following introduces the Laplace distribution and its standard tail bound:
\begin{definition}
    [Laplace distribution] Given parameter $b$, the Laplace distribution (with scale $b$) is the distribution with probability density function 
    \[\text{Lap}(x|b) = \frac{1}{2b} \exp\left(-\frac{|x|}{b}\right).\]
\end{definition}
We use $\text{Lap}(b)$ to denote the Laplace distribution with scale $b$.
\begin{lemma}
    [Tail bound on Laplace distribution]
    \label{lem:laplace-tail}
    Let $x$ be a random variable with $\text{Lap}(b)$ distribution. Then, $\mathbf{Pr}[{|x| \geq tb}] \leq \exp{(-t)}$.
\end{lemma}

We also describe mechanisms that are commonly used in differential privacy. We begin by defining the sensitivity of a function.
\begin{definition}
    [$\ell_p$-sensitivity] Let $f:\mathcal{D}^*\rightarrow \mathbb{R}^k$ be a query function on datasets. The $\ell_{p}$-sensitivity of $f$ (with respect to $\mathcal{D}^*$) is 
    \[\Delta^{(p)} (f) = \max_{D,D'\in \mathcal{D}^*,\atop D\sim D'} \|f(D) - f(D')\|_p.\]
\end{definition}

\begin{lemma}[Laplace mechanism \cite{dwork2006calibrating}]
\label{lem:laplace-priv}
     Fix any $\epsilon>0$. Suppose $f:\mathcal{D}^*\rightarrow \mathbb{R}^k$ is a query function with $\ell_1$-sensitivity $\Delta^{(1)}(f)$. Then the mechanism
    \begin{align*}
        \mathcal{M}(D) = f(D) + (Z_1,\cdots,Z_k)^\top
    \end{align*}
    is $(\epsilon,0)$-differentially private, where $Z_1,\cdots, Z_k$ are i.i.d random variables drawn from $\text{Lap}(\Delta^{(1)}(f)/\epsilon)$.
\end{lemma}

\begin{lemma}[Gaussian Mechanism \cite{dwork2014algorithmic}]
\label{lem:gaussian-mechanism}
Fix any $\epsilon \in (0, 1)$ and $\delta \in (0, 1)$. Suppose $f:\mathcal{X}\rightarrow \mathbb{R}^k$ is a query function with $\ell_2$-sensitivity $\Delta^{(2)}(f)$. Then the Gaussian Mechanism $\mathcal{M}$, defined by
\[
\mathcal{M}(D) = f(D) + \mathcal{N}(0, \sigma^2 \mathbf{I}_k),
\]
is $(\epsilon, \delta)$-differentially private if the noise scale $\sigma$ satisfies
\[
\sigma \geq \frac{\Delta^{(2)}(f) }{\epsilon} \sqrt{2 \ln\left(\frac{1.25}{\delta}\right)}.
\]
\end{lemma}

\subsection{Spectral Graph Theory}

\paragraph{Graph Laplacian Matrix.} Consider a non-negative weighted graph $G$. Let $A_G\in \mathbb{R}^n\times \mathbb{R}^n$ be the symmetric adjacency matrix of $G$. That is,
$A_G[u,v] = A_G[v,u] = w_{uv}$, where $u,v\in V$ are vertices, and $w_{u,v}$ is the weight of edge $uv$. (If $u$ is not adjacent with $v$, then $w_{uv} = 0$.) For vertex $v\in V$, denote by $d_v$ the (weighted) degree of $v$, namely $d_v = \sum_{u\neq v} w_{vu}$. Let $D_G\in \mathbb{R}^n\times \mathbb{R}^n$ be a diagonal matrix where $D_G[i,i]$ is the degree of $i\in[n]$. Here, we define the edge adjacency matrix:

\begin{definition}
    [Edge adjacency matrix] Let $G$ be an undirected graph of $n$ vertices and $m$ edges. Consider an arbitrary orientation of edges, then $E_G\in \mathbb{R}^{m\times n}$ is the edge adjacency matrix where 
    $$E_G[e,v] = \left\{
        \begin{aligned}
            &+\sqrt{w_e}, &\text{if } v \text{ is } e\text{'s head,} \\
            &-\sqrt{w_e}, &\text{if } v \text{ is } e\text{'s tail,} \\
            &0, &\text{otherwize.}
        \end{aligned}
    \right.$$ 
\end{definition}
\noindent An important object of interest in graph theory is the Laplacian of a graph:
\begin{definition}
    [Laplacian matrix] For an undirected graph $G$ with $n$ vertices and $m$ edges, the Laplacian matrix $L_G\in \mathbb{R}^{n\times n}$ of $G$ is $$L_G = E_G^\top E_G.$$
\end{definition}

\noindent Equivalently, one can verify that $L_G = D_G - A_G$. Also, we note that for any graph $G$, $L_G$ is a positive semi-definite (PSD) matrix and $L_G \bm{1} = 0$, where $\bm{1}\in \mathbb{R}^n$ is the all one vector. For any vector $x$ in $\mathbb{R}^n$, the quadratic form of $L_G$ is $x^\top L_G x  \geq 0$. In particular, one can verify that 
$$x^\top L_G x = \sum_{(u,v)\in E} w_{u,v} (x(u) - x(v))^2.$$

\paragraph{Demand Sparsity and Expander Graphs.}
To analyze the private expander decomposition, we use the generalized version of the standard cardinality-based sparsity to a \emph{demand-aware} setting.
Given a graph $G=(V, E)$ and a demand vector $\rho: V \to \^R_{\ge 1}$, we extend the demand to any subset $S \subseteq V$ as $\rho(S) = \sum_{v \in S} \rho(v)$. Fix an $U\subseteq V$. The \emph{demand sparsity} of a cut $(S, U \setminus S)$ in an induced subgraph $G[U]$ is defined as:
\begin{equation}\label{eq:expansion_definition}
    \phi_\rho^{G[U]}(S) = \frac{w_{G[U]}(S, U \setminus S)}{\min\{\rho(S), \rho(U \setminus S)\}}.
\end{equation}
We say $G[U]$ is a $(\psi, \rho)$-expander if for all $\emptyset \ne S \subsetneq U$, $\phi_\rho^{G[U]}(S) \ge \psi$.

\subsection{Useful Tools in Probability and Matrix}

\begin{lemma}[Properties and Tail Bounds of the $\chi_1^2$ Distribution, \cite{vershynin2018high}] \label{lem:chi_square_tails}
Let $Z \sim \mathcal{N}(0,1)$ be a standard Gaussian random variable. Its square $X = Z^2$ follows a chi-squared distribution with one degree of freedom, denoted as $X \sim \chi_1^2$. The following properties hold:
\begin{enumerate}
    \item \textbf{Moments:} $\E[X] = 1$ and $\E[X^2] = 3$.
    \item \textbf{Sub-exponential Tail:} There exist absolute constants $C, c > 0$ such that for any $s \ge 0$, $$ \Prb[X > s] \le C e^{-cs}.$$
\end{enumerate}
\end{lemma}

\begin{lemma}[Matrix Gaussian Series Inequality, \cite{tropp2012user}]
\label{lem:matrix-gaussian-series}
Let $A_1, \dots, A_k \in \mathbb{R}^{n \times n}$ be deterministic symmetric matrices, and let $g_1, \dots, g_k \sim \mathcal{N}(0,1)$ be independent standard Gaussian random variables. Consider the random symmetric matrix $Z = \sum_{i=1}^k g_i A_i$. Define the matrix variance proxy parameter $v$ as:
\[
    v = \opnorm{ \sum_{i=1}^k A_i^2 }.
\]
Then, for any failure probability $\beta \in (0, 1)$, with probability at least $1-\beta$, the spectral norm of $Z$ satisfies the tail bound:
\[
    \opnorm{Z} \le \sqrt{2v \log(2n/\beta)}.
\]
\end{lemma}

\begin{lemma}[Hadamard's Inequality, \cite{horn2012matrix}] \label{lem:hadamard}
Let $X \in \mathbb{R}^{n \times n}$ be a symmetric positive semi-definite matrix. Then the determinant of $X$ is upper-bounded by the product of its diagonal entries:
\begin{equation}
    \det(X) \le \prod_{i=1}^n X_{ii}.
\end{equation}
Equality holds if and only if $X$ is a diagonal matrix, or if any diagonal entry $X_{ii} = 0$. 
\end{lemma}

To bound the retained rank of our spectral truncation, we recall on the following variational characterization of eigenvalues.

\begin{theorem}[Ky Fan's Maximum Principle, Corollary 4.3.39 in \cite{horn2012matrix}] \label{thm:kyfan}
Let $A \in \mathbb{R}^{n \times n}$ be a symmetric matrix with eigenvalues $\lambda_1 \ge \lambda_2 \ge \dots \ge \lambda_n$. For any integer $1 \le k \le n$, the sum of the top $k$ eigenvalues satisfies:
$$ \sum_{i=1}^k \lambda_i = \max_{U \in \mathbb{R}^{n \times k}, U^T U = I_k} \text{tr}(U^T A U). $$
Consequently, for any positive semi-definite (PSD) matrix $A \succeq 0$ and any set of $k$ orthonormal vectors $\{u_1, \dots, u_k\}$, we have:
$$ \sum_{i=1}^k u_i^T A u_i \le \sum_{i=1}^k \lambda_i \le \text{tr}(A). $$
\end{theorem}

Furthermore, to explicitly characterize the noise magnitude in the matrix mechanism, we recall the standard non-asymptotic bound for the spectral norm of symmetric Gaussian matrices.

\begin{lemma}[Spectral Norm of Symmetric Gaussian Matrices, Theorem 4.4.5 in \cite{vershynin2018high}] \label{lem:gaussian_norm}
Let $W \in \mathbb{R}^{n \times n}$ be a symmetric random matrix where the upper triangular entries (including diagonal entries) are sampled i.i.d. from $\mathcal{N}(0, \sigma^2)$. Then, for any failure probability $\beta \in (0, 1)$, there exists an absolute constant $C > 0$ such that with probability at least $1 - \beta$, the operator norm of $W$ satisfies:
$$ \norm{W}_2 \le C\sigma\left(\sqrt{n} + \sqrt{\log(1/\beta)}\right). $$
\end{lemma}

Our vertex splitting gadget heavily relies on pseudo-random routing. To formally analyze this without compromising the deterministic bounds required by differential privacy, we define the following hash function family and introduce concentration inequalities.

\begin{definition}[Public Hash Functions] \label{def:hash}
Let $\+H$ be a family of hash functions uniform at randomly mapping a universe $\+U$ to a set of $k$ bins. For each vertex $v \in V$, we publicly sample a fresh hash function $h_v: V \setminus \{v\} \to [\rho(v)]$ from a family of fully independent hash functions.
\end{definition}
Crucially, these hash functions are generated using \emph{public randomness} and strictly do not consume any privacy budget. Next we introduce the standard Chernoff bounds.

\begin{lemma}[Chernoff Bound~\cite{chernoff1952measure}] \label{lem:chernoff}
Let $X_1, \dots, X_m$ be independent Bernoulli random variables (Poisson trials) such that $\Prb[X_i = 1] = p_i$. Let $X = \sum_{i=1}^m X_i$ and denote its expectation by $\mu = \E{X}$. Then, for any $\delta > 0$:
\begin{equation*}
    \Prb[X \ge (1+\delta)\mu] \le \Paren{ \frac{\e^\delta}{(1+\delta)^{1+\delta}} }^\mu.
\end{equation*}
For $\delta \ge 1$, this bound simplifies to $\Prb[X \ge (1+\delta)\mu] \le \exp(-\frac{\delta \mu}{3})$.
\end{lemma}
The following direct corollary of Lemma~\ref{lem:chernoff} is useful for analyzing our stable port gadget:
\begin{corollary}[Balls-into-Bins Concentration] \label{cor:balls_bins}
Suppose $m$ balls are independently distributed into $k$ bins using a random hash function $h \in \+H$. If the expected load of a specific bin is $\mu = m/k \ge 1$, then by applying Lemma \ref{lem:chernoff}, the maximum load across all $k$ bins is bounded by $O(\mu + \log(k/\beta))$ with probability at least $1-\beta$.
\end{corollary}

\section{Power-based Spectral Amplifier for Private Matrix Release}\label{sec:power_amplifier}

Our improvements in private graph analysis start from a new technique for the private estimation of general symmetric matrices, which we refer to as the \emph{power-based spectral amplifier}. While our primary motivation is to deploy this tool for private spectral and cut approximation (in Section~\ref{sec:laplacain_analysis}), we present its privacy and utility analyses in full generality as they may be of independent interest. In particular, we show the following generic theorem for private matrix release:

\begin{theorem}[Power-based Spectral Amplifier; Combination of Theorem~\ref{thm:privacy} and Theorem~\ref{thm:utility}] \label{thm:result_generic_framework}
Fix privacy parameters $0<\eps,\delta<1$ and an integer $n\in \mathbb{N}_+$. Let $M(D) \in \mathbb{R}^{n \times n}$ be a symmetric matrix evaluated on a dataset $D$. For any integer $k \ge 1$, we define the $k$-th Frobenius sensitivity over neighboring datasets $D \sim D'$ as $\Delta_k = \max_{D \sim D'} \|M(D)^k - M(D')^k\|_F$. For any even integer $q \ge 2$, there exists an efficient $(\epsilon, \delta)$-differentially private algorithm such that for any input matrix $M(D)$ satisfying $\text{tr}(M(D)^q) \le B_q$, it outputs a synthetic matrix $\widehat{M}$ satisfying, with high probability:
\begin{equation}
    \norm{M(D)-\widehat M}_2
    \le
    \widetilde O\left(
        \left(\frac{\Delta_q\sqrt n}{\eps}\right)^{1/q}
        +
        \frac{\Delta_1}{\eps}
        \sqrt{
            \frac{B_q\eps}{\Delta_q\sqrt n}
        }
    \right).
\end{equation}
\end{theorem}

As illustrated in Section~\ref{sec:tech_overview-laplacian-oracle}, the core idea of the spectral amplifier is: utilizing a non-linear matrix mapping $\Phi_q(M) = M^q (q\geq 2)$ to amplify the true signal, specifically the spectral gap. This amplification ensures that the true data features contained in the subspace can still be truncated and extracted with high probability after adding privacy noise. Consequently, this approach restricts subsequent queries to this private, low-dimensional subspace, thereby achieving an error bound much smaller than that incurred by naive additive noise mechanisms that add noise directly into the original $n$-dimensional space.


\subsection{Definitions and Algorithm}\label{sec:framework_notions}

Given a dataset $D$, our goal is to privately release a symmetric real matrix $M(D) \in \mathbb{R}^{n \times n}$. For simplicity, we drop the dependency on $D$ and write $M$ when the context is clear. We denote $D \sim D'$ to indicate that $D$ and $D'$ are neighboring datasets. 

\begin{itemize}
    \item \textbf{Power Amplifier:} For any symmetric matrix $M$ with eigen-decomposition $M = \sum_{i\in [n]} \lambda_iv_iv^\top_i$, we define the amplifying operator as $\Phi_q(M) = M^q = \sum_{i\in [n]} \lambda_i^q v_iv^\top_i$ for any \textbf{even integer} $q \ge 2$. Because $q$ is even, the amplified matrix $M^q$ is positive semi-definite (i.e., $M^q \succeq 0$), regardless of whether $M$ contains negative eigenvalues.\footnote{For example, graph adjacency matrices $A$ has both positive and negative eigenvalues.}
    
    \item \textbf{Sensitivity Bounds:} The mechanism design relies on the $\ell_2$-sensitivities (for vectorized matrix, which is equivalent to Frobenius norm) in both the original and the powered spaces:
    \begin{itemize}
        \item Original Frobenius sensitivity: $\Delta_M = \Delta_1 = \max_{D \sim D'} \norm{M(D) - M(D')}_F$
        \item Amplified Frobenius sensitivity: $\Delta_q = \max_{D \sim D'} \norm{M(D)^q - M(D')^q}_F$
    \end{itemize}

    \item \textbf{Trace Bound:} We assume there exists an (pure analytical) upper bound $B_q$ such that for all valid datasets, the trace of the amplified matrix is bounded: $\text{tr}(M^q) \le B_q$.\footnote{Importantly, this bound $B_q$ is a structural property used solely in our utility analysis to bound the retained rank; it \emph{does not} need to be public or provided as an input parameter to the privacy mechanism.}
\end{itemize}
With these parameters, the Power-based Spectral Amplifier Framework proceeds in Algorithm~\ref{alg:spectral-amplifier}.

\begin{algorithm}[h]
\caption{Power-based Spectral Amplifier}
\label{alg:spectral-amplifier}
\begin{algorithmic}[1]
\REQUIRE A dataset $D$ with target symmetric matrix $M(D) \in \mathbb{R}^{n \times n}$, privacy parameters $\epsilon, \delta \in (0, 1)$, failure probability $\beta =1/n^c$, and an even integer $q \ge 2$.
\ENSURE A synthetic matrix $\widehat{M}$.

\STATE Split the privacy budget such that $\epsilon_1 = \epsilon_2 = \epsilon/2$ and $\delta_1 = \delta_2 = \delta/2$; similarly allocate the failure probability $\beta_1 = \beta_2 = \beta/2$\;

\textbf{Step 1: Amplified Spectral Release:}
\STATE Compute the amplified matrix $Y = M^q$.
\STATE Sample a symmetric Gaussian noise matrix $W \in \mathbb{R}^{n \times n}$ by drawing its upper-triangular entries i.i.d. from $\mathcal{N}(0, \sigma_1^2)$ and setting $W_{j,i} = W_{i,j}$ for $i > j$, with noise scale $\sigma_1 = \Delta_q \frac{\sqrt{2\log(1.25/\delta_1)}}{\epsilon_1}$.
\STATE Construct the noisy amplified matrix $\widetilde{Y} = Y + W$.

\textbf{Step 2: Spectral Truncation:}
\STATE Let the truncation threshold be $\eta = C\sigma_1 (\sqrt{n} + \sqrt{\log(1/\beta_1)})$, such that by Lemma \ref{lem:gaussian_norm}, $\norm{W}_2 \le \eta$ holds with probability at least $1-\beta_1$.
\STATE Compute the spectral decomposition $\widetilde{Y} = \sum_{i=1}^n \widetilde{\lambda}_i u_i u_i^\top$. Let $P = \sum_{j:\, \widetilde{\lambda}_j > 2\eta} u_j u_j^\top$ be the orthogonal projector onto the eigenspace associated with eigenvalues greater than $2\eta$.
\STATE Extract an orthonormal basis $U \in \mathbb{R}^{n \times r}$ for $\text{range}(P)$, where $r = \text{rank}(P)$.

\textbf{Step 3: Exact Release:}
\STATE Project the original matrix into the low-dimensional subspace: $C = U^\top M U$.
\STATE Sample a symmetric Gaussian noise matrix $E \in \mathbb{R}^{r \times r}$ by drawing its upper-triangular entries i.i.d. from $\mathcal{N}(0, \sigma_2^2)$ and setting $E_{j,i} = E_{i,j}$ for $i > j$, with noise scale $\sigma_2 = \Delta_M \frac{\sqrt{2\log(1.25/\delta_2)}}{\epsilon_2}$.
\STATE Construct the noisy compressed matrix $\widetilde{C} = C + E$.

\RETURN $\widehat{M} = U \widetilde{C} U^\top$.
\end{algorithmic}
\end{algorithm}

\subsection{Privacy Guarantee}

\begin{theorem}
\label{thm:privacy}
The Power-based Spectral Amplifier (Algorithm~\ref{alg:spectral-amplifier}) is $(\eps, \delta)$-differentially private.
\end{theorem}

\begin{proof}
The algorithm interacts with the private dataset $D$ through two sequential queries, and we bound the total privacy leakage via basic composition (Lemma~\ref{lem:composition}). 

In Step 1, the algorithm releases the symmetric matrix $\widetilde{Y} = M^q + W$. Since $\widetilde{Y}$ is symmetric, its output space is uniquely determined by its $n(n+1)/2$ upper-triangular entries (including the diagonal). The $\ell_2$ sensitivity of these elements is indeed bounded by the Frobenius sensitivity of the full matrix $M^q$: 
$$ \sqrt{\sum_{1 \le i \le j \le n} \left(M^q_{i,j} - (M')^q_{i,j}\right)^2} \le \norm{M^q - (M')^q}_F \le \Delta_q. $$
Thus, by the Gaussian mechanism (Lemma~\ref{lem:gaussian-mechanism}), adding independent Gaussian noise $\mathcal{N}(0, \sigma_1^2)$ calibrated to $\sigma_1 = \Delta_q \sqrt{2\log(1.25/\delta_1)} / \eps_1$ to the upper-triangular entries (and mirroring them to the lower triangle) guarantees $(\eps_1, \delta_1)$-differential privacy.

In Step 2, the algorithm extracts the projection matrix $P$ and its orthogonal basis $U$ as a function solely of the privatized matrix $\widetilde{Y}$. By the post-processing property of differential privacy (Lemma~\ref{lem:postprocess}), this step does not consume additional privacy budget.

In Step 3, the algorithm performs a secondary query on the original dataset, releasing the matrix $C = U^\top M U$ confined within the fixed public subspace spanned by $U$. To bound the $\ell_2$ sensitivity of this linear query, consider any two neighboring datasets yielding symmetric matrices $M$ and $M'$. The difference in the query output is $U^\top (M - M') U$. Using the cyclic property of the trace and the identity $U U^\top = P$, the Frobenius norm of the difference is bounded by:
\begin{align*}
    \norm{U^\top M U - U^\top M' U}_F^2 &= \norm{U^\top(M-M')U}_F^2 \\
    &= \text{tr}\Big( U^\top(M-M')^\top U U^\top (M-M') U \Big) \\
    &= \text{tr}\Big( (M-M') P (M-M') P \Big).
\end{align*}
Since $P$ is an orthogonal projection matrix, its eigenvalues are exactly $0$ or $1$, establishing $P \preceq I$. Consequently, the trace is bounded by $\text{tr}((M-M')^2) = \norm{M-M'}_F^2 \le \Delta_M^2$. Thus, the Frobenius sensitivity of $U^\top M U$ remains bounded by $\Delta_M$. 

Analogous to Step 1, since $C = U^\top M U$ is also symmetric, we only need to release its upper-triangular entries. Their $\ell_2$ sensitivity is bounded by $\norm{C - C'}_F \le \Delta_M$. Adding independent Gaussian noise $\mathcal{N}(0, \sigma_2^2)$ to these entries (and mirroring them to construct the symmetric matrix $E$) with scale $\sigma_2 = \Delta_M \sqrt{2\log(1.25/\delta_2)} / \eps_2$ guarantees $(\eps_2, \delta_2)$-DP for this final query.

Finally, by the adaptive composition theorem (Lemma~\ref{lem:composition}), the overall algorithm satisfies $(\eps_1 + \eps_2, \delta_1 + \delta_2)$-DP, which concludes the proof.
\end{proof}

\subsection{Utility Analysis}

\begin{theorem}[Spectral Error Bound] \label{thm:utility}
Let $q \ge 2$ be an even integer. Algorithm~\ref{alg:spectral-amplifier} generates a synthetic matrix $\widehat{M}$ such that with probability at least $1-1/n^c$:
\begin{equation}\label{eq:power-amplifier-general}
    \norm{M-\widehat M}_2
    \le
    \widetilde O_{\del}\left(
        \left(\frac{\Delta_q\sqrt n}{\eps}\right)^{1/q}
        +
        \frac{\Delta_M}{\eps}
        \sqrt{
            \frac{B_q\eps}{\Delta_q\sqrt n}
        }
    \right).
\end{equation}
\end{theorem}

\begin{proof}
Recall that the reconstructed matrix is given by $\widehat{M} = U \widetilde{C} U^T$. By substituting the definition of the noisy matrix $\widetilde{C} = U^T M U + E$, we have:
$$ \widehat{M} = U (U^T M U + E) U^T = U U^T M U U^T + U E U^T. $$
Notice that $U \in \mathbb{R}^{n \times r}$ consists of orthonormal column vectors spanning the range of $P$, meaning $U U^T = P$ and $U^T U = I_r$. Substituting this yields $\widehat{M} = P M P + U E U^T$. We decompose the total estimation error using the triangle inequality:
$$ \norm{M - \widehat{M}}_2 \le \underbrace{\norm{M - P M P}_2}_{\text{Truncation Bias}} + \underbrace{\norm{U E U^T}_2}_{\text{Subspace Noise}}. $$
Because $U$ has orthonormal columns, for any vector $x$, $\norm{U E U^T x}_2 \le \norm{E}_2 \norm{x}_2$. Hence, the subspace noise is at most $\norm{E}_2$. 

First, we bound the truncation bias. Condition on the high-probability event $\mathcal{E}_1$ that $\norm{W}_2 \le \eta$, which occurs with probability at least $1-\beta_1$ (Lemma \ref{lem:gaussian_norm}). For any unit vector $x \perp \text{range}(P)$, its Rayleigh quotient on $\widetilde{Y}$ satisfies $x^T \widetilde{Y} x \le 2\eta$. Since $\widetilde{Y} = M^q + W$ and $\norm{W}_2 \le \eta$, the Rayleigh quotient on the true amplified matrix is bounded by:
$$ x^T M^q x = x^T \widetilde{Y} x - x^T W x \le 2\eta + \eta = 3\eta, \text{ where } x \perp \text{range}(P). $$
To bound the spectral norm of $M$ in the orthogonal complement $P^\perp$, we evaluate $\norm{Mx}_2^2 = x^T M^2 x$. Let $M = \sum_{i=1}^n \lambda_i v_i v_i^T$ be the eigendecomposition of $M$. We have $x^T M^2 x = \sum_{i=1}^n \lambda_i^2 (v_i^T x)^2$. Because $q \ge 2$, the mapping $f(z) = z^{2/q}$ is a concave function for $z \ge 0$. Writing $\lambda_i^2 = (|\lambda_i|^q)^{2/q}$, and noting that $\sum_{i=1}^n (v_i^T x)^2 = \norm{x}_2^2 = 1$, we apply Jensen's inequality:
\begin{align*}
 \norm{Mx}_2^2 = \sum_{i=1}^n \left(|\lambda_i|^q\right)^{2/q} (v_i^T x)^2 
 &\le \left( \sum_{i=1}^n |\lambda_i|^q (v_i^T x)^2 \right)^{2/q} \\
 &= \left( x^T M^q x \right)^{2/q} \le (3\eta)^{2/q}.
\end{align*}
Taking the square root yields $\norm{Mx}_2 \le (3\eta)^{1/q}$ for any unit vector $x \perp \text{range}(P)$. To bound the matrix operator norm $\norm{M(I-P)}_2$, we note that:
\begin{equation}\label{eq:project_space}
        \norm{M(I-P)}_2 = \sup_{y \in \mathbb{R}^n, \, \norm{y}_2 = 1} \norm{M(I-P)y}_2.
\end{equation}
For any arbitrary unit vector $y \in \mathbb{R}^n$, let $z = (I-P)y$. Since $I-P$ is a projector orthogonal to $P$, then $z \perp \text{range}(P)$ and that $\norm{z}_2 \le \norm{y}_2 = 1$. If $z = \mathbf{0}$, the bound holds trivially. Otherwise, we can normalize it into a valid unit vector $\widehat{z} = z / \norm{z}_2 \perp \text{range}(P)$, then:
\[
    \norm{M(I-P)y}_2 = \norm{Mz}_2 = \norm{z}_2 \norm{M\widehat{z}}_2 \le 1 \cdot (3\eta)^{1/q} = (3\eta)^{1/q}.
\]
Taking this back to \cref{eq:project_space} yields $\norm{M(I-P)}_2 \le (3\eta)^{1/q}$. We then decompose the residual matrix:
$$ \norm{M - PMP}_2 \le \norm{M(I-P)}_2 + \norm{PM(I-P)}_2 \le 2\norm{M(I-P)}_2 \le 2(3\eta)^{1/q}. $$
Next, we bound the variance inside the subspace. Still conditioning on $\mathcal{E}_1$, let $\{u_1, \dots, u_r\}$ be the orthonormal basis forming $P$. For each basis vector, $u_i^T M^q u_i = u_i^T \widetilde{Y} u_i - u_i^T W u_i > 2\eta - \eta = \eta$. Summing this over all $r$ dimensions yields $r\eta < \sum_{i=1}^r u_i^T M^q u_i$. By Ky Fan's Maximum Principle (Theorem~\ref{thm:kyfan}), because $M^q \succeq 0$, this sum is bounded by the total trace $B_q$:
$$ r\eta < \sum_{i=1}^{r} u_i^T M^q u_i \le \text{tr}(M^q) \le B_q \implies r \le \frac{B_q}{\eta}. $$
Conditioning on a second event $\mathcal{E}_2$ that the spectral norm of the Gaussian noise $E \in \mathbb{R}^{r \times r}$ follows the concentration bound (Lemma \ref{lem:gaussian_norm}):
$$ \norm{E}_2 \le O_{\del}\left(\sigma_2 \sqrt{r} + \sigma_2 \sqrt{\log(1/\beta_2)}\right) = \widetilde{O}_{\del}\left( \frac{\Delta_M}{\eps_2} \sqrt{\frac{B_q}{\eta}} \right), $$
where we let $\beta_2=1/n^c$. By a union bound, events $\mathcal{E}_1$ and $\mathcal{E}_2$ hold simultaneously with probability at least $1 - \beta$. Summing the truncation bias and the subspace noise, the overall error is:
$$ \norm{M - \widehat{M}}_2 \le \widetilde{O}_{\del}\left( \eta^{1/q} + \frac{\Delta_M}{\epsilon} \sqrt{\frac{B_q}{\eta}} \right). $$
Substituting $\eta = \widetilde{O}_{\del}(\sigma_1 \sqrt{n}) = \widetilde{O}_{\del}\left( \frac{\Delta_q \sqrt{n}}{\epsilon} \right)$ yields the desired bound \cref{eq:power-amplifier-general}, concluding the proof.
\end{proof}

\begin{remark}[The Square Amplifier Special Case and Column-Sparse Matrices]
\label{rek:general_of_amplifier}
We note that the square amplifier ($q=2$) presents an easy-to-use form of error. In particular, substituting $q=2$ directly into Theorem~\ref{thm:utility} yields:
\begin{equation}\label{eq:square-amplifier-general}
    \norm{M-\widehat M}_2
    \le
    \widetilde O_{\del}\left(
        \sqrt{
            \frac{\Delta_2\sqrt n}{\eps}
        }
        +
        \sqrt{
            \frac{\Delta_M^2\norm{M}_F^2}{\epsilon\Delta_2\sqrt n}
        }
    \right)
\end{equation}

This bounds the error by a function of the trace of the squared matrix ($\text{tr}(M^2) = \norm{M}_F^2$) and its squared sensitivity $\Delta_2 = \norm{M'^2 - M^2}_F$. Crucially, we can bridge these quantities with the matrix's maximum column $\ell_2$-norm, denoted as $R = \max_{u\in [n]} \norm{M e_u}_2 = \|M\|_{1\rightarrow2}$\footnote{For any matrix $\mathbf{A}\in \mathbb{R}^{n\times n}$ and integer $p,q\geq 1$, $   \|\mathbf A\|_{p \to q} = \max {\|\mathbf A\mathbf x\|_q \over \|\mathbf x\|_p}$.}. 
Furthermore, for any update $H$ supported on at most $O(1)$ entries with with $\norm{H}_F = O(1)$, the squared sensitivity expands to $$\norm{(M+H)^2 - M^2}_F \approx \norm{MH + HM}_F \le O(R).$$ Meanwhile, the trace bound is inherently $$\text{tr}(M^2) =\|M\|_F^2 =  \sum_u \norm{Me_u}_2^2 \le nR^2.$$ Thus substituting $\Delta_2 = O(R)$ and $\|M\|_F^2 = nR^2$ into \cref{eq:square-amplifier-general} gives the spectral error on private matrix estimation: 
\begin{equation}\label{eq:sparse_matrix}
    \norm{M-\widehat M}_2 \leq  \widetilde{O}_{\del}\left( n^{1/4} \sqrt{\frac{R}{\epsilon}} \right) = \widetilde{O}_{\del}\left( n^{1/4} \sqrt{\frac{\|M\|_{1\rightarrow 2}}{\epsilon}} \right).
\end{equation}
This new error bound highlights two powerful implications:
\begin{enumerate}
    \item \textbf{Column-Sparse Matrices:} It provides an out-of-the-box private release mechanism for various of sparse symmetric matrices (such as event-level co-occurrence or item-similarity matrices commonly used in NLP~\cite{pennington2014glove, wang2025beyond, he2026harmonizing} or recommendation systems~\cite{linden2003amazon}), outperforming the Gaussian mechanism (Lemma~\ref{lem:gaussian-mechanism}, which incurs $\widetilde{O}_{\del}(\sqrt{n}/\epsilon)$ error according to Lemma~\ref{lem:gaussian_norm}) in sparse regimes.
    \item \textbf{Graph Laplacian Approximation:} More importantly, as we will detail in Section~\ref{sec:laplacain_analysis}, using this error bound as an important building block, we will drive an $\widetilde{O}_{\eps,\del}((nd)^{1/4})$ error in approximating the Laplacian matrices of $n$-vertex graphs with bounded degree $d$, being the \emph{first} to bypass the longstanding $\min\{2d, \widetilde{O}_{\del}(\sqrt{n}/\eps)\}$ worst-case baseline~\cite{blocki2012johnson, upadhyay2021differentially}.
\end{enumerate}
The remainder of this paper focuses on graph spectral and cut approximations with privacy. While we believe that our generic spectral amplifier framework (Algorithm~\ref{alg:spectral-amplifier}) has the potential to facilitate more general private matrix analysis and is of independent interest, we do not find it necessary to discuss more extensions specifically in this paper.
\end{remark}

\section{Differentially Private Graph Laplacian Analysis}
\label{sec:laplacain_analysis}

Let \(G=([n],E)\) be a simple undirected graph with maximum degree at most \(d\). Let $A\in\set{0,1}^{n\times n}$
be the adjacency matrix of $G = ([n],E)$, where $$A_{ij} = 
\begin{cases} 
1 & \text{if } (i,j) \in E, \\
0 & \text{otherwise.}
\end{cases}.$$
We define graph Laplacian matrix as
$$L_G :=\operatorname{diag}(A\mathbf 1)-A.$$
Throughout this paper we adopt the well-established edge-level differential privacy in private graph analysis~\cite{dwork2006calibrating, gupta2010differentially, blocki2012johnson, eliavs2020differentially, liu2024optimal, zou2025generalized, aamandbreaking, chandra2026differentially, zoudifferentially}, in which two graphs
are adjacent if they differ by at most one undirected edge. 

Many advanced graph analysis tasks such as computing effective resistances~\cite{liu2024optimal}, expander decomposition or answering cut queries (Section~\ref{sec:cut}) rely inherently on the graph Laplacian matrix. 
In this section, we apply the power-based spectral amplifiers developed in Section~\ref{sec:power_amplifier} to privately approximate the graph Laplacian. To capture different structural requirements in downstream applications, we develop two distinct spectral primitives:

\begin{center}
  \begin{tcolorbox}[sharpish corners, colback=white, width=1\linewidth]
    \begin{center}
      \textbf{\emph{Our Private Spectral Primitives}}
    \end{center}
    \vspace{6pt}
    \begin{itemize}
        \item \textbf{Primitive I: The Square Amplifier (\Cref{sec:laplacian-square}).} By utilizing the $M^2$ amplifier, we provide a private Laplacian estimator with an additive error of $\widetilde{O}((nd)^{1/4})$. This bypasses the trivial worst-case baseline $\min \{2d, \widetilde{O}(\sqrt{n}/\eps)\}$ for unweighted graphs.
        
        \item \textbf{Primitive II: The Bootstrapped Fourth-Power Amplifier (Section~\ref{sec:bootstrapped_fourth_power}).} By applying an iterated cubic-walk preprocessing (bootstrap) followed by the $M^4$ amplifier, we derive an alternative error bound. Specifically, for a tunable parameter $q$, with high probability, the primitive guarantees an additive error of:
        \begin{equation}\label{eq:better_spectral_primitive}
            \opnorm{L_G- L_{\widetilde{G}}}
          \le
          \widetilde O_{\eps,\del}
          \left(
               \frac d{\sqrt q}
               +n^{1/8}d^{1/2}q^{1/8}
               +n^{1/4}q^{1/4}
          \right).
        \end{equation}
        While not universally beating the square amplifier, this specific error structure is essential for triggering the recursive expander decomposition procedure to finally get the $\widetilde{O}(n^{13/12+o(1)})$ error on private cut approximation(Section~\ref{sec:cut}).
    \end{itemize}
  \end{tcolorbox} 
\end{center}

At a high level, the standard approach of adding independent Gaussian noise directly into the Laplacian matrix incurs an worst-case error of $\widetilde{O}_{\del}(\sqrt{n}/\epsilon)$~\cite{upadhyay2021differentially}. To do better, our algorithm follows the three steps for both primitives:
\begin{enumerate}
    \item \textbf{Adjacency Estimation:} We first estimate the unweighted graph adjacency matrix $A$ using the power-based amplifier $\Phi_q(A) = A^q$ (incorporating the iterated cubic-walk preprocessing when $q=4$). 
    \item \textbf{Degree Estimation:} We independently approximate the graph's degree sequence using the standard Laplace mechanism, which has an $O(1)$ $\ell_1$-sensitivity under edge-level DP.
    \item \textbf{SDP Projection:} Simply combining the noisy adjacency matrix and noisy degrees yields an intermediate matrix $\widetilde{L}$ that is typically \emph{not} a valid graph Laplacian (e.g., possessing positive off-diagonal entries or non-zero row sums). This breaks the structural requirements for downstream spectral algorithms. To resolve this, we adopt the semidefinite programming (SDP) projection technique formulated by Upadhyay et al.~\cite{upadhyay2021differentially}, projecting $\widetilde{L}$ back onto the convex cone of valid non-negative weighted graph Laplacians.
\end{enumerate}

\subsection{Graph Laplacian Estimation via Square Amplifier}
\label{sec:laplacian-square}

Let $G = ([n], E)$ be a simple unweighted graph with maximum degree bounded by $d$. In this section we will formalize the application of the square amplifier ($q=2$) to achieve the $\widetilde{O}((nd)^{1/4})$ spectral approximation. First, we explicitly bound the amplified parameters $\Delta_2$ and $B_2$ for the unweighted graph adjacency matrix under edge-level differential privacy. All $\widetilde O_{\delta}(\cdot)$ factors hide polylogarithmic factors in $n, 1/\delta$, and $\log d$, but \emph{not} polynomial factors in $1/\eps$.

\begin{lemma}[Square Sensitivity for Unweighted Graphs]\label{lem:graph-square-sensitivity}
Given two neighboring graphs $G \sim G'$, let $A$ and $A'$ be their respective adjacency matrices. Then the original sensitivity is $\Delta_A = \sqrt{2}$, the squared Frobenius sensitivity satisfies $\Delta_2 = \norm{A'^2 - A^2}_F \le O(\sqrt{d})$, and the trace bound is $B_2 = \operatorname{tr}(A^2) \le nd$.
\end{lemma}

\begin{proof}
Let $A' = A + H$, where the update corresponding to a single edge between vertices $u$ and $v$ is $H = s(e_u e_v^\top + e_v e_u^\top)$ for $s \in \{-1, +1\}$.
Clearly, $\norm{H}_F = \sqrt{2} = O(1)$, establishing $\Delta_M = \sqrt{2}$. 
For the squared sensitivity, we have:
\[
    A'^2 - A^2 = (A+H)^2 - A^2 = AH + HA + H^2.
\]
Because $A$ is the adjacency matrix of a graph with maximum degree $d$, its local column $\ell_2$-norm is bounded by the degree: $\norm{A e_u}_2 = \sqrt{\sum_{i=1}^n A_{iu}^2} \le \sqrt{d}$. Consequently,
\[
    \norm{AH}_F = \norm{A(s e_u e_v^\top + s e_v e_u^\top)}_F \le \norm{A e_u}_2 \norm{e_v}_2 + \norm{A e_v}_2 \norm{e_u}_2 \le 2\sqrt{d}.
\]
By symmetry, $\norm{HA}_F \le 2\sqrt{d}$. Furthermore, $H^2 = e_u e_u^\top + e_v e_v^\top$, meaning $\norm{H^2}_F = \sqrt{2}$. Applying the triangle inequality, we bound the squared sensitivity as follows:
\[
    \Delta_2 \le \norm{AH}_F + \norm{HA}_F + \norm{H^2}_F \le 4\sqrt{d} + \sqrt{2} = O(\sqrt{d}).
\]
Finally, since the trace of $A^2$ is exactly the sum of vertex degrees:
\[
    B_2 = \operatorname{tr}(A^2) = \norm{A}_F^2 = \sum_{u=1}^n \operatorname{deg}(u) \le nd.
\]
\end{proof}

To enforce that our final output is a valid graph Laplacian, we define $\mathbb{L}_n$ as the convex set of all valid Laplacians of non-negative weighted graphs on $n$ vertices:
$$ \mathbb{L}_n = \left\{ X \in \mathbb{R}^{n \times n} \mid X = X^\top, X_{ij} \le 0 \text{ for } i \ne j, X \mathbf{1} = \mathbf{0} \right\}. $$
We combine the square amplifier with the SDP projection in Algorithm~\ref{alg:laplacian-square}.

\begin{algorithm}[h]
\caption{Private Graph Laplacian Oracle (Square Amplifier)}
\label{alg:laplacian-square}
\begin{algorithmic}[1]
\REQUIRE Graph $G = ([n], E)$ with maximum degree $d$, privacy parameters $\epsilon, \delta \in (0, 1)$.
\ENSURE The non-negative weighted graph Laplacian $\widehat{L} \in \mathbb{L}_n$.

\textbf{Adjacency Estimation:} 
\STATE Execute the Power-based Spectral Amplifier Framework (Algorithm~\ref{alg:spectral-amplifier}) on $A$ with $q=2$, privacy parameters $(\epsilon/2, \delta)$, sensitivity bounds $\Delta_A = \sqrt{2}$, $\Delta_2 = \Theta(\sqrt{d})$, and $B_2 = nd$. Let the output be $\widehat{A}$.

\textbf{Degree Estimation:} 
\STATE Compute the true degree vector $D \mathbf{1} = (d_1, \dots, d_n)^\top$.
\FOR{each $i \in [n]$}
    \STATE Sample independent Laplace noise $Z_i \sim \operatorname{Lap}(4/\epsilon)$.
    \STATE Compute noisy degree $\widehat{d}_i = d_i + Z_i$.
\ENDFOR
\STATE Let $\widehat{D} = \operatorname{diag}(\widehat{d}_1, \dots, \widehat{d}_n)$.

 \textbf{Intermediate Matrix:} 
\STATE Construct the intermediate pseudo-Laplacian $\widetilde{L} = \widehat{D} - \widehat{A}$.

 \textbf{SDP Projection:} 
\STATE Solve the following semidefinite program to find the optimal valid projection $\widehat{L}$:
\[
\min_{X \in \mathbb{L}_n, \, \gamma \in \mathbb{R}} \gamma \quad \text{subject to} \quad -\gamma I_n \preceq \widetilde{L} - X \preceq \gamma I_n
\]

\RETURN $\widehat{L}$.
\end{algorithmic}
\end{algorithm}

\begin{theorem}[Private Laplacian Estimation via $A^2$]\label{thm:laplacian-square}
For any $\epsilon, \delta \in (0, 1)$ and maximum degree $d \ge 1$, Algorithm~\ref{alg:laplacian-square} is $(\epsilon, \delta)$-differentially private. Moreover, for any input graph $G$, it outputs a non-negative weighted graph Laplacian $\widehat{L}$ such that with probability at least $1-1/n^c$:
$$ \norm{L_G - \widehat{L}}_2 \le \widetilde{O}_{\del}\left( \frac{(nd)^{1/4}}{\sqrt{\eps}} + \frac{1}{\epsilon}\right). $$
\end{theorem}

\begin{proof}
We first see the privacy guarantee. We partition the budget: In line 1, running the Square Amplifier requires $(\epsilon/2, \delta)$-DP by \Cref{thm:privacy}. In line 5, we estimate the degree sequence. Adding or removing a single edge alters two vertex degrees by at most $1$, so the global $\ell_1$-sensitivity of the degree vector is $2$. Applying the Laplace mechanism (Lemma~\ref{lem:laplace-priv}) with scale $2 / (\epsilon/2) = 4/\epsilon$ guarantees $(\epsilon/2, 0)$-DP. Therefore, the joint release of $\widehat{A}$ and $\widehat{D}$ satisfies $(\epsilon, \delta)$-DP by the basic composition theorem (Lemma~\ref{lem:composition}). The SDP projection in Step 9 operates exclusively on the privatized matrix $\widetilde{L}$ without accessing the underlying graph topology, consuming zero additional privacy budget due to post-processing (Lemma~\ref{lem:postprocess}). Thus, the algorithm is overall $(\epsilon, \delta)$-DP.

For the utility bound, we first compute the spectral distance of the intermediate matrix $\widetilde{L}$. By the triangle inequality:
\begin{equation}\label{eq:triangle_laplacian}
    \norm{\widetilde{L} - L_G}_2 = \norm{(\widehat{D} - \widehat{A}) - (D - A)}_2 \le \norm{\widehat{D} - D}_2 + \norm{A - \widehat{A}}_2. 
\end{equation}
Applying Theorem~\ref{thm:utility} (specifically the $q=2$ form in \cref{eq:square-amplifier-general}) and substituting the parameters from Lemma~\ref{lem:graph-square-sensitivity} ($\Delta_A = \sqrt{2}$, $\Delta_2 = \Theta(\sqrt{d})$, $B_2 = nd$) evaluated under budget $\epsilon/2$, the approximation error on adjacency matrix is:
\begin{align*}
    \norm{\widehat{A} - A}_2 &\le \widetilde O_{\del}\left( \sqrt{ \frac{\sqrt{d}\sqrt n}{\eps} } + \frac{1}{\eps} \sqrt{ \frac{nd\eps}{\sqrt{d}\sqrt n} } \right) 
    = \widetilde{O}_{\del}\left( \frac{(nd)^{1/4}}{\sqrt{\eps}} \right).
\end{align*}
For the degree matrix, the error is only contained on the diagonal: $\norm{\widehat{D} - D}_2 = \max_{i \in [n]} |Z_i|$. By the tail bounds of the Laplace distribution (Lemma~\ref{lem:laplace-tail}) and a union bound over $n$ variables, with high probability:
$$ \norm{\widehat{D} - D}_2 \le {O}\left(\frac{\log n}{\epsilon}\right). $$
Substituting both bounds into \cref{eq:triangle_laplacian}:
$$ \norm{\widetilde{L} - L_G}_2 \le 
\norm{\widehat{A} - A}_2+\norm{\widehat{D} - D}_2
= \widetilde{O}_{\del}\left( \frac{(nd)^{1/4}}{\sqrt{\epsilon}} + \frac{1}{\epsilon} \right).$$
Finally, we analyze the SDP projection. The constraint $-\gamma I_n \preceq \widetilde{L} - X \preceq \gamma I_n$ is identical to $\norm{\widetilde{L} - X}_2 \le \gamma$. The SDP computes the operator-norm projection onto $\mathbb{L}_n$:
$$ \widehat{L} = \arg\min_{X \in \mathbb{L}_n} \norm{\widetilde{L} - X}_2. $$
Because the true unweighted graph Laplacian inherently satisfies $L_G \in \mathbb{L}_n$, the optimality of $\widehat{L}$ dictates that $\norm{\widehat{L} - \widetilde{L}}_2 \le \norm{L_G - \widetilde{L}}_2$. Applying the triangle inequality yields the final spectral error:
\begin{align*}
    \norm{L_G - \widehat{L}}_2  &\le \norm{L_G - \widetilde{L}}_2 + \norm{\widetilde{L} - \widehat{L}}_2 \le 2 \norm{L_G - \widetilde{L}}_2
    \le \widetilde{O}_{\del}\left( \frac{(nd)^{1/4}}{\sqrt{\epsilon}} +\frac{1}{\epsilon}\right).
\end{align*}
The constant factor is absorbed by the $\widetilde{O}_{\del}$ notation, completing the proof.
\end{proof}

\subsection{A Bootstrapped Fourth-Power Spectral Primitive}
\label{sec:bootstrapped_fourth_power}

In this section, we refine the algorithm for private graph Laplacian estimation beyond the direct square amplifier, deriving an alternative spectral error bound. In particular, we show that there is a polynomial-time, $(\eps,\delta)$-DP algorithm that, for any bounded-degree graph, outputs a valid non-negative weighted synthetic graph whose Laplacian satisfies the error interface stated in \cref{eq:better_spectral_primitive}. Although this error bound does not universally beat $\widetilde{O}_{\eps,\del}((nd)^{1/4})$ across all parameter regimes, we will demonstrate in Section~\ref{sec:cut} that its dependence on $d$ and $q$ makes it useful for driving the recursive expander decomposition framework.

A direct application of the power-based amplifier framework (Algorithm~\ref{alg:spectral-amplifier}) with the operator $\Phi_4(A) = A^4$ to the original graph would not give satisfying bounds, as its Frobenius sensitivity $\|A^4 - (A')^4\|_F$ is governed by the graph's local two-walk column norm ($\max_{v\in V}\norm{A^2e_v}_2^2$), which can be as large as $O(d^3)$ in the worst case, leading to large privacy noise. To resolve this barrier, we introduce a \emph{bootstrapped} mechanism. Before applying the $A^4$ amplifier, we utilize an iterated cubic-walk ($A^3$) scoring step to privately identify and ``peel off'' a sparse set of edges that cause high column norms. This preprocessing step reduces the local two-walk column norms of the residual graph down to $\widetilde{O}_{\del}(d^2q/\eps)$ for some $q\leq d$. After that, we apply the $A^4$ amplifier via our framework to obtain a refined spectral approximation. Again, throughout this section, $G=([n],E)$ is a simple undirected graph with maximum degree at most $d$, adjacency matrix $A$, and Laplacian $L_G=\operatorname{diag}(A \mathbf{1})-A$. All $\widetilde O_{\delta}(\cdot)$ factors hide polylogarithmic factors in $n, 1/\delta$, and $\log d$, but \emph{not} polynomial factors in $1/\eps$.

\subsubsection{Local Column Norm for Two-Walk}

For a graph $G$ with adjacency matrix $A$, we define the maximum local energy:
\begin{equation}
\label{eq:kappa2-def}
  \kappa_2(G)
  :=
  \max_{v\in V}\norm{A^2e_v}_2^2
  =
  \max_{v\in V}\sum_{u\in V}\codeg_G(u,v)^2,
\end{equation}
where $\codeg_G(u,v)$ denotes the number of common neighbors between $u$ and $v$ in $G$. For every $d$-bounded graph and any fixed vertex $v$, we have $$\sum_u\codeg(u,v)^2 \le d\sum_u\codeg(u,v) \le d\sum_{x\in N(v)}\operatorname{deg}(x) \le d^3.$$ Thus, we inherently have $\kappa_2(G) \le d^3$. The quantity $\kappa_2$ is the critical obstacle to high-power private spectral release, as it controls the sensitivity of both cubic and fourth-power queries:

\begin{lemma}[Sensitivity of the cubic query under a local-energy promise]
\label{lem:cubic-sensitivity}
Let $A,A'$ be adjacency matrices of neighboring graphs of maximum degree at most $d$. Suppose both graphs satisfy $\kappa_2 \le K$ for a given threshold $K \ge d^2$. Then:
\[
  \fro{(A')^3-A^3}
  \le
  C\sqrt K.
\]
\end{lemma}

\begin{proof}
Let $A'=A+H$, where $H=\pm(e_ue_v^\top+e_ve_u^\top)$. We expand the difference:
\[
(A+H)^3-A^3
=A^2H+AHA+HA^2+AH^2+HAH+H^2A+H^3.
\]
The first and third terms satisfy
\[
  \fro{A^2H}, \fro{HA^2}
  \le
  2\max_v\norm{A^2e_v}_2
  \le 2\sqrt K.
\]
Moreover, since $\max_{u \in [n]} \norm{A e_u}_2 = \max_{u \in [n]} \sqrt{\sum_{v=1}^n A_{uv}^2} \le \sqrt{d}$, we have:
\[
  \fro{AHA}
  \le 2\norm{Ae_u}_2\norm{Ae_v}_2
  \le 2d.
\]
By our premise $K \ge d^2$, the remaining terms involve at least two factors of $H$ and have Frobenius norm $O(\sqrt d+1)$, which is dominated by $O(\sqrt K)$. Therefore, $\fro{(A')^3-A^3} \le C\sqrt K$.
\end{proof}

\begin{lemma}[Sensitivity of the fourth-power query under a local-energy promise]
\label{lem:fourth-sensitivity}
Let $A,A'$ be adjacency matrices of neighboring graphs of maximum degree at most $d$. Suppose both graphs satisfy $\kappa_2 \le K$ for a given threshold $K \ge d^2$. Then:
\[
  \fro{(A')^4-A^4}
  \le
  C d\sqrt K.
\]
\end{lemma}

\begin{proof}
Let $A'=A+H$, where $H=\pm(e_ue_v^\top+e_ve_u^\top)$. We expand $(A+H)^4-A^4$ binomially. We categorize the terms into those containing exactly one $H$ (one-$H$ terms) and those containing multiple (multi-$H$ terms).

The one-$H$ terms take the form $A^i H A^{3-i}$ for $i\in\{0,1,2,3\}$. Using the facts that $\norm{A^2e_v}_2\leq\sqrt{\kappa_2}\le \sqrt K$, $\norm{Ae_v}_2\le \sqrt{d}$, and $\norm{A}_2\leq d$, we bound them as follows:
\[
  \fro{A^2HA}
  \le
  \norm{A^2e_u}_2\norm{Ae_v}_2+\norm{A^2e_v}_2\norm{Ae_u}_2
  \le 2\sqrt{Kd},
\]
\[
  \fro{A^3H}
  \leq \norm{A^3e_u}_2\norm{e_v}_2+\norm{A^3e_v}_2\norm{e_u}_2
  \leq (\norm{A^2e_u}_2+\norm{A^2e_v}_2)\norm{A}_2\leq 2d\sqrt{K}.
\]
By symmetry, $\fro{HA^3}$ and $\fro{AHA^2}$ are similarly bounded by $2d\sqrt{K}$.
Furthermore, the Frobenius norm of all multi-$H$ terms is at most $O(d^2) \le O(d\sqrt K)$ because $K \ge d^2$. Summing these bounded terms completes the proof.
\end{proof}


\subsubsection{The Improved Private Spectral Primitive}

We now assemble the components to prove the following formal theorem corresponding to the error interface in \cref{eq:better_spectral_primitive}.

\begin{theorem}[Improved Bounded-Degree Private Spectral Release]
\label{thm:improved-spectral}
For any $\eps, \delta \in (0, 1)$, $n$-vertex input graph $G$ with maximum degree $d \ge 1$, and target parameter $q \ge \operatorname{polylog}(n,1/\del)/\eps$, there is a polynomial-time $(\eps,\del)$-DP algorithm that outputs a non-negative weighted graph $\widehat{G}$ with Laplacian $L_{\widehat{G}}$ such that with probability at least $1-1/n^c$:
$$    \norm{L_G-L_{\widehat{G}}}_2
  \le
  \widetilde O_{\delta}\left(
      \frac{d}{\sqrt{\eps q}}
      +
      \frac{n^{1/8}d^{1/2}q^{1/8}}{\eps^{3/8}}
      +
      \frac{n^{1/4}q^{1/4}}{\eps^{3/4}}
      +\frac{1}{\eps}
  \right).$$
\end{theorem}

\paragraph{Iterated Cubic-Score Cover.}
We use a cubic-walk strategy (Algorithm~\ref{alg:cubic-cover}) to decompose $G$ into a covered part $P$ and a residual graph $R$ with low local column norms for two-walks. $P$ is a private set of vertex pairs identified via noisy cubic-walk scores. The cover routine repeatedly releases these scores: if the current residual satisfies $\kappa_2\le K_i$, then $A^3$ has sensitivity $O(\sqrt{K_i})$. Thresholding the noisy $(A^3)_{uv}$ values successfully isolates the heavy pairs. In Algorithm~\ref{alg:cubic-cover}, we also use a propose-test-release subroutine to explicitly test whether $\kappa_2\le K_i$ to safely calibrate noise according to $K_i$.

\begin{algorithm}[h]
\caption{\textsc{IteratedCubicScoreCover}}
\label{alg:cubic-cover}
\begin{algorithmic}[1]
\REQUIRE Graph $G=(V,E)$, degree bound $d$, target parameter $q$, privacy parameters $(\eps,\del)$.
\ENSURE Covered pairs $P\subseteq\binom V2$ or $\bot$ (if safety test fails).
\STATE $R_0\gets G$, $P\gets\emptyset$, $K_0\gets 2d^3$.
\STATE Let $S\gets\lceil C\log\log d\rceil+2$.
\STATE Allocate per-round budgets: $\eps'\gets\eps/S$, $\del'\gets 2\del/(3S)$. For each round, split $\eps' = \eps_{test} + \eps_{rel}$ where $\eps_{test} = \eps_{rel} = \eps'/2$.
\FOR{$i=0,1,\ldots,S-1$}
    \IF{$K_i\le \widetilde{O}(d^2q/\eps_{rel})$}
        \STATE $R\gets R_i$ and \textbf{break}.
    \ENDIF
    
\textbf{Propose-Test-Release (PTR) Check}
    \STATE Compute true local energy $\kappa_{2,i} = \kappa_2(R_i)$, and let $\Delta_\kappa =3d^2$\;
    \STATE Let the safety margin be $M_i = \Delta_\kappa \frac{\log(1/\del')}{\eps_{test}}$.
    \STATE Sample Laplace noise $\xi_i \sim \operatorname{Lap}(\Delta_\kappa / \eps_{test})$.
    \IF{$\kappa_{2,i} + \xi_i > K_i - M_i$}
        \RETURN $\bot$
    \ENDIF
    
 \textbf{Private Release}
    \STATE Let $A_i$ be the adjacency matrix of $R_i$.
    \STATE Privately release $\widetilde B_i=A_i^3+Z_i$, where $Z_i\sim \mathcal{N}(0,\sigma^2)$ is Gaussian noise with $\sigma=\widetilde{O}_{\del}\left(\frac{\sqrt{K_i}}{\eps_{rel}}\right)$.
    \STATE Let $\zeta_i=\widetilde O_{\del}(\frac{\sqrt{K_i}}{\eps_{rel}})$ such that $\|Z_i\|_{\max}\le\zeta_i$ with high probability.
    \STATE Let $\tau_i=\max(2\zeta_i,dq/\eps_{rel})$.
    \STATE $P_i\gets\{uv:\widetilde B_i(u,v)\ge \tau_i\}$.
    \STATE $P\gets P\cup P_i$.
    \STATE $R_{i+1}\gets R_i\setminus(E(R_i)\cap P_i)$.
    \STATE $K_{i+1}\gets C_1d\tau_i$ (where $C_1 \ge 4$ is a sufficiently large constant).
\ENDFOR
\RETURN $P$.
\end{algorithmic}
\end{algorithm}

\begin{lemma}[The Iterated Cover]
\label{lem:iterated-cover-utility}
Assume $q \ge \operatorname{polylog}(d)/\eps$ and $d \ge 1$. \textsc{IteratedCubicScoreCover} is $(\eps,\del)$-DP for any input graph $G$. Furthermore, with probability at least $1-1/n^c$, it outputs $P\subseteq {n \choose 2}$ without aborting, satisfying $\Delta(P) \le \widetilde{O}_{\del}\left(\frac{\eps d^2}{q}\right)$ and $\kappa_2(R) \le \widetilde{O}_{\del}(\frac{d^2q}{\eps})$. Here $R=([n], E(G)\setminus (E(G)\cap P))$ is the remaining graph and $\Delta(P)$ is the maximum degree of the edge support $P$.
\end{lemma}

\begin{proof}
\textbf{Privacy Guarantee.} The algorithm runs for at most $S = O(\log \log d)$ iterations. In each round $i$, the algorithm splits the budget into $\eps_{test}$ and $\eps_{rel}$. 
First, we analyze the sensitivity of the true local energy $\kappa_2(G) = \max_v \sum_u \operatorname{codeg}(u,v)^2$. Adding or removing a single edge $e=\{x,y\}$ in a $d$-bounded graph changes the common neighborhood $\operatorname{codeg}(u,v)$ by at most 1, and this perturbation only occurs for pairs where $\{u,v\} \cap \{x,y\} \neq \emptyset$. Thus, for any fixed vertex $v$, the sum $\sum_u \operatorname{codeg}(u,v)^2$ changes by at most $\sum_{u \in N(v)} (2\operatorname{codeg}(u,v) + 1) \le 2d^2 + d \le 3d^2$. Consequently, $\kappa_2(G)$ has an exact global sensitivity bounded by $\Delta_\kappa \le 3d^2$. By the standard Laplace mechanism (Lemma~\ref{lem:laplace-priv}), releasing $\kappa_{2,i} + \xi_i$ satisfies $(\eps_{test}, 0)$-DP. 

Next, we analyze the overall privacy via the propose-test-release framework. Let $G,G'$ be a pair of neighboring graphs. By the principle of adaptive composition (Lemma~\ref{lem:composition}), we analyze the $i$-th round conditioned on an arbitrary fixed transcript of previous outputs in rounds $0 ,\cdots, i-1$. Because the set of removed edge pairs $P_{<i} = \bigcup_{j=0}^{i-1} P_j$ can be fully recovered from these public outputs, it is identical for both executions between inputs $G$ and $G'$, conditioned on any fix transcripts in the first $i-1$ round. Consequently, the residual graphs $R_i = G \setminus P_{<i}$ and $R'_i = G' \setminus P_{<i}$ differ by at most the single original edge separating $G$ and $G'$. For any such pair of residual graphs, we consider three different cases based on their true local energies $\kappa_{2,i}(R_i)$ and $\kappa_{2,i}(R'_i)$:
\begin{itemize}
    \item {Case 1 (Both Unsafe):} If both $\kappa_{2,i} > K_i$ and $\kappa'_{2,i} > K_i$, the probability that the algorithm passes the test for either graph is bounded by $\Prb[\xi_i \le -M_i] \le \frac{1}{2}\exp(-\eps_{test} M_i / \Delta_\kappa) = \del'/2$. Thus, the algorithm outputs $\bot$ with probability at least $1 - \del'/2$. For any set of outcomes $\mathcal{O}$, we have 
    \begin{align*}
        \Prb[\mathcal{M}(R_i) \in \mathcal{O}] &= \Prb[\mathcal{M}(R_i) \in \mathcal{O} \setminus \{\bot\}] + \Prb[\mathcal{M}(R_i) = \bot \land \bot \in \mathcal{O}]\\
        &\le \del'/2 + \del' + \Prb[\mathcal{M}(R_i') = \bot \land \bot \in \mathcal{O}] \\
        &\le e^{\eps'} \Prb[\mathcal{M}(R_i') \in \mathcal{O}] + 1.5\del'.
    \end{align*}
    satisfying $(\eps', 1.5\del')$-DP.
\item{Case 2 (Both Safe):} If both $\kappa_{2,i} \le K_i$ and $\kappa'_{2,i} \le K_i$, the true global sensitivity of the cubic query $A_i^3$ is deterministically bounded by $O(\sqrt{K_i})$ due to Lemma \ref{lem:cubic-sensitivity}. If the test fails, the algorithm outputs $\bot$. If it passes, the Gaussian mechanism is calibrated to the true sensitivity, guaranteeing $(\eps_{rel}, \del'/2)$-DP for the released matrix. Because the safety test itself provides $(\eps_{test}, 0)$-DP, standard basic composition dictates that the entire procedure (including the probability of outputting $\bot$) satisfies $(\eps_{test} + \eps_{rel}, \del'/2) \subset (\eps', \del')$-DP.
\item{Case 3 (Boundary):} Without loss of generality, assume $\kappa_{2,i} \le K_i$ (safe) but $\kappa'_{2,i} > K_i$ (unsafe). since they are neighbors, we have $\kappa'_{2,i} \le \kappa_{2,i} + \Delta_\kappa \le K_i + \Delta_\kappa$. The probability that the unsafe graph $R_i'$ passes the test is bounded by:
    \[
        \Prb[\text{pass} \mid R_i'] = \Prb[\xi'_i \le K_i - \kappa'_{2,i} - M_i]  \le \frac{1}{2}\exp\left(-\eps_{test} \frac{M_i - \Delta_\kappa}{\Delta_\kappa}\right) = \frac{\del'}{2} e^{\eps_{test}} \le \del',
    \]
    assuming $\eps_{test} \le \ln 2$. Crucially, the probability that the safe graph $R_i$ passes the test is identically bounded:
    \[
        \Prb[\text{pass} \mid R_i] = \Prb[\xi_i \le K_i - \kappa_{2,i} - M_i] < \Prb[\xi_i \le \Delta_\kappa - M_i] \le \del'.
    \]
    For any set of outcomes $\mathcal{O}$, let $\mathcal{O}_{\text{mat}} = \mathcal{O} \setminus \{\bot\}$, for the direction from $R_i$ to $R_i'$:
    \begin{align*}
        \Prb[\mathcal{M}(R_i) \in \mathcal{O}] 
        &= \Prb[\mathcal{M}(R_i) = \bot] \cdot \mathbbm{1}_{\bot \in \mathcal{O}} + \Prb[\mathcal{M}(R_i) \in \mathcal{O}_{\text{mat}}] \\
        &\le e^{\eps_{test}} \Prb[\mathcal{M}(R_i') = \bot] \cdot \mathbbm{1}_{\bot \in \mathcal{O}} + \Prb[\text{pass} \mid R_i] \\
        &\le e^{\eps_{test}} \Prb[\mathcal{M}(R_i') \in \mathcal{O}] + \del' \le e^{\eps'} \Prb[\mathcal{M}(R_i') \in \mathcal{O}] + \del',
    \end{align*}
    where the relation $\Prb[\mathcal{M}(R_i) = \bot] \le e^{\eps_{test}} \Prb[\mathcal{M}(R_i') = \bot]$ follows from the $(\eps_{test}, 0)$-DP property of Laplace mechanism (Lemma~\ref{lem:laplace-priv}). $\Prb[\mathcal{M}(R_i') \in \mathcal{O}] \le e^{\eps'} \Prb[\mathcal{M}(R_i) \in \mathcal{O}] + \del'$ holds by identical reasoning. This proves $(\eps', \del')$-DP on the boundary.
\end{itemize}
Therefore, the single round satisfies $(\eps', 1.5\del')$-DP for all possible inputs, without conditioning on high-probability utility events. Finally, basic composition (Lemma~\ref{lem:composition}) over $S$ rounds guarantees the overall $(\eps, \del)$-DP.


\textbf{Utility Guarantee.} 
For the utility guarantee, we condition on the high-probability event that for all rounds, the entry-wise maximum of the Gaussian noise matrix is bounded ($\|Z_i\|_{\max} \le \zeta_i$), and the safety test (lines 8-13) never aborts. By a union bound over the $O(\log \log d)$ rounds and all matrix entries, this joint event holds with probability at least $1-1/n^c$ for $\del = n^{-c'}$. 

Consider a single cover round on the residual graph $R_i$ with adjacency matrix $A_i$. If $uv \in P_i$, then $(A_i^3)_{uv} \ge \tau_i - \zeta_i \ge \tau_i/2$. For any vertex $u$, the row sum $\sum_v (A_i^3)_{uv}$ represents the total number of length-three walks originating from $u$, which is inherently at most $d^3$. Therefore, the maximum degree of the covered edges in this round is bounded by:
\begin{equation}
\label{eq:degree-Pi}
  \Delta(P_i)
  \le
  \max_{u}\frac{\sum_v(A_i^3)_{uv}}{\min_{v:uv\in P_i}(A_i^3)_{uv}}
  \le
  \frac{d^3}{\tau_i/2} = \frac{2d^3}{\tau_i}.
\end{equation}

Conversely, if an edge $uv$ remains in the residual graph $R_{i+1}$, it must be that $uv \notin P_i$, implying $(A_i^3)_{uv} < 2\tau_i$. Because $R_{i+1} \subseteq R_i$, the number of walks monotonically decreases, so $(A_{i+1}^3)_{uv} \le (A_i^3)_{uv} < 2\tau_i$. Thus, combing \cref{eq:local-walk-energy-and-cubic}, for any vertex $u$, the true local energy is bounded by:
\begin{equation}
\label{eq:energy-rec}
  \kappa_2(R_{i+1}) = \max_u \sum_{v \in N_{i+1}(u)} (A_{i+1}^3)_{uv} \le  2d\tau_i.
\end{equation}
Notice that the algorithm explicitly defines the public threshold $K_{i+1} = C_1 d \tau_i$ with a constant $C_1 \ge 4$. This guarantees a safety gap $K_{i+1} - \kappa_2(R_{i+1}) \ge \frac{1}{2} K_{i+1}$. Since the algorithm ensures $K_{i+1} \ge \widetilde{O}(d^2 q / \eps_{rel})$, and we assume $q \ge \operatorname{polylog}(d)/\eps$, this gap is asymptotically larger than the required Laplace margin $M_{i+1} = \widetilde{O}(d^2/\eps_{test})$, implying that the propose-test-release passes with high probability in the next round.

Now, we evaluate the recurrence. If the initial energy $K_0 = 2d^3 \le \widetilde{O}(d^2 q/\eps_{rel})$, the algorithm trivially satisfies the breaking condition at round $i=0$ and immediately outputs $P=\emptyset$, naturally satisfying all required bounds. Thus, we focus on the non-trivial regime where $K_0 > \widetilde{O}(d^2 q/\eps_{rel})$, which implies $$d \ge \widetilde{\Omega}(q/\eps_{rel}) \ge \Omega(1/\eps_{rel}^2),$$ meaning the base term $2d\eps_{rel}^2 \ge 1$. By the algorithmic definition, $\tau_i = \widetilde{O}_{\del}\left(\max(\sqrt{K_i}/\eps_{rel}, dq/\eps_{rel})\right)$. 
Notice that the $\max$ operation implies a phase transition at $K_i \approx d^2q^2$. In the noise-dominated regime where $K_i \ge d^2q^2$, substituting the threshold into $K_{i+1}\gets C_1d\tau_i$ gives the recurrence relation $K_{i+1} = \widetilde{O}_{\del}\left(d\sqrt{K_i}/\eps_{rel}\right)$. Ignoring polylogarithmic factors, solving this exact recurrence for the base case $K_0 = 2d^3$ yields 
$$K_i = \frac{d^2}{\eps_{rel}^2}(2d\eps_{rel}^2)^{2^{-i}}.$$
For any $d \ge 1$ and $\eps_{rel} > 0$, as $i \to \infty$, the exponential term $(2d\eps_{rel}^2)^{2^{-i}}$ rapidly approaches $1$. Thus, $K_i$ converges to the fixed point $d^2/\eps_{rel}^2$. 
Crucially, because we assume $q \ge \operatorname{polylog}(d)/\eps$, the phase-transition target $d^2q^2$ is larger than this fixed point $d^2/\eps_{rel}^2$. To ensure the sequence reaches the target $K_T \le d^2q^2$, it suffices to make sure that the residual factor satisfies:
$$(2d\eps_{rel}^2)^{2^{-T}} \leq (2d)^{2^{-T}} \leq 2 \leq \eps_{rel}^2 q^2,$$
which yields $T = O(\log\log d)$ since our premise $q \ge \operatorname{polylog}(d)/\eps$ naively guarantees $q \ge \sqrt{2}/\eps_{rel}$. Assume that in the $T$-th round, $K_T$ becomes smaller than $\widetilde{O}_{\del}(d^2q^2)$ for the first time, i.e., $K_T\le d^2q^2$ and $K_{T-1}>d^2q^2$. At this point, the Gaussian noise bound drops below the floor $\widetilde{O}_{\del}(dq/\eps_{rel})$, thus the threshold transitions to $\tau_T=\widetilde\Theta_{\del}(dq/\eps_{rel})$. Then in the $T+1$ round, 
\[
  \kappa_2(R_{T+1})
  \le d\tau_{T}
  =\widetilde{O}_{\del}(d^2q/\eps_{rel}).
\]
Therefore the algorithm successfully satisfies the breaking condition and stops after the $T+1$ round. Since $\tau_i\ge dq/\eps_{rel}$ for all $i\leq T+1$, substituting this into \cref{eq:degree-Pi} gives:
\[
  \Delta(P_{i})
  \le
  \frac{2d^3}{\tau_{i}}
  =
  \widetilde{O}_{\del}\left(\frac{\eps_{rel} d^2}{q}\right).
\]
Summing over $T+1=O(\log\log d)$ rounds only alters the polylogarithmic factor. Recalling that $\eps_{rel} = \Theta(\eps / \log\log d)$, we conclude:
\[
  \Delta(P)
  \le
  \widetilde{O}_{\del}\left(\frac{\eps d^2}{q}\right).
\]
This completes the proof.
\end{proof}

\paragraph{Releasing the Covered Part.}
\begin{proposition}[Gaussian Laplacian series]
\label{prop:gaussian-laplacian-series}
Let $P\subseteq \binom V2$ be an edge support of maximum degree at most $d$ on $n$ vertices, and let $\xi_e\sim \mathcal{N}(0,\sigma^2)$ be independent Gaussian random variables. Then, with probability at least $1-1/n^{c}$,
\[
  \opnorm{\sum_{e\in P}\xi_eL_e}
  \le
  \widetilde{O}_{\del}(\sigma\sqrt d).
\]
\end{proposition}

\begin{proof}
Let $g_e = \xi_e / \sigma$, so that $g_e \sim \mathcal{N}(0,1)$ are independent standard Gaussian random variables. The target random matrix can be identically rewritten as $Z = \sum_{e \in P} g_e (\sigma L_e)$. We first compute the matrix variance proxy $v$:
\[
  v = \opnorm{\sum_{e\in P} (\sigma L_e)^2} = \sigma^2 \opnorm{\sum_{e\in P} L_e^2}.
\]
For any single edge $e=\{u,v\}$, its unweighted Laplacian matrix $L_e$ satisfies the identity $L_e^2 = 2L_e$. Substituting this identity simplifies the variance proxy to:
\[
  v = 2\sigma^2 \opnorm{\sum_{e\in P} L_e} = 2\sigma^2 \opnorm{L_P},
\]
where $L_P = \sum_{e\in P} L_e$ is the Laplacian matrix of the subgraph induced by the edge set $P$. Since the edge support $P$ is promised to have a maximum degree of at most $d$, it is easy to verify that $\opnorm{L_P} \le 2d$. Therefore, the variance proxy is bounded by:
\[
  v \le 2\sigma^2 (2d) = 4\sigma^2 d.
\]
Applying Lemma~\ref{lem:matrix-gaussian-series}, we conclude that with probability at least $1-1/n^c$:
\[
  \opnorm{\sum_{e\in P}\xi_eL_e} \le \sqrt{2(4\sigma^2 d) \log(2n\cdot n^c)} = \widetilde{O}_{\del}(\sigma\sqrt d),
\]
which completes the proof.
\end{proof}

\begin{lemma}[Covered-Part Error]
\label{lem:covered-error}
Given the support $P$ from Algorithm~\ref{alg:cubic-cover}, \textsc{ReleaseOnSupport} is $(\eps,\del)$-DP. With probability at least $1-1/n^c$, its output satisfies:
\[
  \opnorm{\widehat L_P-L_{G[P]}}
  \le
  \widetilde{O}_{\del}\left(\frac d{\sqrt{\eps q}}\right).
\]
\end{lemma}
\begin{proof}
By the principle of adaptive composition (Lemma~\ref{lem:composition}), we analyze this subroutine conditioned on an arbitrary fixed transcript of the public support $P$ generated by Algorithm~\ref{alg:cubic-cover}. For a fixed \(P\), the query \(G\mapsto (\one[e\in E(G)])_{e\in P}\) has \(\ell_2\)-sensitivity at most one under edge-level privacy.  The privacy claim then follows from the Gaussian mechanism (Lemma~\ref{lem:gaussian-mechanism}).

Next, we establish the spectral utility bound. The true Laplacian of the graph $G$ restricted to the edge support $P$ is denoted as $L_{G[P]} = \sum_{e \in P} \one[e \in E] L_e$. By Algorithm~\ref{alg:release-support}, the spectral error matrix is:
\[
  \widehat L_P - L_{G[P]} = \sum_{e \in P} (\one[e \in E] + \xi_e) L_e - \sum_{e \in P} \one[e \in E] L_e = \sum_{e \in P} \xi_e L_e.
\]
Since the elements $\xi_e \sim \mathcal{N}(0, \sigma^2)$ are independent Gaussian random variables and the underlying edge support $P$ has a maximum vertex degree of at most $\Delta(P)$, we apply Proposition~\ref{prop:gaussian-laplacian-series} and with probability at least $1-1/n^c$, the operator norm of the error matrix satisfies:
\[
  \opnorm{\widehat L_P - L_{G[P]}} = \opnorm{\sum_{e \in P} \xi_e L_e} \le \widetilde O_{\del}\left(\sigma \sqrt{\Delta(P)}\right)\le \widetilde O_{\del}\left(\frac{\sqrt{\Delta(P)}}{\eps}\right).
\]
Combining the upper bound of $\Delta(P)$ in Lemma~\ref{lem:iterated-cover-utility} completes the proof.
\end{proof}

\begin{algorithm}[h]
\caption{\textsc{ReleaseOnSupport}}
\label{alg:release-support}
\begin{algorithmic}[1]
\REQUIRE Graph $G=(V,E)$, pair support $P\subseteq\binom V2$, privacy parameters $(\eps,\del)$.
\ENSURE Noisy Laplacian $\widehat L_P$ approximating $L_{G[P]}$.
\STATE Set the Gaussian noise scale $\sigma = \Theta\left(\frac{\sqrt{\log(1/\del)}}{\eps}\right)$.
\FOR{\textbf{each} $e \in P$}
    \STATE Set the indicator weight $w_e = \one[e \in E]$.
    \STATE Sample independent Gaussian noise $\xi_e \sim \mathcal{N}(0,\sigma^2)$.
    \STATE Compute the noisy weight $\widehat w_e = w_e + \xi_e$.
\ENDFOR
\RETURN $\widehat L_P = \sum_{e\in P} \widehat w_e L_e$.
\end{algorithmic}
\end{algorithm}

\paragraph{Thresholded Fourth-Power Release of the Residual.}
With high probability, the residual graph $R$ satisfies $\kappa_2(R)\le K:=\widetilde{O}_{\del}(d^2q/\eps)$ according to Lemma~\ref{lem:iterated-cover-utility}. We directly invoke our Power-based Spectral Amplifier (Algorithm~\ref{alg:spectral-amplifier}) with $q=4$.

\begin{lemma}[Fourth-Power Residual Release]
\label{lem:thresholded_fourth_power_residual}
Let $R$ be an $n$-vertex residual graph with maximum degree $d$ and $\kappa_2(R)\le K$. Instantiating Algorithm~\ref{alg:spectral-amplifier} with $\Phi_4(A_R) = A_R^4$ provides an $(\eps,\delta)$-DP raw Laplacian estimate $\widetilde L_R$ satisfying, with high probability:
\[
  \norm{L_R-\widetilde L_R}_2
  \le
  \widetilde O_{\delta}\left(
    \frac{(nK)^{1/8}d^{1/4}}{\eps^{1/4}}
    +\frac{(nK)^{1/4}}{\sqrt{d\eps}}
    +\frac{1}{\eps}\right).
\]
Substituting $K \le \widetilde O_{\del}(d^2q/\eps)$, this simplifies to:
\[
  \norm{L_R-\widetilde L_R}_2
  \le
  \widetilde O_{\delta}\left(
      \frac{n^{1/8}d^{1/2}q^{1/8}}{\eps^{3/8}}
      +
      \frac{n^{1/4}q^{1/4}}{\eps^{3/4}}
      +\frac{1}{\eps}
  \right).
\]
\end{lemma}

\begin{proof}
By Lemma~\ref{lem:fourth-sensitivity}, the amplified sensitivity is bounded by $\Delta_4 = O(d\sqrt K)$. The trace bound satisfies $B_4=\operatorname{tr}(A_R^4) = \sum_{v\in V}\norm{A_R^2e_v}_2^2 \le nK$. 
Applying the general bound for power amplifiers (\cref{eq:power-amplifier-general}) evaluated at $q=4$ directly yields:
$$ \norm{A_R-\widehat A_R}_2 \le \widetilde O_{\delta}\left( \left(\frac{d\sqrt{K}\sqrt n}{\eps}\right)^{1/4} + \frac{1}{\eps} \sqrt{ \frac{nK\eps}{d\sqrt{K}\sqrt n} } \right) = \widetilde O_{\delta}\left( \frac{n^{1/8}K^{1/8}d^{1/4}}{\eps^{1/4}} +\frac{n^{1/4}K^{1/4}}{\sqrt{d\eps}} \right). $$
Finally, as in the analysis of Theorem~\ref{thm:laplacian-square}, estimating the residual degree vector using the Laplace mechanism with $\ell_1$-sensitivity $O(1)$ independently contributes $\widetilde O_{\delta}(1/\eps)$ to the diagonal. Summing these errors yields the stated bound.
\end{proof}

\paragraph{Assembling the Final Oracle.}
\begin{algorithm}[h]
\caption{Bootstrapped Fourth-Power Private Laplacian Oracle}
\label{alg:bootstrapped_fourth_power_primitive}
\begin{algorithmic}[1]
\REQUIRE Graph $G=([n],E)$ with maximum degree at most $d$, tunable parameter $q$, privacy parameters $(\eps,\del)$, failure probability $\beta = n^{-c}$.
\ENSURE A valid non-negative weighted synthetic graph Laplacian $\widehat L_G \in \mathbb{L}_n$, or $\bot$.

\STATE Split budgets evenly: $\eps_{cov} = \eps/3, \eps_{sup} = \eps/3$, and $\eps_{amp} = \eps/3$. Split $\del$ and $\beta$ similarly.
\STATE $P\gets\textsc{IteratedCubicScoreCover}(G,d,q,\eps_{cov},\del/3,\beta/3)$\;
\IF{$P \text{ is } \bot$}
    \RETURN $\bot$\; 
\ENDIF
\STATE Let $R= ([n], E\setminus (E\cap P))$\;
\STATE $\widetilde L_P\gets\textsc{ReleaseOnSupport}(G,P,\eps_{sup},\del/3,\beta/3)$\;

\STATE Let $K_{\text{final}} = \widetilde{O}(d^2q/\eps)$ be the promised target energy. 
\STATE Split $\eps_{amp}$ into $\eps_{test} = \eps_{amp}/2$ and $\eps_{rel} = \eps_{amp}/2$.
\STATE Let the safety margin be $M_{\text{final}} = 3d^2 \frac{\log(6/\del)}{\eps_{test}}$.
\STATE Sample $\xi_{\text{final}} \sim \operatorname{Lap}(3d^2 / \eps_{test})$.
\IF{$\kappa_2(R) + \xi_{\text{final}} > K_{\text{final}} - M_{\text{final}}$}
    \RETURN $\bot$\; 
\ENDIF

\STATE Run the power-based spectral amplifier (Algorithm~\ref{alg:spectral-amplifier}) on $R$ with $q=4$ and budgets $(\eps_{rel}/2,\del/6,\beta/6)$ plus a degree estimate with budgets $(\eps_{rel}/2,\beta/6)$ to obtain $\widetilde L_R$.
\STATE Form the raw Laplacian estimate: $\widetilde L \gets \widetilde L_P+\widetilde L_R$.
\STATE \textbf{SDP Projection:} Compute the projection $X$ onto the valid Laplacian cone $\mathbb{L}_n$.
\RETURN $\widehat L_G = X$.
\end{algorithmic}
\end{algorithm}
Finally, to guarantee a valid graph Laplacian on non-negative edge graphs, in Algorithm~\ref{alg:bootstrapped_fourth_power_primitive} we integrate the SDP projection exactly as in Section \ref{sec:laplacian-square}.

\begin{proof}[Proof of Theorem~\ref{thm:improved-spectral}]
The privacy of the algorithm follows universally from basic composition and the propose-test-release framework. Crucially, the algorithm employs explicit safety checks (lines 3-5 and lines 9-13) to ensure that the sensitivity bound is satisfied throughout the subsequent subroutines. 
Note that in the the testing phase (lines 9-13), the only event that consumes the approximate DP budget $\delta$ is when a residual graph $R$ with true local energy $\kappa_2(R) > K_{\text{final}}$ accidentally passes the test. By setting the safety margin to $M_{\text{final}} = 3d^2 \frac{\log(6/\del)}{\eps_{test}}$ and applying the Laplace tail bound (Lemma~\ref{lem:laplace-tail}), the probability of such an unsafe graph passing the test is bounded by $\del/6$. 

Therefore, with all but $\del/6$ probability, any graph that proceeds to line 14 satisfies the 4th-power sensitivity bound $O(d\sqrt{K_{\text{final}}})$, guaranteeing that the subsequent spectral amplifier (Algorithm~\ref{alg:spectral-amplifier}) provides $(\eps_{amp}, \del/6)$-DP. Combining the budgets used in the cover (Algorithm~\ref{alg:cubic-cover}), the support release (Algorithm~\ref{alg:release-support}), the PTR test, and the final amplifier via basic composition (Lemma~\ref{lem:composition}) yields the overall $(\eps, \del)$-DP guarantee.

For the utility, we analyze by conditioning on the high-probability event that the algorithm does not abort (i.e., it does not output $\bot$) and all noise tail bounds hold. 

By Lemma~\ref{lem:iterated-cover-utility}, the \textsc{IteratedCubicScoreCover} successfully outputs $(P, R)$ satisfying $\kappa_2(R) \le \widetilde{O}(d^2q/\eps_{cov})$ with probability at least $1-\beta/3$, which ensures that $\kappa_2(R) \le K_{\text{final}}/2$ for the publicly promised target energy $K_{\text{final}} = \widetilde{O}(d^2q/\eps)$ by choosing large enough constants for $K_{\text{final}}$. Because $q \ge \operatorname{polylog}(n)$, the gap $K_{\text{final}} - \kappa_2(R) \ge K_{\text{final}}/2$ asymptotically dominates the safety margin $M_{\text{final}} = \widetilde{O}(d^2/\eps_{test})$. Therefore, to trigger a false abort, the Laplace noise $\xi_{\text{final}}$ must exceed $K_{\text{final}}/2 - M_{\text{final}} \gg \Delta_\kappa \frac{\ln(3/\beta)}{\eps_{test}}$. By the standard tail bound of the Laplace distribution (Lemma~\ref{lem:laplace-tail}), this happens with probability at most $\beta/3$.

Conditioned on this successful execution, the true graph Laplacian partitions exactly as $L_G = L_{G[P]} + L_R$. Applying the triangle inequality to our raw estimate $\widetilde L = \widetilde L_P + \widetilde L_R$, we sum the spectral error bounds from the support release (Lemma~\ref{lem:covered-error}) and the fourth-power residual release (Lemma~\ref{lem:thresholded_fourth_power_residual}). By a union bound over all subroutines, this raw estimation error bounds $\norm{L_G-\widetilde L}_2$ with overall probability at least $1-\beta$. 

Finally, because the true unweighted graph Laplacian inherently satisfies $L_G \in \mathbb{L}_n$, the optimality of the SDP projection guarantees that $\norm{\widehat L_G-\widetilde L}_2 \le \norm{L_G-\widetilde L}_2$. Applying the triangle inequality one last time yields $$\norm{L_G-\widehat L_G}_2 \le \norm{L_G-\widetilde L}_2 + \norm{\widetilde L-\widehat L_G}_2 \le 2\norm{L_G-\widetilde L}_2, $$ which gives the claimed error interface and completes the proof.
\end{proof}

\section{An Edge-sensitive Private Cut Oracle}\label{sec:edge-sensitive-cut-oracle}
Here we construct a new oracle for the private release of graphs that preserves all-cut values. This edge-sensitive cut oracle, together with the spectral approximation results in Section~\ref{sec:laplacain_analysis} serve as key components of our final worst-case cut approximation result in Section~\ref{sec:cut}. However, the development of this oracle is self-contained and independent of previous spectral primitives developed in Section~\ref{sec:laplacain_analysis}. As such, it may be of independent interest. Throughout this section, we denote $\cut_{G}(\cdot) \equiv w_G(\cdot)$ as the cut size function.

\begin{theorem}[Cut Oracle]\label{thm:cut-oracle-main}
Let $G=(V,E)$ be a simple unweighted graph on $n$ vertices, and let $M\geq |E|$ be public.  For $\eps\in (0,1), \delta\in(0,1/2), \gamma\in(0, 1/4)$, 
there is a polynomial-time $(\eps,\delta)$-edge-DP algorithm $\mathsf{CutOracle}$ that outputs a weighted graph $\widehat G$ such that, with probability at least $1-1/n^c$, for all $S\subseteq V$,
\[
  \left|\cut_{\widehat G}(S)-\cut_G(S)\right|
  \leq \gamma\,\cut_G(S)+\alpha,
\]
where
\[
  \alpha
  =\widetilde{O}_{\del}\left(
      \frac{n}{\eps}
      +\left(\frac{n^2M}{\eps^2\gamma}\right)^{1/3}
    \right).
\]
Furthermore, the output has total edge weight at most $M$.
\end{theorem}
To get the error bound $\widetilde{O}_{\del}({n}/{\eps} + \left({n^2M}/{(\eps^2\gamma)}\right)^{1/3})$, we utilize the private mirror descent framework. The challenge in applying this method to graph cuts is the need for a separation oracle: exactly identifying the worst-case cut violation among $2^n$ subsets is both computationally intractable and sensitive. Inspired by Eli{\'a}{\v{s}} et al.~\cite{eliavs2020differentially}, we resolve this by relaxing the combinatorial cut checking problem into a continuous convex optimization problem, which is equivalently a log-determinant regularized SDP. This carefully designed objective encodes all cut violations while enforcing stability on its optimal solution matrix to adopt privacy mechanisms. Further, compared to the direct application of the standard mirror descent framework \cite{eliavs2020differentially,peng2025differentially} which yields an additive error bound of $\widetilde{O}_{\eps,\del}(\sqrt{Mn})$, to obtain a better dependency on $M$ we privately test the objective value at each iterate and use the potential function method to analyze the utility.

\begin{remark}[Remark of Theorem~\ref{thm:cut-oracle-main}]
The primary motivation for designing this oracle is to replace the best-known (before us) edge-sensitive error bound of $\widetilde O(\sqrt{Mn})$ due to Eli{\'a}{\v{s}} et al.~\cite{eliavs2020differentially}. This oracle will serve as a black-box interface in Section~\ref{sec:cut}, where it achieves its full potential: after using the spectral primitive (Theorem~\ref{thm:improved-spectral}) from Section~\ref{sec:bootstrapped_fourth_power} to trigger a sequence of recursive private expander decomposition, for arbitrary dense input graphs, we are able to reduce the number of edges in the residual graph to $n^{5/4 + o(1)}$. Applying our oracle to this residual graph then yields a worst-case bound of $\widetilde{O}((n^2n^{5/4+o(1)})^{1/3}) = \widetilde{O}(n^{13/12 + o(1)})$, which improves upon the previous $\widetilde{O}(n^{1.25})$ bound.
\end{remark}

\subsection{The Regularized Separation Objective and Cut Approximation}

Let $G=(V, E)$ be a simple unweighted graph on $n$ vertices with edge-weight vector $w^* \in \{0, 1\}^{\binom{n}{2}}$. We assume a publicly known upper bound $M \geq \|w^*\|_1$. The goal is to privately release a weighted graph $\widehat{G}$ with edge-weight vector $w \in \mathbb{R}_{\geq 0}^{\binom{n}{2}}$ that preserves all cut values of $G$. To maintain a bounded solution space, we define the feasible region of the  edge-weight vector as $\mathcal{W} := \{w \in \mathbb{R}_{\geq 0}^{\binom{n}{2}} : \|w\|_1 \leq M\}$. For any $w \in \mathcal{W}$, we introduce a \emph{slack variable} $s(w) := M - \|w\|_1$, which captures the difference of total weight sum
. 
For any cut $S \subseteq V$, the Boolean cut query $c_S\in\{0,1\}^{\cX}$ is defined as:
\[
 c_S(\{u,v\})=\one[\{u,v\}\in(S,V\setminus S)],
\] hence the cut value of $S$ on $w$ equals to $\cut_w(S) = \langle w, c_S \rangle$.

Given a target multiplicative approximation parameter $\gamma \in (0, 1)$, we define the \textbf{signed} edge residual vectors between $w$ and $w^*$ as $a_\sigma(w, w^*) \in \mathbb{R}^{\binom{n}{2}}$ and their corresponding slack residuals $a_\sigma^{\text{slack}}(w, w^*) \in \mathbb{R}$ for directions $\sigma \in \{+, -\}$:
\begin{align}
a_+(w, w^*) &= w - (1+\gamma)w^*, \quad &a_+^{\text{slack}}(w, w^*) = s(w) - (1+\gamma)s(w^*), \\
a_-(w, w^*) &= (1-\gamma)w^* - w, \quad &a_-^{\text{slack}}(w, w^*) = (1-\gamma)s(w^*) - s(w).
\end{align}
We also write them in abbreviated form as $a_\sigma, a_\sigma^{\text{slack}}$, if $w,w^*$ are clear in the context. Note that the sum of the residuals is an invariant: $\sum_{e \in \binom{n}{2}} a_\sigma(e) + a_\sigma^{\text{slack}} = -\gamma M$ for $\sigma\in \{+1,-1\}$. To construct a computationally efficient and private separation oracle, we relax the problem of checking the all $2^n$ cuts into a continuous convex optimization over a positive semi-definite matrix $X \in \mathbb{R}^{n \times n}$ and a scalar $t \in \mathbb{R}$. For a fixed parameter $0 < \rho \leq 1/8$, we define the regularized cut separation objective function $F_\sigma(w, w^*)$ as:
\begin{equation}
    \label{eq:hat_objective_f_cut_oracle}
    \widehat{F}_\sigma(w, w^*,X,t):=\sum_{e=\{u,v\} \in \binom{n}{2}} a_\sigma(w,w^*)[e]\cdot\frac{1 - X_{uv}}{2} + a_\sigma^{\text{slack}}(w,w^*) \cdot t,
\end{equation}
\begin{equation}
    \label{eq:objective_f_cut_oracle}
    F_\sigma(w, w^*) := \max_{X\in\mathcal{X}, \rho \leq t\leq 1 } \left\{ \widehat{F}_\sigma(w, w^*,X,t) + \lambda \mathcal{R}(X, t) \right\},
\end{equation}
where the domain of $X$ is 
\[\mathcal{X}=\left\{X\in\mathbb{R}^{n\times n}|X = X^\top, \diag(X) = \mathbf{1}_n, X \succeq \rho I_n\right\},\]
and the regularizer $\mathcal{R}(X, t)$ controls the stability of the optimum:
\[\mathcal{R}(X, t) := \log\det(X) + \log t.\]
The parameter $\lambda$ is used to control the privacy. We use the private mirror descent method introduced by \cite{eliavs2020differentially} to solve the optimization
\begin{equation}
    \min_{w\in\mathcal{W}}\max_{\sigma\in\{+,-\}}F_\sigma(w, w^*).
\end{equation}
Before we introduce our algorithm, we first show that solving this optimization will yield a graph approximation that preserves all cut values. For every cut $S$ and every $0<\rho\leq1/8$, let $\beta^S\in\{\pm1\}^n$ be the sign vector of $S$, with $\beta^S_u=1$ on $S$ and $-1$ outside $S$. We define the matrix \[X_S^{\rho}=(1-2\rho)\beta^S(\beta^S)^\top+2\rho I_n.\]
\begin{lemma}[Intuition Behind Objective Function]\label{lem:intuition_obj_f}
    For every cut $S$ and every $0<\rho\leq1/8$, we have $X_S^{\rho}\in\mathcal{X}$ and 
    \[
    \widehat{F}_\sigma(w, w^*,X_S^\rho,\rho)
 =(1-2\rho)\ip{a_\sigma(w,w^*)}{c_S}-\rho\gamma M.\]
In particular, we have
    \[
    \widehat{F}_+(w, w^*,X_S^\rho,\rho)
    =(1-2\rho)\left(\cut_w(S)-(1+\gamma)\cut_{w^*}(S)\right)-\rho\gamma M,\]
    and
    \[
    \widehat{F}_-(w, w^*,X_S^\rho,\rho)
    =(1-2\rho)\left((1-\gamma)\cut_{w^*}(S)-\cut_w(S)\right)-\rho\gamma M.\]
\end{lemma}
\begin{proof}
The diagonal of $X_S^\rho$ is one.  Its eigenvalues are $2\rho$ on $(\beta^S)^\perp$ subspace and $n(1-2\rho)+2\rho$ in the $\beta^S$ direction, so $X_S^\rho\succeq\rho I$. Hence $X_S^\rho\in\mathcal{X}$. For any $e=\{u,v\}$, write $b_e = e_u - e_v$, we have
\[
 \frac{1-(X_S^\rho)_{uv}}{2}=\frac14b_e^\top X_S^\rho b_e
 =\frac{1-2\rho}{4}(\beta^S_u-\beta^S_v)^2+\frac{2\rho}{4}\norm{b_e}_2^2
 =(1-2\rho)c_S(e)+\rho.
\]
Combining this with $\sum_{e \in \binom{n}{2}} a_\sigma(e) + a_\sigma^{\text{slack}} = -\gamma M$ gives
\[
 \widehat{F}_\sigma(w, w^*,X_S^\rho,\rho)=\sum_{e=\{u,v\} \in \binom{n}{2}} a_\sigma(w,w^*)_e\cdot\frac{1 - (X_S^\rho)_{uv}}{2}+a_\sigma^{\text{slack}}\cdot \rho
 =(1-2\rho)\ip{a_\sigma}{c_S}-\rho\gamma M.
\]
Using $\cut_w(S) = \langle w, c_S \rangle$ completes the proof.
\end{proof}

We use the following lemma to bound the error introduced by the regularizer term: 
\begin{lemma}\label{lem:regularizer-bound}
For a given $\rho\in(0,1/8)$ and any $X\in\mathcal{X},\rho \leq t\leq 1$, we have,
\[
 \mathcal R(X,t)\leq0.
\]
Furthermore, for any cut $S$, we have,
\[\mathcal R(X_S^\rho,\rho)\geq-C_0 n\log(1/\rho)\]
for a universal constant $C_0$.
\end{lemma}
\begin{proof}
Hadamard's inequality (Lemma~\ref{lem:hadamard}) gives $\det X\leq1$ because $\diag(X)=\mathbf{1}_n$. Plus that $t\leq 1$ gives $\mathcal \mathcal R(X,\rho)\leq 0$. Since the eigenvalues of $X_S^\rho$ are $2\rho$ on $(\beta^S)^\perp$ subspace and $n(1-2\rho)+2\rho$ in the $\beta^S$ direction, we have
\[
 \log\det X_S^\rho
 =(n-1)\log(2\rho)
  +\log\bigl(n(1-2\rho)+2\rho\bigr).
\]
Therefore, by $\rho\leq1/8$,
\[R(X_S^\rho,\rho)=(n-1)\log(2\rho)
  +\log\bigl(n(1-2\rho)+2\rho\bigr)+\log\rho\geq-C_0 n\log(1/\rho)\]
for a universal constant $C_0$.
\end{proof}

The next lemma ensures that solving the optimization problem (\cref{eq:objective_f_cut_oracle}) sufficiently well yields a graph approximation that preserves all cut values with multiplicative error.

\begin{lemma}\label{lem:opt implies cut approx}
Let $\alpha > 0$. Assume $\rho = \min(\frac{1}{8},\frac{\alpha}{16\gamma M})$ and $\lambda$ satisfies $C_0 \lambda n \log(1/\rho) \leq \frac{\alpha}{32}$ for the universal constant $C_0$ decided in Lemma~\ref{lem:regularizer-bound}, then if
\[\max_{\sigma \in \{+,-\}} F_\sigma(w, w^*) \leq \frac{5}{8}\alpha,\]
then
\[|\cut_w(S) - \cut_{w^*}(S)| \leq \gamma \cut_{w^*}(S) + \alpha\]
holds for all $S \subseteq V$.
\end{lemma}

\begin{proof}
We prove the contrapositive. Assume there exists a cut $S$ such that
\[|\cut_w(S) - \cut_{w^*}(S)| > \gamma \cut_{w^*}(S) + \alpha.\]
Using $\cut_w(S) = \langle w, c_S \rangle$, this implies that there exists a sign $\sigma \in \{+, -\}$ such that $\langle a_\sigma, c_S \rangle > \alpha$. By Lemma~\ref{lem:intuition_obj_f} and $\rho=\min(\frac{1}{8},\frac{\alpha}{16\gamma M})$, we have 
\[
\widehat{F}_\sigma(w, w^*,X,t)
=(1-2\rho) \langle a_\sigma, c_S \rangle - \rho \gamma M\geq\frac{3\alpha}{4}-\frac{\alpha}{16}\geq\frac{11}{16}\alpha.\]
Furthermore, by Lemma~\ref{lem:regularizer-bound} and the assumption $C_0 \lambda n \log(1/\rho) \leq \frac{\alpha}{32}$, we have:
\[\lambda\mathcal{R}(X_S^\rho, \rho)
\geq-\lambda C_0 n\log(1/\rho)
\geq-\frac{\alpha}{32}.\]

Because $F_\sigma(w, w^*)$ maximizes over all $X\in\mathcal{X}$ and $\rho\leq t\leq1$, evaluating it at $(X_S^\rho, \rho)$ provides a lower bound of $F_\sigma(w, w^*)$. Therefore,
\begin{align*}
    F_\sigma(w, w^*)&\geq
    \widehat{F}_\sigma(w, w^*,X_S^\rho,\rho)
    +\lambda\mathcal{R}(X_S^\rho, \rho) \\
    &\geq \frac{11\alpha}{16}-\frac{\alpha}{32} > \frac{5}{8}\alpha.
\end{align*}
This contradicts with the assumption.
\end{proof}

\subsection{The Private Mirror Descent Algorithm}
In this section, we use the Private Mirror Descent framework by \cite{eliavs2020differentially} to solve the optimization problem defined in (\cref{eq:objective_f_cut_oracle}). We use the mirror function $\psi(w): \mathcal{W} \to \mathbb{R}$ defined as below:
\[\psi(w) := \sum_{e \in \binom{n}{2}} w_e \log w_e + s(w) \log s(w)\]
Generally, in each iteration, we solve the SDP defined in \cref{eq:objective_f_cut_oracle}, privately release the clipped gradient of the objective function, and finally perform a mirror descent update step with respect to the mirror function $\psi$. The details are in Algorithm~\ref{alg:cut_oracle}.

\begin{algorithm}[h]
\caption{Private Edge-sensitive Cut Oracle via Mirror Descent}
\label{alg:cut_oracle}
\begin{algorithmic}[1]
\REQUIRE Graph $G = ([n],E)$ with public edge bound $|E|\leq M$, parameters $\gamma,\epsilon, \delta \in (0, 1)$.
\ENSURE An edge-weight vector $w \in \mathcal{W}$.

\STATE Set parameters: target error $\alpha=\widetilde{O}_{\del}(
      \frac{n}{\eps}
      +\left(\frac{n^2M}{\eps^2\gamma}\right)^{1/3})$, threshold $B=\Theta(\log\frac{M}{\alpha})$, time limit $T=\Theta(\frac{BM\log n}{\gamma\alpha})$, $\lambda=\nu=\Theta(\frac{\sqrt{T}\log^{3/2}(T/\del)}{\eps})$, $\rho = \min(\frac{1}{8},\frac{\alpha}{16\gamma M})$, learning rate $\eta = \Theta(\frac{\gamma}{B})$.
\STATE Choose the initial solution $w_1(e) = \frac{M}{\binom{n}{2}+1}$ for all $e \in \binom{n}{2}$ and $s(w_1) = \frac{M}{\binom{n}{2}+1}$.
\FOR{$t = 1,2, \dots, T$}
    \item[]\textbf{Solving the SDP:}
    \STATE For both directions $\sigma \in \{+, -\}$, find the maximizer $X_{t,\sigma} \in \mathcal{X}$ and $t_{t,\sigma} \in [\rho, 1]$ of $F_\sigma(w_t, w^*)$, where $F_{\sigma}$ is defined in \cref{eq:objective_f_cut_oracle}. 
    
    \textbf{Test and Release:}
    \STATE Release private estimates $\widetilde{F}_{t,\sigma} = F_{t,\sigma} + Z_{t,\sigma}$, where $Z_{t,\sigma} \sim \mathcal{N}(0, \nu^2)$.
    \IF{$\max(\widetilde{F}_{t,+}, \widetilde{F}_{t,-}) \leq 19\alpha/32$}
        \RETURN $w_t$.
    \ENDIF
    
    \textbf{Release the Private Clipped Gradient:}
    \STATE Let $\sigma_t=\argmax_{\sigma\in\{+,-\}}\widetilde{F}_{t,\sigma}$.
    \STATE Choose a random vector $z \sim \mathcal{N}(0, X_{t,\sigma_t})$ and $z_0 \sim \mathcal{N}(0, t_{t,\sigma_t})$.
    \STATE Let $\widehat g_t(\{u,v\}) = \frac{1}{4}(z_u - z_v)^2, \forall \{u,v\} \in \binom{n}{2}$ and $\widehat g_t^{\text{slack}}=z_0^2$.
    \STATE Let $g_t(\{u,v\}) = \min\left(B, \widehat g_t(\{u,v\})\right), \forall \{u,v\} \in \binom{n}{2}$ and $g_t^{\text{slack}}=\min(B,\widehat g_t^{\text{slack}})$.
    
    \textbf{Mirror Descent Step:}
    \STATE Let $\theta_t = 1$ if $\sigma_t = +$, and $\theta_t = -1$ if $\sigma_t = -$.
    \STATE For each edge $e \in \binom{n}{2}$:
    \begin{equation}
    \label{eq:mirror descent update}
        w_{t+1}(e) = M \cdot \frac{w_t(e) \exp(-\theta_t\eta g_t(e))}{s(w_t) \exp(-\theta_t\eta g_t^{\text{slack}}) + \sum_{e' \in \binom{n}{2}} w_t(e') \exp(-\theta_t\eta g_t(e'))}.
    \end{equation}
    \STATE Update the slack variable: $s(w_{t+1})=M-\norm{w_{t+1}}_1$.
\ENDFOR
\RETURN $\perp$.
\end{algorithmic}
\end{algorithm}

\begin{remark}[Remark of Algorithm~\ref{alg:cut_oracle}]

Lemma~\ref{lem:gradient_unbias} shows that the unclipped gradient estimator $\widehat g_t$ is an unbiased approximation of the real gradient of $F_{\sigma}(w_t,w^*)$. However, we cannot directly use the unclipped gradient estimator, 
as the proof of Lemma~\ref{lem:potential progress} requires a constant upper bound on the gradient to establish the per-iteration decrease of the potential function. 
Also note that the update step indeed follows the mirror descent update rule~\cite{bubeck2015convex}:
\[w_{t+1} = \arg\min_{w \in \cW} \left\{ \langle \theta_t \eta g_t, w \rangle + \theta_t \eta g_t^{\text{slack}} s(w) + D_\psi(w \| w_t) \right\},\]
where $D_\psi(w \| w_t)$ is the Bregman divergence $D_\psi(w' \| w) := \sum_{e} w'_e \log \frac{w'_e}{w_e} + s(w') \log \frac{s(w')}{s(w)}$.
\end{remark}

\subsection{Differential Privacy Analysis}
In this section, we analyze the privacy of Algorithm~\ref{alg:cut_oracle}. Recall that we employ the standard edge level privacy, i.e., two unweighted graphs $G$ and $G'$ are considered neighboring (denoted as $G \sim G'$) if they differ by exactly one edge insertion or deletion. Therefore, their corresponding true edge-weight vectors satisfy $\|w^* - (w^{*})^\prime\|_1 \leq 1$.

For notations, denote $P_{\sigma}(w,w^*)=(a_{\sigma}(w,w^*),a_{\text{slack}}(w,w^*))$ and denote $Q(X, t)=(\{\frac{1-X_{uv}}{2}\}_{uv},t)$. Note that $Q(X, t)\in \mathbb{R}^{{n\choose 2} + 1}$ and $\|Q(X,t)\|_\infty\leq1$ since $X \succeq 0$ and $\diag(X) = \mathbf{1}_n$ implies $|X_{uv}|\leq1$. Then we have $\widehat{F}_{\sigma}(w,w^*,X,t)=\langle P_{\sigma}(w,w^*), Q(X, t)\rangle$. We first prove that the value of the objective function $F_\sigma$ has bounded sensitivity, which guarantees the privacy of releasing $\widetilde{F}_{t,\sigma}$.

\begin{lemma}[Sensitivity of $F_\sigma$]\label{lem:f_sensitivity}
For any fixed edge-weight vector $w_t \in \mathcal{W}$ and any two neighboring original edge-weight vectors $w^*, w^{*\prime}$, the objectives satisfy:
\[|F_\sigma(w_t, w^*) - F_\sigma(w_t, w^{*\prime})|\leq 4\]
for both $\sigma \in \{+,-\}$. 
\end{lemma}
Consequently, by adding Gaussian noise $Z_{t,\sigma} \sim \mathcal{N}(0, \nu^2)$ with $\nu \geq \frac{4\sqrt{2\ln(1.25/\delta_0)}}{\eps_0}$, the release of $\widetilde{F}_{t,\sigma}$ is $(\eps_0, \delta_0)$-differentially private for both $\sigma \in \{+,-\}$ (Lemma~\ref{lem:gaussian-mechanism}).
\begin{proof}
Let $P_{\sigma}=(a_\sigma(w_t,w^*),a_\sigma^{\text{slack}})$ be the signed residual vectors for $w^*$, and $P_{\sigma}'=(a_\sigma(w_t,w^{*\prime}),a_\sigma^{\text{slack}})$ be the signed residual vector for $w^{*\prime}$.
For $\sigma = +$, by the definition, the change in the edge residual vector is:
\[\|a_+(w_t, w^*) - a_+(w_t, w^{*\prime})\|_1 = (1+\gamma)\|w^* - w^{*\prime}\|_1 \leq 1+\gamma.\]

Since $s(w)=M-\norm{w}_1$, the change in the slack residual is:
\[|a_+^{\text{slack}}(w_t, w^*) - a_+^{\text{slack}}(w_t, w^{*\prime})| = (1+\gamma)|s(w^*) - s(w^{*\prime})| =(1+\gamma)|\norm{w^*}_1-\norm{w^{*\prime}}_1|\leq 1+\gamma.\]
Thus, the total $\ell_1$-norm change of $(a_+, a_+^{\text{slack}})$ is bounded by $c = 2(1+\gamma) \leq 4$ (since $\gamma \in (0,1)$). Similarly, for $\sigma = -$, the total $\ell_1$-norm change of $(a_-, a_-^{\text{slack}})$ is bounded by $2(1-\gamma) \leq 2 \leq 4$. We can conclude that $\norm{P_{\sigma}-P_{\sigma}'}_1\leq4$.

Let $(X^*, t^*)$ and $(X^{*\prime}, t^{*\prime})$ be the unique optimizers of the regularized objectives parameterized by $w^*$ and $w^{*\prime}$, respectively. By the optimality of $(X^*, t^*)$, the objective value for $w^{*\prime}$ evaluated at the $(X^{*\prime}, t^{*\prime})$ must be smaller, i.e.,
\begin{align*}
    F_\sigma(w_t, w^*) 
    &= \langle P_{\sigma}, Q(X^*, t^*)\rangle + \lambda \mathcal{R}(X^*, t^*) \\
    &\geq \langle P_{\sigma}, Q(X^{*\prime}, t^{*\prime})\rangle + \lambda \mathcal{R}(X^{*\prime}, t^{*\prime}).
\end{align*}
Conversely, the objective for $w^{*\prime}$ is exactly evaluated at its own optimizer:
\[
F_\sigma(w_t, w^{*\prime})
 =\langle P_{\sigma}', Q(X^{*\prime}, t^{*\prime})\rangle+ \lambda \mathcal{R}(X^{*\prime},t^{*\prime}).
\]
Combining the above yields:
\begin{align*}
    F_\sigma(w_t, w^{*\prime}) - F_\sigma(w_t, w^*)
&\leq \langle P_{\sigma}'-P_{\sigma}, Q(X^{*\prime}, t^{*\prime})\rangle\\ 
&\leq \|P_{\sigma}-P_{\sigma}'\|_1 \|Q(X^{*\prime},t^{*\prime})\|_\infty \\
&\leq 4 \times 1 = 4.
\end{align*}
The last inequality comes from the fact $\|Q(X,t)\|_\infty\leq1$ for any $X\in\mathcal{X}$ and $t\in[\rho,1]$. A symmetric argument establishes that $F_\sigma(w_t, w^*) - F_\sigma(w_t, w^{*\prime}) \leq 4$, concluding the proof.
\end{proof}

Next, we prove the release of the gradient $\widehat{g}_t$ is private. First we need to bound the sensitivity of the optimal solution. In particular, denote $\Sigma = \left(\begin{array}{cc}
    X^{*} & 0 \\
    0 & t^{*}
\end{array}\right)$ and $\Sigma' = \left(\begin{array}{cc}
    X^{*\prime} & 0 \\
    0 & t^{*\prime}
\end{array}\right)$. Then the unclipped gradients $(\widehat{g},\widehat{g}^{\text{slack}})$ for $w^*$ and $w^{*\prime}$ are sampled from the multivariate Gaussian distributions $\mathcal{N}(0, \Sigma)$ and $\mathcal{N}(0, \Sigma')$. Thus this lemma shows that  covariance matrices from neighboring graphs are close in the relative Frobenius norm.

\begin{lemma}[Stability of the Optimizer]\label{lem:optimizer_stability}
Let $(X^*, t^*)$ and $(X^{*\prime}, t^{*\prime})$ be the unique optimizers of the regularized objectives parameterized by $w^*$ and $w^{*\prime}$, respectively. Denote the corresponding covariance matrices as $\Sigma = \left(\begin{array}{cc}
    X^{*} & 0 \\
    0 & t^{*}
\end{array}\right)$ and $\Sigma' = \left(\begin{array}{cc}
    X^{*\prime} & 0 \\
    0 & t^{*\prime}
\end{array}\right)$. If $\lambda \geq 12$, then:

\[\|\Sigma^{-1/2}(\Sigma' - \Sigma)\Sigma^{-1/2}\|_F \leq O(\frac{1}{\lambda})\leq\frac{1}{2}.\]

\end{lemma}
\begin{proof}
To match the dimensions of the optimization variables $X \in \mathcal{X} \subset \mathbb{R}^{n \times n}$ and $t \in \mathbb{R}$ encoded in the covariance matrix $\Sigma = \begin{pmatrix} X & 0 \\ 0 & t \end{pmatrix} \in \mathbb{R}^{(n+1) \times (n+1)}$, we formulate the linear objective coefficients as matrices in the exact same space. Let $C_{uv} = \begin{pmatrix} \frac{1}{2}(e_u e_v^\top + e_v e_u^\top) & 0 \\ 0 & 0 \end{pmatrix} \in \mathbb{R}^{(n+1) \times (n+1)}$ and $C^{\text{slack}} = \begin{pmatrix} 0 & 0 \\ 0 & 1 \end{pmatrix} \in \mathbb{R}^{(n+1) \times (n+1)}$.
We define the gradient matrix $G_{\sigma}$ with respect to $\Sigma$:
\[ G_{\sigma} = -\frac{1}{2}\sum_{e=\{u,v\}} a_{\sigma}(w_t,w^*)_{uv} C_{uv} + a_{\sigma}^{\text{slack}}(w_t,w^*) C^{\text{slack}}, \]
and define $G_{\sigma}'$ analogously for $w^{*\prime}$. Because the feasible space of $\Sigma$ is a convex set and the objective is concave, the constrained optimal solutions satisfy the first-order variational inequality. Specifically, since both $\Sigma$ and $\Sigma'$ are feasible, cross-evaluating the gradients $\nabla F_\sigma = G_\sigma + \lambda \Sigma^{-1}$ and $\nabla F_\sigma' = G_\sigma' + \lambda \Sigma'^{-1}$ at their respective optima yields:
\[\langle G_{\sigma} + \lambda \Sigma^{-1}, \Sigma' - \Sigma \rangle \leq 0,\quad
\langle G_{\sigma}' + \lambda \Sigma'^{-1}, \Sigma - \Sigma' \rangle \leq 0\]
Summing these two inequalities gives:
\[\lambda \langle \Sigma^{-1} - \Sigma'^{-1}, \Sigma' - \Sigma \rangle \leq \langle G_{\sigma}' - G_{\sigma}, \Sigma' - \Sigma \rangle.\]
Define $E = \Sigma^{-1/2}(\Sigma' - \Sigma)\Sigma^{-1/2}$. We denote the eigenvalues of $I+E$ as $\mu_i$. Since $I+E=\Sigma^{-1/2}\Sigma'\Sigma^{-1/2}$, it holds that $\mu_i>0$ and $\Sigma' = \Sigma^{1/2}(I+E)\Sigma^{1/2}$. Furthermore, $\mu_i \leq 1 + |\mu_i - 1| \leq 1 + \|E\|_F$ and $\norm{E}_F^2=\sum_i(\mu_i-1)^2$. Thus the left-hand side can be lower bounded as follows:
\begin{align*}
    \lambda \langle \Sigma^{-1} - \Sigma'^{-1}, \Sigma' - \Sigma \rangle
    &=\lambda \tr\left( (\Sigma^{-1} - \Sigma'^{-1})(\Sigma' - \Sigma) \right)\\
    &=\lambda \tr\left( (\Sigma^{-1} - \Sigma^{-1/2}(I+E)^{-1}\Sigma^{-1/2})(\Sigma^{1/2}(I+E)\Sigma^{1/2} - \Sigma) \right)\\
    &= \lambda \tr(E (I-(I+E)^{-1})) 
    = \lambda \sum_i \frac{(\mu_i-1)^2}{\mu_i}\\
    &\geq \lambda \frac{\sum_i(\mu_i-1)^2}{1+\norm{E}_F}
    = \lambda \frac{\norm{E}_F^2}{1+\norm{E}_F}.
\end{align*}

For the right-hand side, utilizing our earlier definition of $C_{uv}$, we note that $\Sigma^{1/2} C_{uv} \Sigma^{1/2} = \frac{1}{2} \begin{pmatrix} x_u x_v^\top + x_v x_u^\top & 0 \\ 0 & 0 \end{pmatrix}$ where $x_u = X^{1/2} e_u$. Its squared Frobenius norm expands to $$\frac{1}{2}( (x_u^\top x_v)^2 + \norm{x_u}_2^2 \norm{x_v}_2^2 ) = \frac{1}{2}(X_{uv}^2 + X_{uu}X_{vv}) = \frac{1}{2}(X_{uv}^2 + 1).$$ Since $|X_{uv}|\leq 1$, we establish $\norm{\Sigma^{1/2} C_{uv} \Sigma^{1/2}}_F \leq 1$. Thus by the Cauchy-Schwarz inequality,
\[\langle C_{uv},\Sigma' - \Sigma\rangle=\tr(C_{uv}\Sigma^{1/2} E \Sigma^{1/2})=\tr(\Sigma^{1/2}C_{uv}\Sigma^{1/2}E)\leq\norm{\Sigma^{1/2}C_{uv}\Sigma^{1/2}}_F\norm{E}_F\leq\norm{E}_F.\]
Similarly, since $\norm{\Sigma^{1/2} C^{\text{slack}} \Sigma^{1/2}}_F = t \le 1$, we have $\langle C^{\text{slack}},\Sigma' - \Sigma\rangle\leq\norm{E}_F$. Then we have,
\begin{align*}
    \langle G_{\sigma}' - G_{\sigma}, \Sigma' - \Sigma \rangle 
    =&-\frac{1}{2}\sum_{uv} (a_{\sigma}(w_t,w^{*\prime})_{uv} - a_{\sigma}(w_t,w^*)_{uv})\langle C_{uv},\Sigma' - \Sigma\rangle\\
    &+(a_{\sigma}^{\text{slack}}(w_t,w^{*\prime})-a_{\sigma}^{\text{slack}}(w_t,w^*))\langle C^{\text{slack}},\Sigma' - \Sigma\rangle\\
    &\leq\norm{P_{\sigma}-P_{\sigma}'}_1\norm{E}_F.
\end{align*}
Combining both sides yields
\[\lambda \frac{\norm{E}_F^2}{1+\norm{E}_F} \leq \norm{P_{\sigma}-P_{\sigma}'}_1\norm{E}_F.\]
From the proof of Lemma~\ref{lem:f_sensitivity}, we have $\norm{P_{\sigma}-P_{\sigma}'}_1\leq 4$. Thus if $\lambda \geq 12$, this simplifies to 
\[\|E\|_F=\|\Sigma^{-1/2}(\Sigma' - \Sigma)\Sigma^{-1/2}\|_F \leq\frac{\norm{P_{\sigma}-P_{\sigma}'}_1}{\lambda-\norm{P_{\sigma}-P_{\sigma}'}_1}=O(\frac{1}{\lambda})\leq\frac{1}{2}.\]
\end{proof}

With the stability bounded in the relative Frobenius norm, we utilize the the following lemma proven by \cite{eliavs2020differentially} to prove the gradient is privately released.

\begin{lemma}[Close Covariances Give Approximate DP, Lemma 4.10 in \cite{eliavs2020differentially}]
    \label{lem:dp_gradient}
    Let $\del_0$ a fixed parameter and $\Sigma,\Sigma'$ be symmetric positive definite matrices s.t. $\|\Sigma^{-1/2}(\Sigma' - \Sigma)\Sigma^{-1/2}\|_F \leq 1/2$. Denote $\mathrm{pdf}_{\Sigma}(x)$ and $\mathrm{pdf}_{\Sigma'}(x)$ the probability density functions of $\mathcal{N}(0,\Sigma)$ and $\mathcal{N}(0,\Sigma')$ respectively. Let $\eps_0=O(\log\frac{1}{\del_0}\|\Sigma^{-1/2}(\Sigma' - \Sigma)\Sigma^{-1/2}\|_F)$. Then we have 
    $$\mathrm{pdf}_{\Sigma}(x)\leq
    e^{\eps_0}\mathrm{pdf}_{\Sigma'}(x)$$
    with probability at least $(1-\del_0)$ over $x\in \mathcal{N}(0,\Sigma)$.
\end{lemma}

We now state and formally prove the main privacy theorem for Algorithm~\ref{alg:cut_oracle}.

\begin{theorem}[Privacy of Algorithm~\ref{alg:cut_oracle}]\label{thm:main_privacy}
For any privacy parameters $\eps, \delta \in (0, 1)$, Algorithm~\ref{alg:cut_oracle} with parameter $\lambda=\nu=\Theta(\frac{\sqrt{T}\log^{3/2}(T/\del)}{\eps})$\footnote{For simplicity, we assume $\frac{\sqrt{T}\log^{3/2}(T/\del)}{\eps}=\Omega(1)$.} satisfies $(\eps, \delta)$-differential privacy.
\end{theorem}
\begin{proof}
Algorithm~\ref{alg:cut_oracle} executes for at most $T$ iterations. In each iteration $t$, the algorithm accesses the sensitive true edge-weight vector $w^*$ exactly three times: releasing \emph{two} independent noisy objective functions ($\widetilde{F}_{t,+}$ and $\widetilde{F}_{t,-}$) and generating \emph{one} covariance sample sequence ($z, z_0$) based on $X_{t,\sigma_t},t_{t,\sigma_t}$. Consequently, the entire algorithm consists of $3T$ adaptive mechanisms.

By Lemma~\ref{lem:f_sensitivity} and the Gaussian Mechanism (Lemma~\ref{lem:gaussian-mechanism}), setting $\nu=\Theta(\frac{\log(1/\del_0)}{\eps_0})$ ensures that each of the two objective evaluations is $(\eps_0, \delta_0)$-differentially private. For the covariance sample, combining Lemma~\ref{lem:optimizer_stability} and Lemma~\ref{lem:dp_gradient} bounds the density ratio by $e^{\eps_0}$ with probability $(1-\del_0)$. Symmetrically applying it in the reverse neighboring direction (as formalized in \cite{eliavs2020differentially}) establishes that drawing the gradient samples satisfies $(\eps_0, \delta_0)$-differential privacy.

Therefore, setting the per-mechanism budgets to $\eps_0 = \Theta(\frac{\eps}{\sqrt{3T \log(1/\delta)}})$ and $\delta_0 = \Theta(\frac{\delta}{3T})$, by applying the Advanced Composition Theorem (Lemma~\ref{lem:advanced_composition}) over the $3T$ total mechanisms, Algorithm~\ref{alg:cut_oracle} is $(\eps, \delta)$-differentially private.
\end{proof}

\subsection{Utility Analysis}
We analyze the utility via a potential function based on the Bregman divergence in round $t$:
$$\Phi(t) := D_\psi(w^* \| w_t)=\sum_{e} w^*_e \log \frac{w^*_e}{(w_t)_e} + s(w^*) \log \frac{s(w^*)}{s(w_t)}.$$ 
Note that this analysis differs from the utility analysis in the private mirror descent framework established in \cite{eliavs2020differentially} and \cite{peng2025differentially}. For simplicity, we condition on the good event when all the Gaussian noises $\{Z_{t,\sigma}\}$ are bounded by $O(\nu\sqrt{\log (nT)})=\widetilde{O}_{\del}(\sqrt{T}/\eps)$ with probability at least $1-1/n^c$, by the Gaussian Tail Bound (Lemma~\ref{lem:gaussian_norm}) and a union bound over $T$ rounds.

First, we prove that the unclipped estimator $\widehat g$ is unbiased and the clip process incurs small additional error.

\begin{lemma}
\label{lem:gradient_unbias}
For every \(u,v\in\binom{V}{2}\),
\[
  \mathbb E[\widehat g_t (\{u,v\})]=\frac{1 - (X_{t,\sigma_t})_{uv}}{2},
  \qquad
  \mathbb E[\widehat g_t^2 (\{u,v\})]\le\frac{3 - 3(X_{t,\sigma_t})_{uv}}{2},
\]
\[\mathbb E[\widehat g_t^{\text{slack}}]=t_{t,\sigma_t},
  \qquad
  \mathbb E[(\widehat g_t^{\text{slack}})^2]\le3t_{t,\sigma_t}.\]
Furthermore,
\[
  0
  \le
  \mathbb E [\widehat g_t (\{u,v\})-g_t (\{u,v\})]
  \le
  2e^{-B/2},
  \qquad
  0
  \le
  \mathbb E[\widehat  g_t^{\text{slack}}-g_t^{\text{slack}}]
  \le
  2e^{-B/2}.
\]
\end{lemma}
\begin{proof}
Since $X_{uu} = X_{vv} = 1$ and $z \sim \mathcal{N}(0, X_{t,\sigma_t})$, the linear combination $(z_u - z_v)$ is a zero-mean Gaussian variable with variance:
\[\mathbb{E}[(z_u - z_v)^2] =\left(\mathbb{E}[z_u^2] - 2\mathbb{E}[z_u z_v] + \mathbb{E}[z_v^2]\right) 
    = (X_{t,\sigma_t})_{uu} - 2(X_{t,\sigma_t})_{uv} + (X_{t,\sigma_t})_{vv}= 2 - 2X_{uv}.\]
Since for Gaussian random variable $Y\sim \mathcal{N}(0,\sigma)$ there is $\mathbb{E}[Y^2] = \sigma^2$ and $\mathbb{E}[Y^4] = 3\sigma^4$, we have:
\[\mathbb{E}[\widehat g_t (\{u,v\})]=\mathbb{E}\left[\frac{1}{4}(z_u - z_v)^2\right] = \frac{1 - (X_{t,\sigma_t})_{uv}}{2}, \quad \mathbb{E}[g_t^{\text{slack}}]=\mathbb{E}[z_0^2] = t_{t,\sigma_t},\]
\[\mathbb{E}\left[\widehat{g}_t(\{u,v\})^2\right] = \frac{3}{16}\left(2 - 2(X_{t,\sigma_t})_{uv}\right)^2 \le \frac{3 - 3(X_{t,\sigma_t})_{uv}}{2}, \quad \mathbb{E}[(\widehat{g}_t^{\text{slack}})^2] = 3t_{t,\sigma_t}^2 \le 3t_{t,\sigma_t}.\]
because $X_{t,\sigma_t}\in\mathcal{X}$ ensures $|(X_{t,\sigma_t})_{uv}| \le 1$ and $t_{t,\sigma_t}\leq1$. Since $g_t(e) = \min(B, \widehat g_t(e))$ and $g_t^{\text{slack}} = \min(B, \widehat g_t^{\text{slack}})$, for every $e=\{u,v\}\in\binom{V}{2}$:
\[
\mathbb E [\widehat g_t (e)-g_t (e)] = \mathbb{E}\left[ (\widehat g_t(e) - B) \cdot \mathbb{1}[\widehat g_t(e) > B] \right] = \int_B^\infty \Prb[\widehat g_t(e) > x] dx.
\]

Since we can write $\widehat g_t(\{u,v\}) = \frac{1-(X_{t,\sigma_t})_{uv}}{2} \cdot \xi^2$, where $\xi \sim \mathcal{N}(0, 1)$ and $\xi^2$ follows a chi-squared distribution with one degree of freedom ($\chi_1^2$). Applying the standard tail bound $\Prb[\xi^2 > y] \leq e^{-y/2}$ for $y \geq 1$ (Lemma~\ref{lem:chi_square_tails}), and noting that $\frac{1-(X_{t,\sigma_t})_{uv}}{2} \leq 1$, the tail probability for all $x \geq B \geq 1$ can be bounded by:
\[
\Prb[\widehat g_t(e) > x] = \Prb[\xi^2 > \frac{2x}{1-(X_{t,\sigma_t})_{uv}}]\leq\Prb[\xi^2 > x] \leq e^{-x/2}.
\]
Integrating this tail bound yields:
\[
\mathbb E [\widehat g_t (e)-g_t (e)] \leq \int_B^\infty e^{-x/2} dx = 2e^{-B/2}.
\]
Following the identical process for the slack variable where $\widehat g_t^{\text{slack}} = t \cdot \xi^2$ with $t \leq 1$, we obtain $\mathbb E[\widehat  g_t^{\text{slack}}-g_t^{\text{slack}}]
  \le
  2e^{-B/2}$, which concludes the proof.
\end{proof}

The next lemma shows that the expected potential function is decreasing after each update step.

\begin{lemma}[Expected Potential Progress]\label{lem:potential progress}
Suppose Algorithm~\ref{alg:cut_oracle} does not terminate at iteration $t$. Assume $\alpha\geq\Omega_{\del}(\frac{\sqrt{T}}{\eps})$. By setting $B = \Theta( \log\frac{M}{\alpha})$\footnote{Without loss of generality, we assume $M >\alpha$ as $M$ is public. Otherwise, we simply release the empty graph.} such that $\eta B\leq\frac{1}{4}$ and learning rate $\eta = \Theta(\frac{\gamma}{B})$ such that $\eta \leq \frac{\gamma}{6(1+\gamma)}$, there exists a constant $c$ such that the update step guarantees an expected negative potential progress:
\[
\mathbb{E}[\Phi(t+1) - \Phi(t) \mid w_t] \leq -c \frac{\gamma \alpha}{B}.
\]
\end{lemma}

\begin{proof}
In the following proof, we always condition on $w_t$. Recall that we condition on the good event that all the Gaussian noises $\{Z_{t,\sigma}\}_{t,\sigma}$ are bounded by $O(\nu\sqrt{\log T})=\widetilde{O}_{\del}(\sqrt{T}/\eps)\leq\frac{\alpha}{32}$, by the assumption $\alpha\geq\Omega_{\del}(\frac{\sqrt{T}}{\eps})$. Suppose Algorithm~\ref{alg:cut_oracle} does not terminate at iteration $t$, then the chosen direction $\sigma_t$ ensures a significant violation:
\begin{equation}
\label{eq:utility_violation}
    \widehat{F}_{\sigma_t}(w, w^*,X_{t,\sigma_t},t_{t,\sigma_t})
    \geq \frac{19\alpha}{32}-\frac{\alpha}{32}=\frac{9}{16}\alpha.
\end{equation}
For simplicity, we set:
\begin{align*}
    U &:= \sum_{u,v} w_t(\{u,v\})\frac{1 - (X_{t,\sigma_t})_{uv}}{2} + s(w_t)t_{t,\sigma_t},\\
    U^* &:= \sum_{u,v} w^*(\{u,v\})\frac{1 - (X_{t,\sigma_t})_{uv}}{2} + s(w^*)t_{t,\sigma_t},\\
    V&:=\sum_{u,v} (w_t(\{u,v\})-w^*(\{u,v\}))g_t + (s(w_t)-s(w^*))g_t^{\text{slack}},\\
    \widehat{V}&:=\sum_{u,v} (w_t(\{u,v\})-w^*(\{u,v\}))\widehat{g}_t + (s(w_t)-s(w^*))\widehat{g}_t^{\text{slack}}.
\end{align*}
Let $W_t=\sum_{e} w_t(e) e^{-\theta_t\eta g_t(e)} + s(w_t) e^{-\theta_t\eta g_t^{\text{slack}}}$. The mirror descent update rule in \cref{eq:mirror descent update} yields that:
\begin{align*}
    \Phi(t+1) - \Phi(t) 
    &=\sum_{e} w^*_e \log \frac{(w_{t})_e}{(w_{t+1})_e} + s(w^*) \log \frac{s(w_t)}{s(w_{t+1})}\\
    &=\sum_{e} w^*_e \log \frac{(w_{t})_e}{(w_{t})_ee^{-\theta_t\eta g_t(e)}M/W_t} + s(w^*) \log \frac{s(w_t)}{s(w_{t})e^{-\theta_t\eta g_t^{\text{slack}}}M/W_t}\\
    &=\theta_t\eta \langle w^*, g_t \rangle + \theta_t\eta s(w^*) g_t^{\text{slack}}+M\log\frac{W_t}{M}.
\end{align*}
Since $e^{x} \leq 1 + x + x^2$ holds for every $x\leq1/4$ and $\eta |g_t(e)|\leq\eta B\leq\frac{1}{4}$ by the assumption\footnote{If we use the unclipped gradient, since it may be unbounded, the inequality may not hold.}, together with $\log(1+y) \leq y$, we have,
\begin{align*}
    \log\frac{W_t}{M}
    &=\log \left( \sum_{e} \frac{w_t(e)}{M} e^{-\theta_t\eta g_t(e)} + \frac{s(w_t)}{M} e^{-\theta_t\eta g_t^{\text{slack}}} \right)\\
    &\leq\log \left( \sum_{e} \frac{w_t(e)}{M} (1-\theta_t\eta g_t(e)+\eta^2 g_t^2(e)) + \frac{s(w_t)}{M} (1-\theta_t\eta g_t^{\text{slack}}+\eta^2 (g_t^{\text{slack}})^2)\right)\\
    &\leq\sum_{e} \frac{w_t(e)}{M} (-\theta_t\eta g_t(e)+\eta^2 g_t^2(e)) + \frac{s(w_t)}{M} (-\theta_t\eta g_t^{\text{slack}}+\eta^2 (g_t^{\text{slack}})^2).
\end{align*}
Therefore,
\begin{align*}
    \Phi(t+1) - \Phi(t) &\leq -\theta_t\eta \left( \langle w_t - w^*, g_t \rangle + (s(w_t) - s(w^*))g_t^{\text{slack}} \right) + \eta^2 \left( \sum_{e} w_t(e) g_t(e)^2 + s(w_t) (g_t^{\text{slack}})^2 \right)\\
    &=-\theta_t\eta V+ \eta^2 \left( \sum_{e} w_t(e) g_t(e)^2 + s(w_t) (g_t^{\text{slack}})^2 \right).
\end{align*}
By Lemma~\ref{lem:gradient_unbias}, $\mathbb E[\widehat g_t^2 (\{u,v\})]\le\frac{3 - 3(X_{t,\sigma_t})_{uv}}{2}$ and $\mathbb E[(\widehat g_t^{\text{slack}})^2]\le3t_{t,\sigma_t}$, together with $g_t\leq\widehat g_t$ and $g_t^{\text{slack}}\leq\widehat g_t^{\text{slack}}$ we have,
\[\mathbb E[g_t^2 (\{u,v\})]\leq \mathbb E[\widehat g_t^2 (\{u,v\})]\le\frac{3 - 3(X_{t,\sigma_t})_{uv}}{2},
  \qquad
  \mathbb E[( g_t^{\text{slack}})^2]\leq\mathbb E[(\widehat g_t^{\text{slack}})^2]\le3t_{t,\sigma_t}.\]
Therefore,
\[\mathbb{E}\left[\eta^2 \left( \sum_{e} w_t(e) g_t(e)^2 + s(w_t) (g_t^{\text{slack}})^2 \right)\right]\leq 3\eta^2 U.\]
By Lemma~\ref{lem:gradient_unbias}, $\mathbb E [\widehat g_t (\{u,v\})-g_t (\{u,v\})]
  \le
  2e^{-B/2}$ and $
  \mathbb E[\widehat  g_t^{\text{slack}}-g_t^{\text{slack}}]
  \le
  2e^{-B/2}.$ Thus,
\begin{align*}
\mathbb{E}[|\widehat{V}-V|] &\leq 2e^{-B/2} \left| \sum_e (w_t(e) + w^*(e)) + (s(w_t) + s(w^*)) \right| \\
&= 4M e^{-B/2}.
\end{align*}
For the case when $\sigma_t=+$, \cref{eq:utility_violation} implies $U-(1+\gamma)U^*\geq\frac{9}{16}\alpha$. Note that $\mathbb{E}[\widehat V]=U-U^*$ again because of  Lemma~\ref{lem:gradient_unbias}. Then we have,
\[
\mathbb{E}[\widehat V]=U - U^* \geq \frac{\gamma}{1+\gamma}U + \frac{9}{16(1+\gamma)}\alpha.
\]
Similarly, for the case when $\sigma_t=-$, \cref{eq:utility_violation} implies $(1-\gamma)U^*-U\geq\frac{9}{16}\alpha$. We have,
\[
\mathbb{E}[-\widehat V]=U^* - U \geq \frac{\gamma}{1-\gamma}U + \frac{9}{16(1-\gamma)}\alpha.
\]
Combining the above, we obtain,
\begin{align*}
    \mathbb{E}[\Phi(t+1) - \Phi(t)]&\leq-\eta\mathbb{E}[\theta_t\widehat V]+\eta\mathbb{E}[|\widehat V-V|]+ 3\eta^2 U\\
    &\leq-\eta U \left( \frac{\gamma}{1+\theta_t\gamma} - 3\eta \right) - \frac{9\eta\alpha}{16(1+\theta_t\gamma)} + 4\eta M e^{-B/2}.
\end{align*}
By the assumption that $\eta \leq \frac{\gamma}{6(1+\gamma)}<\frac{\gamma}{6(1-\gamma)}$, the first term is negative for all $U \geq 0$\footnote{In the purely additive error setting ($\gamma=0$), this term reduces to $3\eta^2 U$. Bounding $U$ via the Cauchy-Schwarz inequality finally yields an additive error of $\widetilde{O}(\sqrt{Mn})$, offering no improvement over the error bounds in \cite{eliavs2020differentially}.}. By setting the clipping threshold to be $B = \Theta(\log\frac{M}{\alpha})$, the third term can be smaller than $\frac{\eta\alpha}{100}$. By setting $\eta=\Theta(\frac{\gamma}{B})$, the requirement $\eta \leq \frac{\gamma}{6(1+\gamma)}$ and $\eta B\leq\frac{1}{4}$ can hold. Therefore, utilizing the fact that $\gamma \in (0, 1/4)$ and $\eta=\Theta(\frac{\gamma}{B})$, we conclude:
\[
\mathbb{E}[\Phi(t+1) - \Phi(t) \mid w_t] \leq 0- \frac{9\eta\alpha}{20} + \frac{\eta\alpha}{100} \leq -c \frac{\gamma \alpha}{B}.
\]

\end{proof}

The next lemma computes the number of the update steps needed in Algorithm~\ref{alg:cut_oracle} for obtaining the required accuracy.

\begin{lemma}[Stopping Time]\label{lem:stop time}
Suppose that \Cref{alg:cut_oracle} terminates after $\tau$ updates. If every nonterminal update $t$ satisfies
\[
 \E[\Phi(t+1)-\Phi(t)\mid w_t]
 \leq-\Delta,
\]
then,
\[
 \E[\tau]\leq O(\frac{M\log n}{\Delta}).
\]
Consequently, with probability $\geq5/6$,
\[
 \tau\leq O(\frac{M\log n}{\Delta}).
\]
\end{lemma}
Before proving this lemma, we introduce a standard tool in martingale analysis.

\begin{proposition}[Theorem 4.2.12 in \cite{durrett2019probability}]
\label{thm:optional_stopping}
Let a sequence of random variables $X = (X_t)_{t \ge 0}$ be a supermartingale (i.e., $\E[X_{t+1}|X_1,\dots,X_t]\leq X_t$) and $\tau$ be a stopping time, both with respect to a filtration $(\mathcal{F}_t)_{t \ge 0}$. If the stopping time $\tau$ is bounded, then, $\mathbb{E}[X_\tau] \le \mathbb{E}[X_1]$.
\end{proposition}

\begin{proof}[Proof of Lemma~\ref{lem:stop time}]
Let $t \ge 1$ denote the iteration index. We define the stochastic process $X_t := \Phi(t) + \Delta \cdot t$ for any $t\leq\tau$.\footnote{For any $t>\tau$, we let $X_t=X_\tau$.} By the assumption, for any $t\leq\tau$, we have,
\begin{align*}
     \E[X_{t+1} \mid w_t] 
     &= \E[\Phi(t+1) \mid w_t] + \Delta \cdot (t+1)
     =\E[\Phi(t+1)-\Phi(t) \mid w_t] +\Phi(t) + \Delta \cdot (t+1)\\
     &\le - \Delta + \Phi(t) + \Delta \cdot (t+1) = \Phi(t) + \Delta \cdot t = X_t.
\end{align*}
This implies that $(X_t)_{t \ge 1}$ is a supermartingale with respect to the filtration generated by $(w_t)_{t \ge 1}$. 

Let $\tau$ denote the first iteration where the halting condition $\max(\widetilde{F}_{t,+}, \widetilde{F}_{t,-}) \leq 19\alpha/32$ is met (with $\tau = \infty$ if it is never met). Since Algorithm~\ref{alg:cut_oracle} explicitly enforces a maximum iteration limit of $T$, the actual number of executed steps is deterministically bounded by $\tau_T = \min(\tau, T) \le T$. Because this bounded stopping time is valid for the optional stopping theorem (Proposition~\ref{thm:optional_stopping}), we have:
\[
 \E[X_{\tau_T}] \le \E[X_1] = \Phi(1) + \Delta.
\]
Since the Bregman divergence is non-negative, $\Phi(\tau_T) \ge 0$, which yields:
\[
 \Delta \cdot \E[\tau_T] \le \E[\Phi(\tau_T) + \Delta \cdot \tau_T] = \E[X_{\tau_T}] \le \Phi(1) + \Delta.
\]
By the identical calculation bounding the initial potential, we have $\Phi(1) \le 2M \log n$. Therefore, the expected number of steps before halting or hitting the horizon is bounded by:
\[
 \E[\tau_T] \le \frac{2M \log n}{\Delta} + 1 = O\left(\frac{M \log n}{\Delta}\right).
\]
To ensure the algorithm halts successfully before hitting the hard limit $T$, we set $T = \lceil 6 \cdot \E[\tau_T] \rceil = \Theta(\frac{M \log n}{\Delta})$. Applying Markov's inequality unconditionally on the bounded variable $\tau_T$:
\[
 \Prb[\tau > T] \le \Prb[\tau_T \ge T] \le \frac{\E[\tau_T]}{T} \le \frac{1}{6}.
\]
Thus, with probability at least $5/6$, the algorithm successfully terminates with a valid output within $T$ iterations.
\end{proof}

Now we state the main utility theorem for \Cref{alg:cut_oracle}.

\begin{theorem}[Utility of \Cref{alg:cut_oracle}]
\label{thm:utility_of_cut_oracle}
Let $G=(V,E)$ be a simple unweighted graph on $n$ vertices, and let $M\geq |E|$ be public.  For $\eps\in (0,1), \delta\in(0,1/2), \gamma\in(0, 1/4)$, \Cref{alg:cut_oracle} outputs a weighted graph $\widehat G$ such that, with probability at least $2/3$, for all $S\subseteq V$,
\[
  \left|\cut_{\widehat G}(S)-\cut_G(S)\right|
  \leq \gamma\,\cut_G(S)+\alpha,
\]
where
\[
  \alpha
  =\widetilde{O}_{\del}\left(
      \frac{n}{\eps}
      +\left(\frac{n^2M}{\eps^2\gamma}\right)^{1/3}
    \right).
\]
\end{theorem}
\begin{proof}
Combining Lemma~\ref{lem:potential progress} and Lemma~\ref{lem:stop time}, if $\alpha=\Omega_{\del}(\frac{\sqrt{T}}{\eps})$, setting the total iteration limits to $T = \Theta(\frac{M\log n}{\gamma\alpha/B})=\Theta(\frac{BM\log n}{\gamma\alpha})$ ensures that the algorithm outputs $w_t$ such that $\max(\widetilde{F}_{t,+}, \widetilde{F}_{t,-}) \leq 19\alpha/32$ with probability at least $5/6$. Since $\{Z_{t,\sigma}\}_{t,\sigma}$ are bounded by $O(\nu\sqrt{\log T})=\widetilde{O}_{\del}(\sqrt{T}/\eps)\leq\frac{\alpha}{32}$ with probability at least $5/6$, we have, $\max_{\sigma}F_{t,\sigma}(w_t,w^*)\leq\frac{19\alpha}{32}+\frac{\alpha}{32}=\frac{5\alpha}{8}$. By Lemma~\ref{lem:opt implies cut approx}, if $\rho = \min(\frac{1}{8},\frac{\alpha}{16\gamma M})$ and $C_0 \lambda n \log(1/\rho) \leq \frac{\alpha}{32}$, we can conclude that,
\[|\cut_w(S) - \cut_{w^*}(S)| \leq \gamma \cut_{w^*}(S) + \alpha\]
holds for all $S \subseteq V$.

Since we set $\lambda=\Theta(\frac{\sqrt{T}\log^{3/2}(T/\del)}{\eps})$ and $B=\Theta(\log\frac{M}{\alpha})$, combining them with $C_0 \lambda n \log(1/\rho) \leq \frac{\alpha}{32}$ (already ensures $\alpha=\Omega_{\del}(\frac{\sqrt{T}}{\eps})$) and $T =\Theta(\frac{BM\log n}{\gamma\alpha})$ yields 
\[\alpha=\widetilde{O}_{\del}\left(\frac{\sqrt{T}n}{\eps}\right)=\widetilde{O}_{\del}\left(\frac{\sqrt{M}n}{\sqrt{\gamma\alpha}\eps}\right).\]
Also noting that $T\geq1$, thus the total additive error
\[
\alpha = \widetilde{O}_{\del}\left( \frac{n}{\eps} + \left(\frac{n^2 M}{\eps^2 \gamma}\right)^{1/3} \right).
\]
It can be verified that all assumptions in \Cref{lem:opt implies cut approx,thm:main_privacy,lem:potential progress} are satisfied under our choice of parameters.
\end{proof}


\begin{proof}[Proof of Theorem~\ref{thm:cut-oracle-main}]
We construct the final algorithm $\mathsf{CutOracle}$ by executing the core subroutine (Algorithm~\ref{alg:cut_oracle}) repeatedly. Let $K = \lceil C \log(1/\beta) \rceil$ for a sufficiently large absolute constant $C$. The oracle runs Algorithm~\ref{alg:cut_oracle} up to $K$ times independently. If an execution returns a valid edge-weight vector $w_t$ (instead of "$\perp$"), $\mathsf{CutOracle}$ halts immediately and outputs $w_t$. If all $K$ executions return "$\perp$", it outputs an empty graph.

\paragraph{Privacy Guarantee.} 
In the worst case, the algorithm executes the subroutine $K$ times. We allocate a local privacy budget $(\epsilon_0, \delta_0)$ to each individual run. By the advanced composition Theorem (Lemma~\ref{lem:advanced_composition}), it suffices to set:
\[
    \epsilon_0 = \Theta\left(\frac{\epsilon}{\sqrt{K \log(1/\delta)}}\right), \quad \delta_0 = \frac{\delta}{2K}.
\]
As the halting condition only depends on the privately released $\widetilde{F}_{t,\sigma}$, the adaptive stopping time is private due to post-processing(Lemma~\ref{lem:postprocess}). Finally, combing Theorem~\ref{thm:main_privacy} concludes the proof.

\paragraph{Utility Guarantee.}
We bound the failure probability by analyzing two failure events. First, let $\mathcal{E}_{\text{noise}}^c$ be the event that any drawn Gaussian noise $|Z_{t, \sigma}|$ across all $K$ runs and $T$ iterations exceeds $\alpha/32$. Since $Z_{t, \sigma} \sim \mathcal{N}(0, \nu^2)$, applying standard Gaussian tail bounds (Lemma~\ref{lem:gaussian_norm}) and a union bound over the $2KT$ random variables guarantees that setting $\alpha = \Omega(\nu \sqrt{\log(KT/\beta)})$ is enough to guarantee $\Prb[\mathcal{E}_{\text{noise}}^c] \le \beta/2$. Second, let $\mathcal{E}_{\text{all abort}}$ be the event that all $K$ independent runs fail to halt within their $T$ horizons. Because the $K$ runs are mutually independent, we apply the per-run utility analysis (Theorem~\ref{thm:utility_of_cut_oracle}), which bounds the failure probability of a single run by $1/3$. Thus, the probability is $\Prb[\mathcal{E}_{\text{all abort}}] \le (1/3)^K \le \beta/2$ (by our choice of constant $C$).

By the union bound, the probability that the algorithm either suffers from large noise or abort is bounded by $\Prb[\mathcal{E}_{\text{noise}}^c \cup \mathcal{E}_{\text{all abort}}] \le \beta/2 + \beta/2 = \beta$. Therefore, with probability at least $1-\beta$, neither failure event occurs. Because $\mathcal{E}_{\text{all abort}}$ does not occur, the algorithm successfully halts and outputs a graph; because $\mathcal{E}_{\text{noise}}^c$ does not occur, the analysis in Theorem~\ref{thm:utility_of_cut_oracle} guarantees this output graph satisfies the target error bound:
\[
    \alpha = \widetilde{O}_{\del}\left( \frac{n}{\epsilon_0} + \left(\frac{n^2 M}{\epsilon_0^2 \gamma}\right)^{1/3} \right).
\]
Substituting $\epsilon_0 = \Theta\left(\frac{\epsilon}{\sqrt{\log(1/\beta) \log(1/\delta)}}\right)$, the error scales as:
\[
    \alpha = \widetilde{O}_{\del}\left( \frac{n \sqrt{\log(1/\beta)}}{\epsilon} + \left(\frac{n^2 M \log(1/\beta) }{\epsilon^2 \gamma}\right)^{1/3} \right),
\]
where the $\widetilde{O}_{\del}(\cdot)$ hides polylogarithmic factors in $n, 1/\delta$.

\paragraph{Implementation.}Follow a similar proof to \cite{eliavs2020differentially}, \Cref{alg:cut_oracle} can be implemented in time $\tilde{O}_{\eps,\del,\gamma}(n^{20/3}\log^{O(1)}(n))$  with the same level of guarantee for additive error and privacy. Here we only give a proof sketch.

By \Cref{alg:cut_oracle} and Lemma~\ref{thm:utility_of_cut_oracle}, we have $T=\Theta(\frac{BM\log n}{\gamma\alpha})=\tilde{O}_{\eps,\del,\gamma}(M^{2/3}n^{-2/3})$. In each iteration, note that the update step runs in $O(n^2)$ time. Denote $X^*,t^*$ be the maximizer of the SDP $\max_{X\in\mathcal{X},\rho\le t\le 1}\hat{F}_{\sigma}(w,w^*,X,t)$ and $\Sigma^* = \left(\begin{array}{cc}
    X^{*} & 0 \\
    0 & t^{*}
\end{array}\right)$. We use the algorithm of \cite{lee2015faster} to find an approximate solution $X,t$ (let $\Sigma = \left(\begin{array}{cc}
    X & 0 \\
    0 & t
\end{array}\right)$) of the SDP in \Cref{alg:cut_oracle}
in time $\tilde{O}(n^6\log^{O(1)}(n/\mu))$, s.t., $\norm{(\Sigma^*)^{-\frac{1}{2}}(\Sigma^*-\Sigma)(\Sigma^*)^{-\frac{1}{2}}}_F\leq\mu$. Therefore, across all $K$ rounds and $T$ iterations, our algorithm can be implemented in time $\tilde{O}_{\eps,\del,\gamma}(M^{2/3}n^{16/3}\log^{O(1)}(n))=\tilde{O}_{\eps,\del,\gamma}(n^{20/3}\log^{O(1)}(n))$. According to \cite{eliavs2020differentially}, the additive error only differs by a constant factor when we choose $\mu=1/n^{O(1)}$. And the privacy loss in each iteration is still $O(\frac{1}{\lambda})$ by Lemma~\ref{lem:optimizer_stability}, hence the privacy guarantee is the same as before. This concludes the proof of Theorem~\ref{thm:cut-oracle-main}.
\end{proof}

\section{An $\widetilde{O}(n^{13/12 + o(1)})$ Additive Error for Cut Approximation}\label{sec:cut}

In this section, we study private algorithms for approximating graph cuts. Given a private synthetic of graph Laplacian, the quadratic form $x_S^\top (L_G - {L}_{\widehat{H}}) x_S$ immediately bounds the cut error, and one can easily verify that the additive cut error for any set $S$ is directly bounded by $\|L_G - L_{\widehat{G}}\|_2 \cdot \min(|S|, |V \setminus S|)$. However, the additive error of this direct corollary on cut approximation typically reduces to $\widetilde{O}(n^{1.5})$ or worse on dense graphs. 

If allowing for multiplicative approximation, the best-known polynomial-time edge-differentially private algorithm suffers an $\widetilde{O}(n^{1.25})$ additive error~\cite{aamandbreaking}. To make further improvement, we give a \textit{black-box} reduction by assuming a generic primitive parameterized by a spectral error interface $\Pi(n,d, \epsilon, \delta)$, then develop a demand-aware recursive expander decomposition framework.
\begin{assumption}[Private Spectral Primitive] \label{assump:spectral}
For any $n$-node unweighted graph $H$ with maximum degree bounded by $d$, and for any privacy parameters $\epsilon, \delta \in (0, 1)$ and failure probability $\beta \in (0,1)$, there exists a polynomial-time $(\epsilon, \delta)$-edge-DP algorithm that outputs a non-negative weighted graph Laplacian $\widehat{L}_H$ such that, with probability at least $1-\beta$:
\begin{equation}
    \normt{L_H - \widehat{L}_H} \le \Pi(n, d, \epsilon, \delta, \beta).
\end{equation}
We assume $\Pi = \text{poly}(n,d,1/\epsilon,\log(1/\del), \log(1/\beta))$ and is non-decreasing in $n$ and $d$, and non-increasing in $\epsilon$, $\delta$ and $\beta$. For brevity, when the privacy parameters are clear from the context or not the primary focus of an asymptotic discussion, we may use the shorthand $\Pi(n, d)$.
\end{assumption}
In this section, we will show that under $(\epsilon,\delta)$-differential privacy, we reduce the additive cut error to a bound purely determined by the recursive fixed point of this spectral primitive subject to privacy composition, maintaining a multiplicative error (Theorem~\ref{thm:cut_main}). Finally, we use the spectral bound stated in Theorem~\ref{thm:improved-spectral} (Section~\ref{sec:bootstrapped_fourth_power}) to instantiate Assumption~\ref{assump:spectral} to get a polynomial time private algorithm with $\widetilde{O}(n^{13/12 + o(1)})$ additive error on cut approximation.

\subsection{Stable Port Gadget}\label{sec:stable_port_gadget}

To break the algorithmic dependency on the maximum degree in the spectral oracle (e.g., Theorem~\ref{thm:laplacian-square} or Theorem~\ref{thm:improved-spectral}), we need to transform an input graph $R$ (which might be a residual graph in some recursive steps) into a port proxy graph ${}^\sharp R$ with a smaller bounded maximum degree. For the ease of reading, we decouple this transformation into two distinct phases: (1) a differentially private mechanism to publish a demand vector $\rho \in \mathbb{Z}_{\geq 1}^n$, and (2) a sensitivity-preserving topological routing that constructs ${}^\sharp R$ conditioned on the public $\rho$.

\subsubsection{Private Demand Allocation}
We first define the demand vector $\rho$ in terms of vertices in the original graph $R$, which determines the number of micro-nodes (ports) allocated to each vertex. Such allocation should be sufficient to route all incident edges while keeping the total graph volume bounded.

\begin{lemma}[Private Demand Allocation] \label{lem:demand_allocation}
Let $R = (V, E_R)$ be an unweighted simple graph with $|V|=n$ and a public edge upper bound $M \ge |E_R|$. Let $D = \max\{1, M/n\}$. Given privacy parameter $\epsilon_{deg}$ and failure probability $\beta_{deg}$, there exists an $(\epsilon_{deg}, 0)$-edge-DP mechanism $\mathsf{BuildDemand}(\eps_{deg}, R, M)$ that outputs a demand vector $\rho: V \to \^Z_{\ge 1}^n$ satisfying:
\begin{enumerate}
    \item \textbf{Capacity:} With probability at least $1 - \beta_{deg}$, for all $v \in V$, $\rho(v) \ge deg_R(v)/D$.
    \item \textbf{Demand Mass:} With probability at least $1 - \beta_{deg}$, the total demand is bounded by $$\rho(V) \le O(n) + O\left(\frac{n\ln(n/\beta_{deg})}{D\cdot \eps_{deg}}\right) = \widetilde{O}(n).$$ 
Further, $\rho(V)$ is deterministically bounded by $\rho(V) \le \operatorname{poly}(n, 1/\epsilon_{deg}, \log(1/\beta_{deg}))$.
\end{enumerate}
\end{lemma}

\begin{proof}
The true degree vector $deg_R(\cdot)$ has an $\ell_1$-sensitivity of exactly $2$ under edge-neighboring relations, as adding or removing an edge $\{u, v\}$ changes the degrees of both $u$ and $v$ by exactly $1$. We apply the Laplace mechanism (Lemma~\ref{lem:laplace-priv}) by adding independent noise variables $Z_v \sim \lap(2/\epsilon_{deg})$ and computing the noisy degrees $\widetilde{d}_v = deg_R(v) + Z_v$. This is $(\eps_{deg}, 0)$-DP.

To ensure that the capacity condition holds even under worst-case negative noise, we define a high-probability noise bound $a = \frac{2}{\epsilon_{deg}} \ln\Paren{\frac{n}{\beta_{deg}}}$. We further apply a deterministic upper cap: since the true degree satisfies $deg_R(v) < n$, we hard-cap the shifted noisy degree at $n+2a$:
\begin{equation} \label{eq:demand_allocation}
    b_v = \rho(v) = 1 + \ctp{\frac{\min\{n+2a, \max\{0, \widetilde{d}_v + a\}\}}{D}}.
\end{equation}
Deterministically, this hard cap limits the total demand to $$\rho \leq n\left(2 + \frac{n + \frac{4}{\epsilon_{deg}}\ln(n/\beta_{deg})}{D}\right).$$

By the tail bounds of the Laplace distribution (Lemma~\ref{lem:laplace-tail}) and a union bound over all $n$ vertices, the good event that $|Z_v| \le a$ for all $v \in V$ holds with probability at least $1 - \beta_{deg}$. Conditioning on this good event, $Z_v \ge -a$ guarantees that $\widetilde{d}_v + a = deg_R(v) + Z_v + a \ge deg_R(v)$. Because the true degree $deg_R(v) < n \le n+2a$, the deterministic cap never pulls the value below the true degree. Consequently, the allocated number of ports unconditionally guarantees $b_v \ge deg_R(v)/D$ on this good event, satisfying the Capacity property.

To bound the total demand mass $\rho(V)$ accurately on this good event, we sum the unclipped values over all vertices:
\begin{equation*}
    \rho(V) = \sum_{v \in V} b_v \le n + \sum_{v \in V} \frac{deg_R(v) + 2a + 1}{D} \le n + \frac{2M}{D} + \frac{n(2a+1)}{D}.
\end{equation*}
Since $D = \max\{1, M/n\}$, we have $2M/D \le 2n$. Therefore, substituting the constants yields $\rho(V) \le 3n + n(2a+1)/D = \widetilde{O}_{\epsilon_{deg}}(n)$, concluding the proof.
\end{proof}

\subsubsection{Stable Port Routing}
Treating the demand vector $\rho$ defined in \cref{eq:demand_allocation} as a fixed \emph{public} parameter, we now construct the topology of the proxy graph ${}^\sharp R$. We note that the construction should be stable in terms of edge change: changing one edge in $R$ should only trigger a constant number of edge changes in ${}^\sharp R$.

\begin{lemma}[Stable Port Graph] \label{lem:stable_gadget}
Given a public demand vector $\rho \in \mathbb{Z}^n$ and failure probability $\beta_{util}$, there is a routing mechanism $\mathsf{BuildGadget}(R,\rho)$ using public hashes to construct a port graph ${}^\sharp R$ on the lifted vertex set $V({}^\sharp R) = \{(v,i): v \in V, i \in [\rho(v)]\}$. This construction satisfies:
\begin{enumerate}
    \item \textbf{Degree Bound:} The maximum degree $\Delta({}^\sharp R) \le {O}(D) + O\left(\ln({n}/{\beta_{util}})\right)$.
    \item \textbf{Edge-Stability:} If $R$ and $R'$ differ by exactly one edge, the resulting graphs ${}^\sharp R$ and ${}^\sharp R'$ differ by at most $4$ edges.
    \item \textbf{Cut Preservation:} If $\rho$ satisfies the Capacity property of Lemma~\ref{lem:demand_allocation}, then with probability at least $1-\beta_{util}$, for all $S \subseteq V$, defining the lifted cut ${}^\sharp S = \{(v,i): v \in S, i \in [\rho(v)]\}$ yields $w_R(S, V \setminus S) = w_{{}^\sharp R}({}^\sharp S, V({}^\sharp R) \setminus {}^\sharp S)$.
\end{enumerate}
\end{lemma}

\begin{proof}
For each vertex $v \in V$, we instantiate its $\rho(v)$ ports. To distribute the original incident edges among these ports, we draw fresh, fully independent public hash functions $h_v: V \setminus \{v\} \to [\rho(v)]$. For any original edge $e = \{u, v\} \in E_R$, its projection is deterministically mapped to the port pair $\{(u, h_u(v)), (v, h_v(u))\}$. To enforce a worst-case upper bound on the maximum degree of ${}^\sharp R$, we apply a deterministic truncation mechanism. For each port $(v, j)$, we define its candidate incident set as $C_{v,j}(R) = \{ \{u, v\} \in E_R : h_v(u) = j \}$. Instead of accepting all candidates, we sort $C_{v,j}(R)$ according to a globally fixed, public lexicographical order based on the vertex IDs. We truncate the sorted list, retaining only the first $B = 2\e D + 3 \ln(2 \rho(V) / \beta_{util})$ elements, yielding the truncated set $\mathrm{Top}_B(C_{v,j}(R))$. A candidate edge is added to ${}^\sharp R$ if and only if it survives the truncation at \emph{both} endpoints. This hard cutoff guarantees that the maximum degree $\Delta({}^\sharp R) \le B$. Since Lemma~\ref{lem:demand_allocation} bounds $\rho(V) \le \widetilde{O}(n)$, the threshold simplifies to $B = {O}(D) + O\left(\ln({n}/{\beta_{util}})\right)$, which proves the Degree Bound property.

To verify the Edge-Stability, suppose the private input graph changes by exactly one edge $e = \{x, y\}$. Because the hash mappings and the lexicographical sorting order are entirely public and independent of the graph topology, this perturbation modifies the candidate sets of exactly two ports: $C_{x, h_x(y)}$ and $C_{y, h_y(x)}$. Within either candidate set, the addition or removal of $e$ translates to inserting or deleting exactly one element in a deterministically sorted list. This operation shifts the truncation boundary, altering the membership of the $\mathrm{Top}_B$ set by at most two elements (the perturbed element itself and the boundary element pushed out of or pulled into the list). Consequently, the symmetric difference between $E({}^\sharp R)$ and $E({}^\sharp R')$ is bounded by at most four edges, ensuring $O(1)$ sensitivity globally regardless of the graph's density.

Finally, we evaluate the probability of truncation under the assumption that the Capacity condition holds. Given $\rho(v) \ge deg_R(v)/D$ (Lemma~\ref{lem:demand_allocation}), the expected load for any port $(v, j)$ is $\mu_{v,j} = \E{|C_{v,j}(R)|} = deg_R(v)/\rho(v) \le D$. By applying the standard Chernoff bound (Lemma \ref{lem:chernoff}), the probability that the load exceeds the threshold $B$ is upper bounded by $$\Prb[|C_{v,j}(R)| \ge B] \le \frac{\beta_{util}}{\rho(V)}.$$ Taking a union bound over all $\rho(V)$ ports across the entire graph, the probability that any port overflows its capacity $B$ is bounded by $\beta_{util}$. When no overflow occurs, every edge in $E_R$ passes the truncation at both endpoints, establishing an exact bijection between the edges of $R$ and ${}^\sharp R$. Therefore, mapping the vertices of $S$ to their corresponding ports ${}^\sharp S$ preserves the cut edges across the boundary, yielding $w_R(S, V \setminus S) = w_{{}^\sharp R}({}^\sharp S, V({}^\sharp R) \setminus {}^\sharp S)$ holds with probability at least $1- \beta_{util}$, which concludes the proof.
\end{proof}
\subsection{Demand-Sensitive Private Cut Oracle}\label{sec:demand_cut_oracle}

Building upon the stable port gadget in Section~\ref{sec:stable_port_gadget}, we construct a private cut oracle capable of achieving additive errors proportional only to the demand mass $\rho(S)$ (defined in \cref{eq:demand_allocation}) and the spectral primitive's error interface $\Pi$, rather than the true maximum degree $d$. 

\begin{lemma}[Demand-Sensitive Private Cut Approximation] \label{lem:demand_oracle}
Let $H = R[U]$ be any induced subgraph of the input graph $R = (V, E_R)$ subject to the public upper bound $|E_R| \le M$. Given the public demand vector $\rho$ from Lemma \ref{lem:stable_gadget}, there exists an $(\epsilon, \delta)$-edge-DP mechanism $\mathsf{SyntheticOracle}(H, \rho, \epsilon, \delta, \beta)$ outputting a synthetic non-negative weighted graph $\widehat{H}$ such that, with probability at least $1-\beta$, for all $S \subseteq U$:
\begin{equation}
    |w_H(S, U \setminus S) - w_{\widehat{H}}(S, U \setminus S)| \le \widetilde{O} \Paren{ \Pi(n, M/n, \epsilon, \delta, \beta) } \min\{\rho(S), \rho(U \setminus S)\}.
\end{equation}
\end{lemma}

\begin{proof}
The oracle $\mathsf{SyntheticOracle}(H, \rho,\epsilon,\delta)$ simply first builds the proxy graph ${}^\sharp H$ with $\rho$ and then run the generic spectral oracle (Assumption~\ref{assump:spectral}) on ${}^\sharp H$ to generate a synthetic graph. 

Instead of inducing from a globally truncated port graph ${}^{\sharp}R$ (which could leak information and introduce more sensitivity from edges outside $U$), we construct the local port proxy graph ${}^\sharp H$ with $\mathsf{BuildGadget}(H, \rho)$ (Lemma~\ref{lem:stable_gadget}) directly on the induced subgraph $H = R[U]$.  We define the lifted vertex set exclusively for $U$ as ${}^\sharp U = \bigcup_{v \in U} \{(v, j) : j \in [\rho(v)]\}$. The total number of micro-nodes in ${}^\sharp H$ evaluates to the local demand mass, bounded by the global capacity:
\begin{equation*}
    |V({}^\sharp H)| = |{}^\sharp U| = \sum_{v \in U} \rho(v) = \rho(U) \le \rho(V) \le \widetilde{O}(n).
\end{equation*}

We then invoke the private spectral primitive (Assumption \ref{assump:spectral}) on ${}^\sharp H$. Because the addition or removal of an original edge in $R$ modifies at most $c=4$ edges in ${}^\sharp H$\footnote{Note that the truncation in $\mathsf{BuildGadget}(H, \rho)$ (Lemma~\ref{lem:stable_gadget}) is applied exclusively to the local edge set $E(H)$, edges outside $U$ have no effect. Thus, any single edge change in the global graph $R$ modifies at most $4$ edges in ${}^\sharp H$.}, treating the graph ${}^\sharp H$ under single-edge DP incurs a group privacy reduction of size $c=4$. By standard group privacy guarantees for approximate differential privacy, running the primitive with calibrated privacy parameters $\epsilon' = \epsilon / c$ and $\delta' = \delta / \sum_{j=0}^{c-1} e^{j\epsilon'}$ ensures $(\epsilon, \delta)$-edge-DP. The mechanism outputs a synthetic graph Laplacian $\widehat{L}_{{}^\sharp H}$. We define the error matrix $E = L_{{}^\sharp H} - \widehat{L}_{{}^\sharp H}$. Given that $\Pi(\cdot, \cdot, \cdot, \cdot, \cdot)$ is monotonically non-decreasing in dimensions and monotonically non-increasing in privacy parameters, with probability at least $1-\beta$, the spectral norm of $E$ satisfies:
\begin{equation} \label{eq:spectral_norm_E}
    \normt{E} \le \widetilde{O} \Paren{ \Pi(|V({}^\sharp H)|, \Delta({}^\sharp H), \epsilon', \delta', \beta) }.
\end{equation}
Here, $\Delta({}^\sharp H)$ is bounded by $O(\max\{1, M/n\}) + O(\log(n/\beta'))$ due to Lemma~\ref{lem:stable_gadget}. According to Assumption~\ref{assump:spectral}, the polylogarithmic inflation in vertices and degrees, alongside the constant splitting of privacy budgets, can be absorbed into the $\widetilde{O}$ notation. This cleanly yields:
\begin{equation} \label{eq:spectral_norm_E_simplified}
\normt{E} \le \widetilde{O} \Paren{ \Pi(n, \max\{1, M/n\}, \epsilon, \delta, \beta) } := \Lambda_\rho(M, \epsilon, \delta, \beta).
\end{equation}
To form the final synthetic graph $\widehat{H}$ on the \textbf{original} vertex space $U$, we collapse the micro-nodes back to their source vertices. Note that, the Laplacian $\widehat{L}_H$ of the synthetic graph $\widehat{H}$ is defined by summing the corresponding block submatrices of $\widehat{L}_{{}^\sharp H}$:
\begin{equation*}
    \widehat{L}_H(u,v) = \sum_{i=1}^{\rho(u)} \sum_{j=1}^{\rho(v)} \widehat{L}_{{}^\sharp H}((u,i), (v,j)).
\end{equation*}
For any cut $S \subseteq U$, we lift it to the port space by defining its indicator vector $1_{{}^\sharp S} \in \{0,1\}^{|{}^\sharp U|}$, where $(1_{{}^\sharp S})_{(v, i)} = 1$ if and only if $v \in S$. Due to the exact cut preservation of the stable gadget, all inter-cluster gadget edges exactly correspond to original edges with probability at least $1-\beta_{util} = 1-1/n^c$. Consequently, the cut weight in $H$ can be formulated identically as the quadratic form of the port Laplacian:
\begin{equation*}
    w_H(S, U \setminus S) = w_{{}^\sharp H}({}^\sharp S, {}^\sharp U \setminus {}^\sharp S) = 1_{{}^\sharp S}^\top L_{{}^\sharp H} 1_{{}^\sharp S}.
\end{equation*}
Simultaneously, the matrix block-sum definition guarantees that the cut value in the synthetic graph $\widehat{H}$ translates to the same quadratic form under the synthetic Laplacian:
\begin{equation*}
    w_{\widehat{H}}(S, U \setminus S) = 1_{{}^\sharp S}^\top \widehat{L}_{{}^\sharp H} 1_{{}^\sharp S}.
\end{equation*}
The absolute additive error is therefore equivalent to the quadratic form evaluated over the error matrix $E$:
\begin{equation} \label{eq:additive_error_S}
    |w_H(S, U \setminus S) - w_{\widehat{H}}(S, U \setminus S)| = \abs{ 1_{{}^\sharp S}^\top (L_{{}^\sharp H} - \widehat{L}_{{}^\sharp H}) 1_{{}^\sharp S} } = \abs{ 1_{{}^\sharp S}^\top E 1_{{}^\sharp S} }.
\end{equation}
Applying the standard operator norm inequality, we bound the quadratic form utilizing \cref{eq:spectral_norm_E} and the $\ell_2$-norm of the indicator vector:
\begin{equation} \label{eq:error_bound_S}
    \abs{ 1_{{}^\sharp S}^\top E 1_{{}^\sharp S} } \le \normt{E} \normt{1_{{}^\sharp S}}^2 \le \Lambda(M, \epsilon, \delta) \sum_{v \in S} \sum_{i=1}^{\rho(v)} 1^2 = \Lambda(M, \epsilon, \delta) \sum_{v \in S} \rho(v) = \Lambda(M, \epsilon, \delta) \rho(S).
\end{equation}
To establish the complementary bound for the complement cut $U \setminus S$, we note that by the definition of graph Laplacian matrix:
$L_{{}^\sharp H} \mathbf{1} = \mathbf{0}$ and $\widehat{L}_{{}^\sharp H} \mathbf{1} = \mathbf{0}$ hold unconditionally, implying $$E \mathbf{1} = (L_{{}^\sharp H}  - \widehat{L}_{{}^\sharp H} )\mathbf{1} = L_{{}^\sharp H} \mathbf{1} - \widehat{L}_{{}^\sharp H} \mathbf{1}= \mathbf{0}.$$ We express the all-ones vector as $\mathbf{1} = 1_{{}^\sharp S} + 1_{{}^\sharp U \setminus {}^\sharp S}$ and expand the quadratic form of the complement indicator vector:
\begin{equation*}
    1_{{}^\sharp U \setminus {}^\sharp S}^\top E 1_{{}^\sharp U \setminus {}^\sharp S} = (\mathbf{1} - 1_{{}^\sharp S})^\top E (\mathbf{1} - 1_{{}^\sharp S}) = \mathbf{1}^\top E \mathbf{1} - \mathbf{1}^\top E 1_{{}^\sharp S} - 1_{{}^\sharp S}^\top E \mathbf{1} + 1_{{}^\sharp S}^\top E 1_{{}^\sharp S}.
\end{equation*}
Note that $E \mathbf{1} = \mathbf{0}$ and $\mathbf{1}^\top E = \mathbf{0}^\top$, which yields that:
\begin{equation*}
    1_{{}^\sharp U \setminus {}^\sharp S}^\top E 1_{{}^\sharp U \setminus {}^\sharp S} = 1_{{}^\sharp S}^\top E 1_{{}^\sharp S}.
\end{equation*}
Substituting this symmetric equality into the operator norm inequality mirroring \cref{eq:error_bound_S}:
\begin{equation} \label{eq:error_bound_US}
    \abs{ 1_{{}^\sharp S}^\top E 1_{{}^\sharp S} } = \abs{ 1_{{}^\sharp U \setminus {}^\sharp S}^\top E 1_{{}^\sharp U \setminus {}^\sharp S} } \le \normt{E} \normt{1_{{}^\sharp U \setminus {}^\sharp S}}^2 \le \Lambda(M, \epsilon, \delta) \sum_{v \in U \setminus S} \rho(v) = \Lambda(M, \epsilon, \delta) \rho(U \setminus S).
\end{equation}
Combining Eq. \eqref{eq:error_bound_S} and Eq. \eqref{eq:error_bound_US} simultaneously restricts the absolute error to the minimum of the two respective demand masses, concluding the proof.
\end{proof}

\subsection{Private Expander Decomposition with Demand}\label{sec:expander-decomposition}

To conduct expander decomposition and privately partition the graph into dense components, we utilize the deterministic bicriteria algorithm by Li and Saranurak (\cite{li2021deterministic}). 
\begin{theorem}[Demand-Aware Bicriteria Cut Routine, Theorem 2.14 in \cite{li2021deterministic}] \label{thm:demand_cut_routine}
There is a deterministic polynomial-time algorithm that takes as input a graph $H = (V_H, E_H)$ with non-negative edge weights $w_H$, a vertex demand vector $\rho: V_H \to \mathbb{R}_{\ge 0}$, and a sparsity parameter $\phi > 0$. The algorithm outputs either a cut $S \subset V_H$ or $\emptyset$ satisfying the following properties:
\begin{enumerate}
    \item \textbf{Sparsity Guarantee:} If the algorithm outputs $S \neq \emptyset$, then $S$ is the side with the smaller or equal demand, i.e., $\rho(S) \le \rho(V_H \setminus S)$, and its demand-sparsity satisfies:
    \begin{equation*}
        w_H(S, V_H \setminus S) \le \phi \cdot \rho(S).
    \end{equation*}
    
    \item \textbf{Bicriteria Maximality (Certifier):} If there exists any other sparse cut $S^* \subset V_H$ with $\rho(S^*) \le \rho(V_H \setminus S^*)$ such that: $w_H(S^*, V_H \setminus S^*) \le \frac{\phi}{d_{\text{exp}}} \cdot \rho(S^*),$ then the output $S \neq \emptyset$ satisfies that:
    \begin{equation*}
        \rho(S) \ge \frac{\rho(S^*)}{d_{\text{size}}}.
    \end{equation*}
\end{enumerate}
Here, $d_{\text{exp}} = O(\log^c n)$ and $d_{\text{size}} = O(1)$ are global constants.
\end{theorem}

We adapt this algorithm to operate over the private synthetic graph output by our demand-sensitive oracle in Section~\ref{sec:demand_cut_oracle}. In particular, we formalize this procedure of the private cut finder in Algorithm~\ref{algo:cut_finder}. 
The algorithm acts as a differentially private wrapper, utilizing the synthetic graph from our demand-sensitive oracle to execute the deterministic routine in Theorem~\ref{thm:demand_cut_routine}.

\begin{algorithm}[h]
\caption{Private Demand Cut Finder ($\mathsf{CutFinder}$)}
\label{algo:cut_finder}
\KwIn{Induced subgraph $H = R[U]$, public demand vector $\rho$, target sparsity parameter $\psi$, privacy parameters $(\epsilon,\delta)$.}
\KwOut{A cut $S \subset U$ satisfying $\rho(S) \le \rho(U \setminus S)$ and $w_H(S, U \setminus S) \le \psi \rho(S)$, or $\emptyset$}

\tcp{Private Graph Synthesis via Spectral Oracle (Lemma \ref{lem:demand_oracle})}
Let $\widehat{H} \gets \mathsf{SyntheticOracle}(H, \rho, \epsilon, \delta)$\;
\tcp{Trim and Cap (Deterministic Weight Regularization)}
Set $\theta \gets \frac{\psi}{10n}$ and $W_{\max} \gets \psi \rho(U)$\;

\For{\textbf{each} edge $e \in E(\widehat{H})$}{
    \uIf{$w_e < \theta$}{
        Delete $e$ from $\widehat{H}$\;
    }
    \ElseIf{$w_e > W_{\max}$}{
        Truncate weight $w_e \gets W_{\max}$\;
    }
}

\tcp{Deterministic Bicriteria Search}
Execute the Demand-Aware Cut Routine (Theorem \ref{thm:demand_cut_routine}) on $\widehat{H}$ with input demand vector $\rho$ and target sparsity $\phi = 0.8\psi$\;
\eIf{the routine returns a valid cut $S$}{
    \Return $S$\;
}{
    \Return $\emptyset$ \tcp*{Certifies $H$ is a valid $(\psi/d_{exp}, \rho)$-expander}
}
\end{algorithm}

\begin{lemma}[Private Demand Cut Finder] \label{lem:cut_finder}
Fix any $U\subseteq V$ and let $H = R[U]$ be an induced subgraph. Assume we have a synthetic graph $\widehat{H}$ satisfying the guarantee of Lemma \ref{lem:demand_oracle} with additive error bound $\Lambda_\rho $ on any cut value $w(S,U\backslash S)$ for $S\subset U$. For any parameter $\psi \ge 10 c_{exp} \Lambda_\rho$ where $c_{exp} = O(\log^c n)$, there exists an $(\epsilon,\delta)$-edge-DP algorithm $\mathsf{CutFinder}(H, \psi)$ (i.e., Algorithm~\ref{algo:cut_finder}) that outputs either a cut $S \subset U$ with $\rho(S) \le \rho(U \setminus S)$ or $\emptyset$, satisfying:
\begin{enumerate}
    \item \textbf{Sparsity:} If $S \neq \emptyset$, then $w_H(S, U \setminus S) \le \psi \rho(S)$.
    \item \textbf{Bicriteria Maximality:} If there exists any $S' \subset U$ satisfying $\rho(S') \le \rho(U \setminus S')$ and $w_H(S', U \setminus S') \le \frac{\psi}{c_{exp}} \rho(S')$, the output $S \neq \emptyset$ satisfies that $\rho(S) \ge \frac{\rho(S')}{c_{size}}$, where $c_{size} = O(1)$.
\end{enumerate}
\end{lemma}

\begin{proof}
We first construct a \emph{trimmed} and capped proxy graph $\widehat{H}'$ from $\widehat{H}$ to bound the dynamic range of the edge weights. We define a minimum weight threshold $\theta = \frac{\psi}{10n}$ and a maximum capacity bound $W_{\max} = \psi \rho(U)$. To form $\widehat{H}'$, we simply delete any edge $e \in \widehat{H}$ with weight $w_e < \theta$, and truncate any edge with $w_e > W_{\max}$ to exactly $W_{\max}$. The ratio of the maximum to non-zero minimum edge weight in $\widehat{H}'$ is therefore bounded by $W_{\max} / \theta = 10 n \rho(U) \le \mathrm{poly}(n)$. We then execute the deterministic demand-vector routine (Theorem~\ref{thm:demand_cut_routine}, or Theorem 2.14 in \cite{li2021deterministic}) on $\widehat{H}'$ with a target sparsity threshold of $0.8\psi$. If the algorithm in Theorem~\ref{thm:demand_cut_routine} fails to find a cut meeting this threshold, we output $\emptyset$; otherwise, we output the returned cut $S$.

First, we assume the algorithm outputs $S \neq \emptyset$ and prove the sparsity guarantee. By the guarantee of such cut in Theorem~\ref{thm:demand_cut_routine}, $w_{\widehat{H}'}(S, U \setminus S) \le 0.8\psi \rho(S)$. We prove by contradiction that the cut $S$ does not sever any capped edges. If it did, the weight of the cut in $\widehat{H}'$ would be at least $W_{\max} = \psi \rho(U)$. Since $S$ is the smaller demand side, $\rho(S) \le \rho(U)/2$, which implies $\psi \rho(U) \ge 2\psi \rho(S) > 0.8\psi \rho(S)$, which leads to a contradiction. Thus, the cut $S$ evaluates identically in $\widehat{H}'$ and $\widehat{H}$ except for the deleted small edges. The total weight of all deleted edges crossing the cut is naively bounded by $\theta \cdot |S| \cdot |U \setminus S| \le \theta n \rho(S) = 0.1\psi \rho(S)$. Consequently, the uncut synthetic weight is bounded by $w_{\widehat{H}}(S, U \setminus S) \le w_{\widehat{H}'}(S, U \setminus S) + 0.1\psi \rho(S) \le 0.9\psi \rho(S)$. Applying the additive error guarantee from Lemma \ref{lem:demand_oracle}, the true weight in $H$ evaluates to:
\begin{align*}
    w_H(S, U \setminus S) &\le w_{\widehat{H}}(S, U \setminus S) + \Lambda_\rho \rho(S) \nonumber \\
    &\le 0.9\psi \rho(S) + \Paren{\frac{\psi}{10 c_{exp}}} \rho(S) \le \psi \rho(S).
\end{align*}

Next, we establish the bicriteria maximality. Assume there exists an optimal sparse cut $S'$ satisfying $w_H(S', U \setminus S') \le \frac{\psi}{c_{exp}} \rho(S')$. The oracle's additive error guarantees that the weight of $S'$ in the synthetic graph $\widehat{H}$ is bounded by:
\begin{equation*}
    w_{\widehat{H}}(S', U \setminus S') \le w_H(S', U \setminus S') + \Lambda_\rho \rho(S') \le \Paren{ \frac{\psi}{c_{exp}} + \Lambda_\rho } \rho(S').
\end{equation*}
By setting $c_{exp} = 4 d_{exp}$ (recall that $d_{exp} = \mathrm{polylog}(n)$ is the internal expansion loss of in Theorem~\ref{thm:demand_cut_routine}) and invoking the premise $\psi \ge 10 c_{exp} \Lambda_\rho$, we see that:
\begin{equation*}
    \frac{\psi}{c_{exp}} + \Lambda_\rho \le \frac{\psi}{4 d_{exp}} + \frac{\psi}{40 d_{exp}} = 0.275 \frac{\psi}{d_{exp}}.
\end{equation*}
Because the trimming and capping operations can only decrease the weight of an existing cut, the weight of $S'$ in $\widehat{H}'$ satisfies that $w_{\widehat{H}'}(S', U \setminus S') \le 0.275 \frac{\psi}{d_{exp}} \rho(S') \le 0.8 \frac{\psi}{d_{exp}} \rho(S')$. Under this condition, the bicriteria guarantee in Theorem~\ref{thm:demand_cut_routine} dictates that the returned cut $S \neq \emptyset$ satisfies $\rho(S) \ge \frac{\rho(S')}{d_{size}}$. Setting $c_{size} = d_{size} = O(1)$ concludes the utility proof. 

For the privacy proof, we note that Algorithm~\ref{algo:cut_finder} does not consume more privacy budgets than the private spectral oracle (Assumption~\ref{assump:spectral}) as all trimming, capping, and the computations in Theorem~\ref{thm:demand_cut_routine} are deterministic post-processing steps applied exclusively to the $(\epsilon,\delta)$-edge-DP synthetic graph $\widehat{H}$ given by the spectral primitive in Assumption~\ref{assump:spectral}.
\end{proof}

We then integrate this demand-aware cut finder into a recursive decomposition framework. Because
the performance metric scales with $\rho(S)$ rather than $|S|$, we slightly generalize the classical
recursion in Nanongkai and Saranurak~\cite{nanongkai2017dynamic} by replacing cardinalities with
demand masses. As in~\cite{nanongkai2017dynamic}, an auxiliary graph $I$ will be used only in the
analysis and need not be maintained by the executable algorithm.

\begin{algorithm}[h]
\caption{Demand-Aware Private Expander Decomposition ($\mathsf{DemandExpDecomp}$)}
\label{algo:expander_decomposition}
\KwIn{Graph $H=G[U]$ with public demand vector
$\rho \in \mathbb{Z}_{\ge 1}^{V(H)}$, target sparsity $\psi$, privacy budgets
$(\epsilon,\delta)$, and failure probability $\beta$.}
\KwOut{A disjoint partition $\+P = \{V_1, \dots, V_k\}$ of $V(H)$.}

\If{$H$ is a singleton}{
    \Return $\{V(H)\}$\;
}

\tcp{Part 1: Global Initialization \& Privacy Budgeting}
$\+R \gets \rho(V(H))$ \tcp*{Total global demand}
$\widehat c_{exp} \gets 2c_{exp}$ \tcp*{Constant slack used in the depth proof}
$\sigma \gets \min\left\{1,\sqrt{\log \widehat c_{exp}/\log \+R}\right\}$ \tcp*{Step decay rate}
$L \gets \min\left\{\ell\ge 1:
(\+R/2+1)/\+R^{(\ell-1)\sigma}\le 1\right\}$ \tcp*{Maximum level}
\For{$i = 1$ \KwTo $L+1$}{
    $\overline{s}_i \gets (\+R/2 + 1) / \+R^{(i-1)\sigma}$
    \tcp*{Pre-compute exact mass thresholds}
}
$\Gamma \gets \left\lceil 3c_{size}\+R^\sigma\right\rceil$
\tcp*{Maximum consecutive larger-side steps}
$D_{\max} \gets
(1+\lceil\log_2 \+R\rceil)\bigl(1+L(1+\Gamma)\bigr)$
\tcp*{Hard recursion-depth cap}
$N_{\max} \gets |V(H)|D_{\max}$ \tcp*{Hard upper bound on oracle calls}
$\epsilon_{step} \gets \epsilon / D_{\max}$,
$\delta_{step} \gets \delta / D_{\max}$,
$\beta_{step} \gets \beta / N_{\max}$\;

\BlankLine
\tcp{Part 2: The Recursive Core}
\SetKwProg{Fn}{Function}{:}{}
\Fn{$\mathsf{RecDecomp}(H', l, d)$}{
    \If{$H'$ is a singleton}{
        \Return $\{V(H')\}$\;
    }
    \If{$d \ge D_{\max}$ \textbf{or} $l > L$}{
        \Return $\{V(H')\}$
        \tcp*{Hard cap; this branch is not reached on the success event}
    }
    $\psi_l \gets \dfrac{\psi}{2}\widehat c_{exp}^{\,L-l+1}$\;
    $S \gets \mathsf{CutFinder}
    (H',\psi_l,\rho,\epsilon_{step},\delta_{step},\beta_{step})$\;
    \If{$S = \emptyset$}{
        \Return $\{V(H')\}$\;
    }
    \eIf{$\rho(S) \ge \overline{s}_{l+1}/c_{size}$}{
        \Return
        $\mathsf{RecDecomp}(H'[S],1,d+1)
        \cup
        \mathsf{RecDecomp}(H'[V(H')\setminus S],l,d+1)$\;
    }{
        \Return $\mathsf{RecDecomp}(H',l+1,d+1)$\;
    }
}
\BlankLine
\Return $\mathsf{RecDecomp}(H,1,0)$\;
\end{algorithm}

\begin{lemma}[Properties of Algorithm~\ref{algo:expander_decomposition}]
\label{lem:demand_ns17}
Given a global privacy budget $(\epsilon,\delta)$ and a global failure probability
$\beta\in(0,1)$, let
$\Lambda_\rho(\epsilon',\delta',\beta')$ be a uniform upper bound on the
additive error of the private oracle in Lemma~\ref{lem:demand_oracle}, when
invoked with local parameters $(\epsilon',\delta',\beta')$.
Let $D_{\max},N_{\max},\epsilon_{step},\delta_{step},\beta_{step}$ be the
public quantities defined in Algorithm~\ref{algo:expander_decomposition}.
For every target expansion parameter satisfying
\begin{equation}
    \psi
    \ge
    10\Lambda_\rho
    \Paren{
        \epsilon_{step},
        \delta_{step},
        \beta_{step}
    },
    \label{eq:demand-decomp-oracle-condition}
\end{equation}
the following hold.
\begin{enumerate}
    \item Algorithm~\ref{algo:expander_decomposition} always terminates,
    outputs a disjoint partition of $V(H)$, and satisfies
    $(\epsilon,\delta)$-edge-DP.
    \item Deterministically, every root-to-leaf path contains at most
    $D_{\max}$ private cut-finder invocations.
    \item With probability at least $1-\beta$, every returned component
    $V_i$ induces a $(\psi,\rho)$-expander, and the total weight of all
    inter-component edges is at most
    \begin{equation*}
        \psi\,\rho(V(H))^{1+o(1)}.
    \end{equation*}
\end{enumerate}
In particular, when $c_{exp},c_{size}=\rho(V(H))^{o(1)}$ and
$\rho(V(H))=\poly(n)$, we have
\[
    D_{\max}=\rho(V(H))^{o(1)}=n^{o(1)}
    \qquad\text{and}\qquad
    N_{\max}=n^{1+o(1)}.
\]
Thus, under the polynomial regularity assumed for $\Lambda_\rho$,
condition~\eqref{eq:demand-decomp-oracle-condition} has the asymptotic form
\begin{equation*}
    \psi \ge
    \widetilde{\Omega}\Paren{
        \Lambda_\rho\Paren{
            \frac{\epsilon}{n^{o(1)}},
            \frac{\delta}{n^{o(1)}},
            \frac{\beta}{n^{1+o(1)}}
        }
    }.
\end{equation*}
\end{lemma}

\begin{proof}
This proof is a demand-weighted analogue of the global expansion-decomposition
analysis of Nanongkai and Saranurak~\cite[Section~5.2]{nanongkai2017dynamic}.
The hard depth cap in Algorithm~\ref{algo:expander_decomposition} ensures
termination and privacy on every transcript. We first analyze utility on the
event that every invocation of $\mathsf{CutFinder}$ satisfies the guarantees
of Lemma~\ref{lem:cut_finder}. For a graph $J$ and a sparsity parameter $\varphi$, define
\begin{equation}
\operatorname{OPT}_\rho(J,\varphi)
:=
\max\left\{
    \rho(X):
    \begin{array}{l}
        \emptyset\neq X\subsetneq V(J),\\
        \rho(X)\le \rho(V(J)\setminus X),\\
        w_J(X,V(J)\setminus X)\le \varphi\rho(X)
    \end{array}
\right\},
\label{eq:demand-opt}
\end{equation}
where the maximum is defined to be zero when no such cut exists.

As in~\cite{nanongkai2017dynamic}, associate an auxiliary graph $I$ with each
recursive call. This graph is used only for analysis:
\begin{itemize}
    \item the root call is represented by $(H,H,1)$;
    \item a smaller-side child is represented by
    $(H'[S],H'[S],1)$;
    \item a larger-side child is represented by
    $(H'[V(H')\setminus S],I,l)$; and
    \item a level-increment child is represented by $(H',H',l+1)$.
\end{itemize}

We claim that every analytical call $(H',I,l)$ made on the success event
satisfies
\begin{equation}
    \operatorname{OPT}_\rho(I,2\psi_l)<\overline{s}_l.
    \label{eq:demand-opt-invariant}
\end{equation}
For $l=1$, every smaller side of a cut in $I$ has demand at most
$\rho(V(I))/2\le \+R/2<\overline{s}_1$. This proves the invariant at the root
and after every smaller-side reset. The invariant is unchanged along a
larger-side recursive call because both $I$ and $l$ remain unchanged.

It remains to consider a level increment. Such an increment occurs only when
\[
    \rho(S)<\frac{\overline{s}_{l+1}}{c_{size}}.
\]
By the choice of the level parameters,
\begin{equation}
    2\psi_{l+1}
    =
    \frac{\psi_l}{c_{exp}}.
    \label{eq:level-parameter-relation}
\end{equation}
The bicriteria maximality guarantee of Lemma~\ref{lem:cut_finder} therefore
implies
\[
    \operatorname{OPT}_\rho(H',2\psi_{l+1})
    =
    \operatorname{OPT}_\rho
    \Paren{H',\frac{\psi_l}{c_{exp}}}
    \le c_{size}\rho(S)
    <\overline{s}_{l+1}.
\]
This proves~\eqref{eq:demand-opt-invariant} inductively.

Because $\overline{s}_L\le 1$ and all nonempty demand masses are positive
integers, the invariant at level $L$ implies
\[
    \operatorname{OPT}_\rho(H',2\psi_L)=0
\]
whenever the analytical anchor is $I=H'$. If
$\mathsf{CutFinder}(H',\psi_L)$ returned a nonempty cut, its sparsity guarantee
would make that cut $\psi_L$-sparse, and hence also $2\psi_L$-sparse, a
contradiction. Thus, on the success event, no call advances beyond level $L$.

We next bound the depth of the recursion before the hard cap. Following
\cite{nanongkai2017dynamic}, call the edge to the smaller-side child a
\emph{left edge}, the edge to the larger-side child a \emph{right edge}, and
a level-increment edge a \emph{down edge}. Along any root-to-leaf path, there
are at most $\lceil\log_2\+R\rceil$ left edges, because every left edge
replaces the current graph by a subgraph of at most half its demand mass.
Between consecutive left edges there are at most $L$ down edges.

It remains to bound consecutive right edges. Fix a maximal sequence of right
edges at one level $l$, beginning immediately after the root, a left edge, or
a down edge. At the beginning of this sequence, the analytical anchor equals
the current graph; denote it by $H_1$. Write
\[
    H_{i+1}=H_i[V(H_i)\setminus S_i],
    \qquad
    X_j=\bigcup_{i=1}^j S_i,
    \qquad
    r=\rho(V(H_1)).
\]
Each cut in this sequence satisfies
\begin{equation}
    \rho(S_i)\ge
    \frac{\overline{s}_{l+1}}{c_{size}},
    \qquad
    \rho(S_i)\le\frac{\rho(V(H_i))}{2},
    \qquad
    w_{H_i}(S_i,V(H_i)\setminus S_i)
    \le\psi_l\rho(S_i).
    \label{eq:right-edge-properties}
\end{equation}
The sets $S_i$ are disjoint, and every edge crossing
$(X_j,V(H_1)\setminus X_j)$ is contained in one of the cuts removed so far.
Consequently,
\begin{equation}
    w_{H_1}(X_j,V(H_1)\setminus X_j)
    \le
    \sum_{i=1}^j
    w_{H_i}(S_i,V(H_i)\setminus S_i)
    \le
    \psi_l\rho(X_j).
    \label{eq:union-cut-bound}
\end{equation}

Set
\[
    a:=\frac{\overline{s}_{l+1}}{c_{size}},
    \qquad
    g:=c_{size}\+R^\sigma,
\]
so that $ga=\overline{s}_l$. Suppose first that
$r\le 3\overline{s}_l$. Then fewer than
$3c_{size}\+R^\sigma$ right edges are possible, since otherwise the disjoint
removed sets would have total demand at least
$3c_{size}\+R^\sigma a=3\overline{s}_l\ge r$, while the current right-side
component has positive demand.

Now suppose that $r>3\overline{s}_l$. If there were at least
$\lceil g\rceil$ right edges, let $j$ be the first index for which
\[
    x:=\rho(X_j)\ge\overline{s}_l,
    \qquad
    x^-:=\rho(X_{j-1})<\overline{s}_l.
\]
Since $S_j$ is the smaller side in $H_j$, we have
\[
    x
    =
    x^-+\rho(S_j)
    \le
    x^-+\frac{r-x^-}{2}
    =
    \frac{r+x^-}{2}
    <
    \frac{r+\overline{s}_l}{2}
    <
    \frac{2r}{3}.
\]
If $x\le r/2$, then~\eqref{eq:union-cut-bound} shows that $X_j$ is a
$2\psi_l$-sparse smaller side with demand at least $\overline{s}_l$.
If $x>r/2$, let $Y=V(H_1)\setminus X_j$. Then
\[
    \rho(Y)=r-x>\frac r3>\overline{s}_l
    \qquad\text{and}\qquad
    x\le 2\rho(Y).
\]
Using~\eqref{eq:union-cut-bound},
\[
    w_{H_1}(Y,V(H_1)\setminus Y)
    \le \psi_l x
    \le 2\psi_l\rho(Y).
\]
Thus, in either case, $H_1$ contains a $2\psi_l$-sparse smaller side of
demand at least $\overline{s}_l$, contradicting
the invariant~\eqref{eq:demand-opt-invariant}. Therefore every such sequence
contains fewer than
\[
    \Gamma=\left\lceil3c_{size}\+R^\sigma\right\rceil
\]
right edges.

There are at most $1+\lceil\log_2\+R\rceil$ blocks separated by left edges,
and each block contains at most $L$ level segments. Hence every
root-to-leaf path contains at most
\[
    (1+\lceil\log_2\+R\rceil)
    \bigl(1+L(1+\Gamma)\bigr)
    =
    D_{\max}
\]
recursive nodes at which a private oracle can be invoked. Therefore, on the
success event, the natural recursion terminates before the hard depth cap is
triggered.

We now establish expansion. On the success event, a non-singleton component
$H'=G[V_i]$ is returned only when
$\mathsf{CutFinder}(H',\psi_l)$ outputs $\emptyset$. By the contrapositive of
the bicriteria maximality guarantee in Lemma~\ref{lem:cut_finder}, no cut
$S'\subsetneq V_i$ satisfies
\[
    w_{H'}(S',V_i\setminus S')
    \le
    \frac{\psi_l}{c_{exp}}
    \min\{\rho(S'),\rho(V_i\setminus S')\}.
\]
By the definition of $\psi_l$,
\begin{equation}
    \frac{\psi_l}{c_{exp}}
    =
    \psi\,\widehat c_{exp}^{\,L-l}
    \ge \psi.
    \label{eq:leaf-expansion-threshold}
\end{equation}
It follows that $G[V_i]$ is a $(\psi,\rho)$-expander. Singleton components
satisfy this property trivially.

To bound the total weight of inter-component edges, consider every recursive
split returning a nonempty cut $S$ at level $l$. By Lemma~\ref{lem:cut_finder},
\[
    w_{H'}(S,V(H')\setminus S)\le\psi_l\rho(S)\le\psi_1\rho(S).
\]
Charge this weight to the demand units in the smaller side $S$. Any fixed
unit of demand can be charged at most $O(\log\+R)$ times, because after each
charge it lies in a recursive component whose total demand is at most half
that of its preceding component. Therefore,
\begin{align*}
    \sum_{\text{inter-component cuts}}
    w_{H'}(S,V(H')\setminus S)
    &\le
    O(\psi_1\+R\log\+R)\\
    &=
    O\Paren{
        \psi\+R\log\+R\,
        \widehat c_{exp}^{\,L}
    }\\
    &=
    \psi\+R^{1+o(1)}.
\end{align*}
The last equality follows from
$L=O(1+1/\sigma)$ and
\[
    \widehat c_{exp}^{\,L}
    =
    \exp\Paren{
        O\Paren{
            \sqrt{\log\+R\log\widehat c_{exp}}
            +\log\widehat c_{exp}
        }
    }
    =
    \+R^{o(1)}.
\]

It remains to prove privacy and bound the failure probability. The hard cap
ensures unconditionally that at most $D_{\max}$ private oracle calls occur
along any root-to-leaf path. Conditioned on any fixed transcript preceding a
recursion depth, the active recursive subproblems are induced on pairwise
disjoint vertex sets. Hence a single input-edge change affects at most one
active subproblem at that depth, and the calls at that depth satisfy
$(\epsilon_{step},\delta_{step})$-DP by parallel composition. Sequential
composition (Lemma~\ref{lem:composition}) across at most $D_{\max}$ depths gives
\[
    \Paren{
        D_{\max}\epsilon_{step},
        D_{\max}\delta_{step}
    }
    =
    (\epsilon,\delta).
\]
The hard-cap and level-cap return operations are post-processing and consume
no additional privacy budget. Thus the algorithm is
$(\epsilon,\delta)$-edge-DP on every transcript, independently of the
oracle's utility behavior.

At any fixed recursion depth, the active nonempty components have disjoint
vertex sets, so there are at most $|V(H)|$ oracle calls. The deterministic
total number of calls is therefore at most
\[
    N_{\max}=|V(H)|D_{\max}.
\]
Each adaptive invocation has conditional utility-failure probability at most
$\beta_{step}$. Independence is not required: a conditional union bound gives
\[
    \Prb[\text{some oracle invocation fails}]
    \le
    N_{\max}\beta_{step}
    =
    \beta.
\]
Finally, for every level $l$,
\[
    \psi_l\ge\psi_L=\psi c_{exp}.
\]
Condition~\eqref{eq:demand-decomp-oracle-condition} therefore implies
\[
    \psi_l
    \ge
    10c_{exp}
    \Lambda_\rho
    \Paren{\epsilon_{step},\delta_{step},\beta_{step}},
\]
so Lemma~\ref{lem:cut_finder} applies simultaneously to all calls on the
success event. Combining the preceding expansion, depth, and charging
arguments proves all claimed utility guarantees with probability at least
$1-\beta$.
\end{proof}

\subsection{Private Cut Approximation via Recursive Expander Decomposition}

Here, we synthesize the preceding mechanisms into a unified multi-round recursive sparsification algorithm. By repeatedly factoring out demand-aware expanders, the residual graph is driven to a highly sparse state parameterized completely by the spectral error $\Pi$ (in Lemma~\ref{lem:demand_oracle}) and the total number of iterative rounds $T$. The final guarantee of our framework is formalized as follows:

\begin{theorem} \label{thm:cut_main}
Assuming access to a private spectral primitive with error interface $\Pi(n, d, \epsilon, \delta)$ (Assumption \ref{assump:spectral}), for any $n$-node unweighted simple graph $G=(V, E)$, privacy parameters $(\epsilon, \delta)$, multiplicative error parameter $\gamma \in (0,1)$, and any positive integer $T \leq \text{poly}(n)$, there exists an $(\epsilon,\delta)$-edge-DP polynomial-time algorithm that outputs a non-negative weighted synthetic graph $\widetilde{G}$ such that, with high probability, for all cuts $C \subseteq V$:
\begin{equation}
    \abs{w_G(C) - w_{\widetilde{G}}(C)} \le \gamma w_G(C) + \widetilde{O}\Paren{ \frac{n}{\epsilon} + \left(\frac{n^2 M_T}{\epsilon^2 \gamma}\right)^{1/3} },
\end{equation}
where the edge count $M_T$ is obtained by iterating the dynamical system
$$
M_{t+1} = \widetilde{O}\left( \frac{n^{1+o(1)}}{\gamma} \, \Pi\left(n, \max\left\{1, \frac{M_t}{n}\right\}, \frac{\epsilon}{T}, \frac{\delta}{T}\right) T \right),
$$
for $t = 0, 1, \dots, T-1$, with initial condition $M_0 = n^2$.
\end{theorem}
Before we prove Theorem~\ref{thm:cut_main}, to better formalize the recursive framework, we recall all mechanisms or modular subroutines developed in the preceding sections. Specifically, 
\begin{enumerate}
    \item We denote the subroutine for generating the demand vector $\rho$ from Lemma \ref{lem:demand_allocation} as $\mathsf{BuildDemand}$, and the demand-sensitive expander decomposition (i.e., Algorithm~\ref{algo:expander_decomposition}) from Lemma \ref{lem:demand_ns17} as $\mathsf{DemandExpDecomp}$, which includes $\mathsf{CutFinder}$ (Algorithm~\ref{algo:cut_finder}) as a subroutine.
    \item  Crucially, we formalize the private synthesis process established in Lemma \ref{lem:demand_oracle} as the subroutine $\mathsf{SyntheticOracle}$, which encompasses the local port graph construction and executing the generic private spectral primitive (Assumption~\ref{assump:spectral}).
    \item  We also denote the $\widetilde{O}(n + (n^2M)^{1/3})$ additive error cut release oracle in Theorem~\ref{thm:cut-oracle-main} (Section~\ref{sec:edge-sensitive-cut-oracle}) as $\mathsf{CutRelease}(G,M, \eps,\del)$.
\end{enumerate}
 With these primitives defined, we present the main recursive sparsification in Algorithm \ref{algo:main_recursive_cut}.

\begin{algorithm}[h]
\caption{Recursive Demand-Private Cut Approximation}
\label{algo:main_recursive_cut}
\KwIn{Unweighted graph $G=(V,E)$, privacy budgets $(\epsilon,\delta)$, failure probability $\beta$, multiplicative error $\gamma$, recursion rounds $T$.}
Set $R_0 \gets G$, $M_0 \gets n^2$ and denote by $c_{exp} = \log^c n$ (Lemma~\ref{lem:cut_finder})\;
Set $\epsilon_{round} \gets \epsilon/(8T)$, $\delta_{round} \gets \delta/(8T)$, $\beta_{round} \gets \beta / (8T n^2)$\;

\For{$t = 0$ \KwTo $T-1$}{
    \tcp{Phase 1: Public Demand Generation (Lemma \ref{lem:demand_allocation})}
    $\rho_t \gets \mathsf{BuildDemand}(R_t, M_t, \epsilon_{round}, \beta_{round})$ \\
    
    \tcp{Phase 2: Error and Threshold Calibration}
    $\Lambda_t \gets \widetilde{O} \Paren{ \Pi(n, M_t/n, \epsilon_{round}, \delta_{round}) }$\; 
    $\gamma_t \gets \gamma/(2T)$ \;
    $\psi_t \gets \Lambda_t c_{exp} n^{o(1)} / \gamma_t$ 
    
\tcp{Phase 3: Recursive Expander Decomposition (Lemma \ref{lem:demand_ns17})}
    $\+P_t \gets \mathsf{DemandExpDecomp}(R_t, \rho_t, \psi_t, \epsilon_{round}, \delta_{round}, \beta_{round})$ \\
    
    \tcp{Phase 4: Local Oracle Releases (Parallel Composition)}
    \For{\textbf{each} internal component $U \in \+P_t$}{
        $\widehat{H}_{t,U} \gets \mathsf{SyntheticOracle}(R_t[U], \rho_t, \epsilon_{round}, \delta_{round}, \beta_{round})$ 
    }
    
    \tcp{Phase 5: Graph Deflation}
    $R_{t+1} \gets \{ e \in E(R_t) : e \text{ crosses different components in } \+P_t \}$ \\
    $M_{t+1} \gets \psi_t \rho_t(V)^{1+o(1)}$
}

\tcp{Phase 6: Terminal Sparse Graph Release}
$H_{\mathsf{sparse}} \gets \mathsf{CutRelease}(R_T, M_T, \eps/2, \del/2, \beta/2)$ \;
\Return $\widetilde{G} = H_{\mathsf{sparse}} + \sum_{t=0}^{T-1} \sum_{U \in \+P_t} \widehat{H}_{t,U}$
\end{algorithm}

\begin{proof}
(Of Theorem~\ref{thm:cut_main}.)
We analyze the utility, dynamical recurrence, privacy, and probability bounds consecutively to establish Theorem~\ref{thm:cut_main}. For any cut $C \subseteq V$, every original unweighted edge $e \in E$ is either absorbed inside some component $U \in \+P_t$ at a specific round $t$, or it survives all $T$ rounds to become a residual edge in $R_T$. Therefore, we can decompose the cut value of $C$ (in the original graph) as follows:
\begin{equation} \label{eq:cut_decomposition_true}
    w_G(C) = \sum_{t=0}^{T-1} \sum_{U \in \+P_t} w_{R_t[U]}(C \cap U, U \setminus C) + w_{R_T}(C).
\end{equation}
On the other hand, by the definition of the out $\widetilde{G}$, the estimated cut size evaluates to:
\begin{equation} \label{eq:cut_decomposition_syn}
    w_{\widetilde{G}}(C) = \sum_{t=0}^{T-1} \sum_{U \in \+P_t} w_{\widehat{H}_{t,U}}(C \cap U, U \setminus C) + w_{\widehat{R}_T}(C),
\end{equation}
where $\widehat{R}_T$ is the output of the terminal release mechanism.

For any internal component $U \in \+P_t$, according to Lemma~\ref{lem:demand_ns17}, the subgraph $R_t[U]$ is an $(\psi_t, \rho_t)$-expander (\cref{eq:expansion_definition}). Invoking the demand-sensitive oracle error (Lemma \ref{lem:demand_oracle}) and substituting the expansion bound $w_{R_t[U]} \ge \psi_t \min\{\rho_t(C \cap U), \rho_t(U \setminus C)\}$, we bound the local absolute deviation. To leave room for the terminal relative error, we configure the internal spectral release to operate at target accuracy $\gamma_{internal} = \gamma/2$:
\begin{align} \label{eq:local_absorption}
    \abs{ w_{R_t[U]} - w_{\widehat{H}_{t,U}} } &\le \Lambda_t \min\{\rho_t(C \cap U), \rho_t(U \setminus C)\} \nonumber \\
    &\le \frac{\Lambda_t}{\psi_t} w_{R_t[U]}(C \cap U, U \setminus C) \le \frac{\gamma}{2} w_{R_t[U]}(C \cap U, U \setminus C).
\end{align}
Summing Eq. \eqref{eq:local_absorption} across all components $U$ and rounds $t$, the aggregated error over all internal releases is then bounded by:
\begin{align}\label{eq:internal_multiplicative_bound}
    \sum_{t=0}^{T-1} \sum_{U \in \+P_t} \abs{ w_{R_t[U]} - w_{\widehat{H}_{t,U}} } &\le \frac{\gamma}{2} \sum_{t=0}^{T-1} \sum_{U \in \+P_t} w_{R_t[U]}(C \cap U, U \setminus C) \nonumber \\
    &\le \frac{\gamma}{2} w_G(C).
\end{align}

To determine the final error of the terminal residual graph $R_T$, we map the trajectory of the residual edge sequence $M_t$. Utilizing the fact $\rho_t(V) \le \widetilde{O}(n)$ (Lemma~\ref{lem:demand_allocation}) and substituting $\psi_t$, the recurrence across $T$ rounds is governed by the spectral error interface $\Pi$ evaluated at the deep recursion budgets $\epsilon_{step} = \epsilon/(8T n^{o(1)})$ and $\delta_{step} = \delta/(8T n^{o(1)})$. Crucially, the expander decomposition incurs an $n^{o(1)}$ penalty on both the residual edges and the privacy parameters (Lemma \ref{lem:demand_ns17}). Invoking Assumption~\ref{assump:spectral} again, the $n^{o(1)}$ shrinkage in privacy parameters inflates the error $\Pi$ by at most an $n^{o(1)}$ multiplicative overhead. We explicitly extract and incorporate this penalty:
\begin{equation*}
    M_{t+1} = \psi_t \rho_t(V)^{1+o(1)} \le \widetilde{O} \Paren{ \Pi\left(n, \max\left\{1, \frac{M_t}{n}\right\}, \frac{\epsilon}{T}, \frac{\delta}{T}\right) n^{1+o(1)} \gamma^{-1} T }.
\end{equation*}

Starting from the initial dense graph configuration $M_0 = n^2$, this recursive process produces the terminal residual edge count $M_T$ after $T$ iterations. We then process $R_T$ via the relative terminal release mechanism $\mathsf{CutOracle}$ (Theorem \ref{thm:cut-oracle-main}). Allocating privacy budgets $\epsilon_{terminal} = \epsilon/2$ and $\delta_{terminal} = \delta/2$, and setting the terminal multiplicative error parameter to $\gamma_{oracle} = \gamma/2$, the oracle guarantees:
\begin{equation} \label{eq:terminal_error}
    \abs{ w_{R_T}(C) - w_{\widehat{R}_T}(C) } \le \frac{\gamma}{2} w_{R_T}(C) + \widetilde{O} \Paren{ \frac{n}{\epsilon/2} + \left( \frac{n^2 M_T}{(\epsilon/2)^2 (\gamma/2)} \right)^{1/3} }.
\end{equation}
Finally, the additive bias simplifies to $\alpha_T = \widetilde{O} ({ {n}/{\epsilon} + \left( {n^2 M_T}/{(\epsilon^2 \gamma)} \right)^{1/3} })$. 

Combining the internal absorption \eqref{eq:internal_multiplicative_bound} with the terminal deviation \eqref{eq:terminal_error}, and exploiting the fact that $w_{R_T}(C) \le w_G(C)$ since $R_T$ is a subgraph of $G$, we bound the global additive error for any cut $C$:
\begin{align*}
    \abs{ w_G(C) - w_{\widetilde{G}}(C) } &\le \frac{\gamma}{2} w_G(C) + \left( \frac{\gamma}{2} w_{R_T}(C) + \alpha_T \right) \\
    &\le \gamma w_G(C) + \widetilde{O} \Paren{ \frac{n}{\epsilon} + \left(\frac{n^2 M_T}{\epsilon^2 \gamma}\right)^{1/3} },
\end{align*}
which concludes the utility proof in Theorem \ref{thm:cut_main}.

The overall differential privacy holds under adaptive and parallel composition. The total privacy budget $\epsilon$ is partitioned into four allocations: $\epsilon/2$ dedicated to the terminal $\mathsf{CutOracle}$ algorithm, and the remaining $\epsilon/2$ distributed evenly across the $T$ rounds, allocating $4\epsilon_{round} = 4\epsilon/(8T)$ per round. Within any round $t$, by Lemma~\ref{lem:stable_gadget}, the gadget construction strictly limits $\ell_1$-sensitivity to $O(1)$, consuming $\epsilon_{round}$ under the Laplace mechanism. The recursion steps in Algorithm~\ref{algo:expander_decomposition} isolates original edges to sub-problems of depth bounded by $n^{o(1)}$. By applying basic composition over this $n^{o(1)}$ depth, we allocate a local privacy budget of $\epsilon_{step} = \epsilon_{round}/n^{o(1)}$ to each decomposition step, thereby collectively consuming $\epsilon_{round}$ globally across the depth under basic composition. The subsequent internal spectral releases operate on the mutually disjoint vertex subsets $U \in \+P_t$; thus, via parallel composition, they collectively consume only $\epsilon_{round}$ globally across the round. The allocation for budget $\del$ over each round is identical. By basic composition over $T$ rounds and the final step with $\mathsf{CutOracle}$, the algorithm preserves $(\epsilon,\delta)$-DP.

Finally, we bound the global failure probability. Failure events encompass degree noise violations, hash collisions exceeding bounds, spectral primitive aborts, and terminal oracle failures. In any given round $t$, the recursive expander decomposition tree produces at most $n^{1+o(1)}$ internal components. Consequently, there are at most $n^{1+o(1)} \le n^2$ local invocations of the hash mappings, spectral oracles, and cut-finders per round. By distributing the global failure probability $\beta$ and allocating $\beta_{round} = \beta / (8T n^2)$ to each individual sub-routine invocation across all $T$ rounds (and allocating $\beta/2$ to the final terminal release), the sum of all failure probabilities over the entire pipeline is bounded by $T \cdot n^2 \cdot \beta_{round} + \beta/2 \le \beta/8 + \beta/2 < \beta$, completing the proof.
\end{proof}

With the generic reduction established in Theorem~\ref{thm:cut_main}, we can now plug in the bootstrapped fourth-power spectral primitive (Theorem~\ref{thm:improved-spectral}) developed in Section~\ref{sec:bootstrapped_fourth_power} to break the $O(n^{1.25})$ additive error barrier. 

\begin{lemma}[A square-root contraction recurrence]
\label{lem:sqrt-contraction}
Let \(x_0,b,c\ge 0\), and suppose that
\[
    x_{t+1}
    \le
    \frac14 x_t+b\sqrt{x_t}+c
\]
for every \(t\ge 0\). Then
\[
    x_t
    \le
    2^{-t}x_0+2(b^2+c)
\]
for every \(t\ge 0\).
\end{lemma}

\begin{proof}
For every \(x\ge 0\), it is easy to verify that
$
    b\sqrt{x}
    \le
    \frac14 x+b^2.
$ Consequently, $x_{t+1}
    \le
    \frac12x_t+b^2+c.$ Iterating this inequality yields
\[
    x_t
    \le
    2^{-t}x_0
    +(b^2+c)\sum_{j=0}^{t-1}2^{-j}
    \le
    2^{-t}x_0+2(b^2+c),
\]
as claimed.
\end{proof}
The following corollary utilize Lemma~\ref{lem:sqrt-contraction} to solve the recursive dynamical system driving by our spectral primitive (Theorem~\ref{thm:improved-spectral}) to its optimal fixed point.
\begin{corollary}
[Private Cut Approximation with
\(\widetilde{O}(n^{13/12+o(1)})\) Additive Error]
\label{cor:cut_13_12}
For any \(n\)-node unweighted simple graph \(G=(V,E)\), privacy
parameters \(\epsilon\in(0,1)\), $\del\in (0,1/2)$ and multiplicative error
parameter \(\gamma\in(0,1/4)\), there exists an
\((\epsilon,\delta)\)-edge-DP polynomial-time algorithm that outputs a
non-negative weighted synthetic graph \(\widetilde G\) such that, with
high probability, simultaneously for all cuts \(C\subseteq V\),
\[
    \abs{w_G(C)-w_{\widetilde G}(C)}
    \le
    \gamma w_G(C)
    +
    \widetilde O_{\delta}\Paren{
        \frac{n^{13/12+o(1)}}
             {\epsilon\gamma^{7/6}}
    }.
\]
\end{corollary}

\begin{proof}
Define
\[
    \bar\gamma:=\min\{\gamma,1/8\},
    \qquad
    \bar\delta:=\min\{\delta,1/4\}.
\]
It suffices to construct an \((\epsilon,\bar\delta)\)-edge-DP
mechanism with multiplicative error \(\bar\gamma\). Indeed,
\(\bar\delta\le\delta\), \(\bar\gamma\le\gamma\), and $
    \bar\gamma^{-7/6}
    =
    O(\gamma^{-7/6}).
$ Assume \(n\ge 2\), since the claim is immediate for \(n=1\), and fix
the number of recursive sparsification rounds in advance as $T:=\left\lceil\log_2 n\right\rceil.$ We apply Theorem~\ref{thm:cut_main} with parameters
\((\epsilon,\bar\delta,\bar\gamma,T)\) and instantiate its spectral
primitive using Theorem~\ref{thm:improved-spectral}. For every \(t\), define
\[
    d_t:=\max\left\{1,\frac{M_t}{n}\right\}.
\]
In particular, \(d_0=n\). Under privacy parameters
\((\epsilon/T,\bar\delta/T)\), Theorem~\ref{thm:improved-spectral}
gives
\[
\begin{split}
    \Pi\left(
        n,d_t,\frac{\epsilon}{T},\frac{\bar\delta}{T}
    \right)
    =
    \widetilde O_{\bar\delta}\left(
        \frac{d_t\sqrt T}{\sqrt{\epsilon q}}
        +
        \frac{n^{1/8}d_t^{1/2}q^{1/8}T^{3/8}}
             {\epsilon^{3/8}}
        +
        \frac{n^{1/4}q^{1/4}T^{3/4}}
             {\epsilon^{3/4}}
        +
        \frac{T}{\epsilon}
    \right).
\end{split}
\]
The recurrence in Theorem~\ref{thm:cut_main} contains one additional
multiplicative factor \(T\). Therefore, after dividing the residual
edge bound by \(n\), we obtain
\begin{equation}
    d_{t+1}
    \le
    1+
    \widetilde O_{\bar\delta}\Paren{
        \frac{n^{o(1)}}{\bar\gamma}
    }
    \Bigg[
        \frac{T^{3/2}d_t}{\sqrt{\epsilon q}}
        +
        \frac{
            T^{11/8}n^{1/8}d_t^{1/2}q^{1/8}
        }{\epsilon^{3/8}}
        +
        \frac{
            T^{7/4}n^{1/4}q^{1/4}
        }{\epsilon^{3/4}}
        +
        \frac{T^2}{\epsilon}
    \Bigg].
    \label{eq:degree-recurrence-with-T}
\end{equation}
The additive \(1\) comes from the
\(d_{t+1}=\max\{1,M_{t+1}/n\}\). Because \(T=O(\log n)\), all powers of \(T\), all polylogarithmic
losses, and all \(n^{o(1)}\) losses in
\eqref{eq:degree-recurrence-with-T} can be bounded by a public factor
\[
    \mathcal L
    =
    \widetilde O_{\bar\delta}\Paren{n^{o(1)}},
    \qquad
    \mathcal L\ge 1.
\]
We enlarge \(\mathcal L\), if necessary, by an additional
polylogarithmic factor so that the choice of \(q\) below also satisfies
the lower-bound requirement on \(q\) in
Theorem~\ref{thm:improved-spectral}. Thus
\eqref{eq:degree-recurrence-with-T} implies
\begin{equation}
\label{eq:simplified-degree-recurrence}
    d_{t+1}
    \le
    1+
    B\left[
        \frac{d_t}{\sqrt{\epsilon q}}
        +
        \frac{n^{1/8}d_t^{1/2}q^{1/8}}
             {\epsilon^{3/8}}
        +
        \frac{n^{1/4}q^{1/4}}
             {\epsilon^{3/4}}
        +
        \frac1\epsilon
    \right],
    \qquad
    B:=\frac{\mathcal L}{\bar\gamma}.
\end{equation}
We now choose the tunable parameter as:
\begin{equation}
\label{eq:q-choice-cut-corollary}
    q
    :=
    \left\lceil
        \frac{16B^2}{\epsilon}
    \right\rceil.
\end{equation}
In particular, $q
    =
    \widetilde O_{\bar\delta}\Paren{
        \frac{n^{o(1)}}{\epsilon\bar\gamma^2}
    }.$ hus, \(q\) is not asserted to be just polylogarithmic in
\(1/\epsilon\) and \(1/\bar\gamma\). Its polynomial dependence on
these parameters is displayed explicitly. By \cref{eq:q-choice-cut-corollary}, we have $
    \frac{B}{\sqrt{\epsilon q}}
    \le
    \frac14.$ Define
\[
    b
    :=
    \frac{
        Bn^{1/8}q^{1/8}
    }{\epsilon^{3/8}}
\]
and
\[
    c
    :=
    1+
    \frac{
        Bn^{1/4}q^{1/4}
    }{\epsilon^{3/4}}
    +
    \frac{B}{\epsilon}.
\]
Then \eqref{eq:simplified-degree-recurrence} becomes
\[
    d_{t+1}
    \le
    \frac14d_t+b\sqrt{d_t}+c.
\]
Applying Lemma~\ref{lem:sqrt-contraction} gives
\begin{equation}
\label{eq:degree-after-t-rounds}
    d_t
    \le
    2^{-t}d_0+2(b^2+c).
\end{equation}
It remains to bound \(b^2+c\). Since \(B\ge1\) and
\(\epsilon\le1\), the choice
\eqref{eq:q-choice-cut-corollary} gives $q
    \le {17B^2}/{\epsilon}.$
Consequently,
\begin{align*}
    b^2
    =
    \frac{
        B^2n^{1/4}q^{1/4}
    }{\epsilon^{3/4}}
    =
    O\left(
        \frac{
            n^{1/4}B^{5/2}
        }{\epsilon}
    \right),
\end{align*}
and similarly,
\begin{align*}
    \frac{
        Bn^{1/4}q^{1/4}
    }{\epsilon^{3/4}}
    =
    O\left(
        \frac{
            n^{1/4}B^{3/2}
        }{\epsilon}
    \right)
    \le
    O\left(
        \frac{
            n^{1/4}B^{5/2}
        }{\epsilon}
    \right).
\end{align*}
Moreover,
\[
    1+\frac{B}{\epsilon}
    \le
    O\left(
        \frac{
            n^{1/4}B^{5/2}
        }{\epsilon}
    \right).
\]
Therefore,
\begin{equation}
\label{eq:b-c-bound}
    b^2+c
    =
    O\left(
        \frac{
            n^{1/4}B^{5/2}
        }{\epsilon}
    \right).
\end{equation}
Because \(T=\lceil\log_2 n\rceil\), we have
\(2^{-T}d_0\le1\). Combining
\eqref{eq:degree-after-t-rounds} and
\eqref{eq:b-c-bound} yields
\begin{align}
    d_T
    =
    O\left(
        \frac{
            n^{1/4}B^{5/2}
        }{\epsilon}
    \right)
    \nonumber
    =
    O\left(
        \frac{
            n^{1/4}\mathcal L^{5/2}
        }{
            \epsilon\bar\gamma^{5/2}
        }
    \right)
    \nonumber
    =
    \widetilde O_{\bar\delta}\Paren{
        \frac{
            n^{1/4+o(1)}
        }{
            \epsilon\bar\gamma^{5/2}
        }
    }.
    \label{eq:terminal-degree-bound}
\end{align}
Here we used
\(\mathcal L^{5/2}
=\widetilde O_{\bar\delta}(n^{o(1)})\). Since \(M_T\le nd_T\), it follows that
\begin{equation}
\label{eq:terminal-edge-bound}
    M_T
    =
    \widetilde O_{\bar\delta}\Paren{
        \frac{
            n^{5/4+o(1)}
        }{
            \epsilon\bar\gamma^{5/2}
        }
    }.
\end{equation}
Finally, Theorem~\ref{thm:cut_main}, instantiated with
multiplicative parameter \(\bar\gamma\), gives additive error
\[
    \widetilde O_{\bar\delta}\Paren{
        \frac n\epsilon
        +
        \left(
            \frac{
                n^2M_T
            }{
                \epsilon^2\bar\gamma
            }
        \right)^{1/3}
    }.
\]
Substituting \eqref{eq:terminal-edge-bound},
\begin{align*}
    \left(
        \frac{
            n^2M_T
        }{
            \epsilon^2\bar\gamma
        }
    \right)^{1/3}
    =
    \widetilde O_{\bar\delta}\left(
        \left(
            \frac{
                n^{2}
                n^{5/4+o(1)}
            }{
                \epsilon^3
                \bar\gamma^{5/2+1}
            }
        \right)^{1/3}
    \right)
    =
    \widetilde O_{\bar\delta}\Paren{
        \frac{
            n^{13/12+o(1)}
        }{
            \epsilon\bar\gamma^{7/6}
        }
    }.
\end{align*}
Because \(n^{13/12}\ge n\), \(\bar\gamma\le1\), and
\(\epsilon\le1\), this term also absorbs the \(n/\epsilon\) term.
Thus the algorithm satisfies
\[
    \abs{w_G(C)-w_{\widetilde G}(C)}
    \le
    \bar\gamma w_G(C)
    +
    \widetilde O_{\bar\delta}\Paren{
        \frac{
            n^{13/12+o(1)}
        }{
            \epsilon\bar\gamma^{7/6}
        }
    }
\]
simultaneously for all cuts \(C\).

Finally, \(\bar\gamma\le\gamma\),
\(\bar\gamma^{-7/6}=O(\gamma^{-7/6})\), and
\(\widetilde O_{\bar\delta}=\widetilde O_\delta\). This gives
\[
    \abs{w_G(C)-w_{\widetilde G}(C)}
    \le
    \gamma w_G(C)
    +
    \widetilde O_{\delta}\Paren{
        \frac{
            n^{13/12+o(1)}
        }{
            \epsilon\gamma^{7/6}
        }
    },
\]
as required. The algorithm is polynomial-time because
\(T=O(\log n)\), the chosen value of \(q\) is polynomial in the
displayed accuracy and privacy parameters up to an \(n^{o(1)}\)
factor, and both Theorem~\ref{thm:cut_main} and
Theorem~\ref{thm:improved-spectral} provide polynomial-time
implementations.
\end{proof}

\section{Hardness of Graph Spectral Estimation with Approximate DP}\label{sec:lowerbound}
Here we prove that, under edge-level differential privacy, any algorithm for estimating the Laplacian of an unweighted graph on $n$ vertices must incur a spectral error of $\Omega(\sqrt{n})$. This matches the trivial $\widetilde{O}(\sqrt{n})$ upper bound, showing that the improvement to $\widetilde{O}((nd)^{1/4})$ achieved by our private spectral primitive (Theorem~\ref{thm:laplacian-square}) is not an artifact of loose analysis but rather an intrinsic barrier: the trivial $\widetilde{O}(\sqrt{n})$ additive error is, in fact, asymptotically optimal in the worst case.

\begin{theorem}
    [Lower bound for spectral release]\label{thm:spectral-lower-bound}
Fix a constant $\eps>0$. There exists a constant $c=c(\eps)>0$ such that the following holds for all sufficiently large $n$: let $\delta \leq\frac{1}{8(1+e^\eps)}$, for any $(\eps,\delta)$-edge-DP algorithm that releases matrix\footnote{We note that the theorem does not restrict that the output is actually a graph. Therefore the lower bound applies to arbitrary symmetric matrix releases.} $Y$ on input graph $G$, there exists some connected graph $G$ with $m=\Theta(n)$ and maximum degree $d = \Theta(n)$, the algorithm must have
\[
   \Prb[\opnorm{Y-L_G} \geq c\sqrt n] > 0.49.
\]
\end{theorem}

\subsection{The Hard Family and Reconstruction}
\label{sec:hard_family}
To prove this theorem, we first introduce the construction of the family of hard instances. Let the vertex set be $V=\{0,1,2,\ldots,n-1\}$. Each hard instance consists of a fixed, public path 
\[
0-1-2-\cdots-(n-1)
\] and a private $k$-star centered $0$. Formally, let $U=\{2,3,\ldots,n-1\}$. For every set of vertices $S\subset U$ of size $k$, we define 
\[
G_S
   =
   B\cup \bigl\{\{0,i\}:i\in S\bigr\},
\]
where $B$ denotes the edge set of the public path. Note that every graph $G_S$ is simple and connected. It has exactly $m=(n-1)+k$ edges: $n-1$ path edges and $k$ private star edges. We denote the set of all possible $G_S$ as the family $\F_{n,k}$.

The next lemma demonstrates how spectral approximation exposes private edges. Specifically, for any two graphs $G_S$ and $G_T$ from this family, the symmetric difference between their edge sets can be bounded by their spectral distance. Since both graphs share the identical public path, this edge difference stems solely from their private stars.

\begin{lemma}[Spectral separation]\label{lem:spectral-separation}
For any two $k$-subsets $S,T\subseteq U$,
\[
   \opnorm{L_{G_S}-L_{G_T}}
   \geq
   \sqrt{|S\symdiff T|}.
\]
\end{lemma}

\begin{proof}
If $S=T$, the claimed inequality is trivial. Thus we assume $S\neq T$. 
Let $a_i=\1_S(i)-\1_T(i)$, for every $i\in U$. By $S\neq T$ we have $a\neq 0$.
Since $|S|=|T|=k$,
\[
   \sum_{i\in U}a_i=0.
\]

Also $a_i\in\{-1,0,1\}$, therefore
\[
   \norm{a}_2^2=\sum_{i\in U}a_i^2=|S\symdiff T|.
\]

The path $B$ is common to both graphs, thus we have:
\[
   L_{G_S}-L_{G_T}
   =
   \sum_{i\in U}a_iL_{\{0,i\}}.
\]

Define a vector $z\in\R^n$ by
\[
   z_0=\frac{1}{\sqrt 2},
   \qquad
   z_i=-\frac{a_i}{\sqrt 2\norm{a}_2}\quad (i\in U),
   \qquad
   z_1=0.
\]

Note $z$ is a unit vector since
\[
   \norm{z}_2^2
   =
   \frac12+\frac{1}{2\norm{a}_2^2}\sum_{i\in U}a_i^2
   =1.
\]

For an edge $\{0,i\}$, $z^\top L_{\{0,i\}}z=(z_0-z_i)^2$. Hence
\begin{align*}
   z^\top(L_{G_S}-L_{G_T})z
   &=
   \sum_{i\in U}a_i(z_0-z_i)^2 \\
   &=
   \frac12\sum_{i\in U}a_i
      \left(1+\frac{a_i}{\norm{a}_2}\right)^2 \\
   &=
   \frac12\sum_{i\in U}a_i
      +\frac{1}{\norm{a}_2}\sum_{i\in U}a_i^2
      +\frac{1}{2\norm{a}_2^2}\sum_{i\in U}a_i^3 .
\end{align*}

Because $a_i\in\{-1,0,1\}$, we have $a_i^3=a_i$.  Since $\sum_i a_i=0$, both the first and third terms vanish.  Therefore
\[
   z^\top(L_{G_S}-L_{G_T})z
   =
   \norm{a}_2
   =
   \sqrt{|S\symdiff T|}.
\]

Since $z$ is a unit vector,
\[
   \opnorm{L_{G_S}-L_{G_T}}
   \geq
   \left|z^\top(L_{G_S}-L_{G_T})z\right|
   =
   \sqrt{|S\symdiff T|}.
\]
\end{proof}

The next lemma shows that the private star of the input graph can be partially recovered, provided we are given a matrix that spectrally approximates the graph with a small error. 

Suppose an algorithm outputs a matrix $Y$.  Define a nearest-neighbor decoder
\[
   \widehat S(Y)
   \in
   \arg\min_{T\subseteq U,\ |T|=k}
   \opnorm{Y-L_{G_T}}.
\]
Ties are broken arbitrarily.  This decoder is not required to be computationally efficient.  Lower bounds may use arbitrary post-processing of the algorithm output.

\begin{lemma}\label{lem:decoding}
If $S\subseteq U$, $|S|=k$, and
\[
   \opnorm{Y-L_{G_S}}\leq \alpha,
\]
then
\[
   |S\symdiff \widehat S(Y)|\leq 4\alpha^2.
\]
\end{lemma}

\begin{proof}
By the definition of $\widehat S(Y)$,
\[
   \opnorm{Y-L_{G_{\widehat S(Y)}}}
   \leq
   \opnorm{Y-L_{G_S}}
   \leq \alpha.
\]
Therefore, by the triangle inequality,
\[
   \opnorm{L_{G_S}-L_{G_{\widehat S(Y)}}}
   \leq
   \opnorm{L_{G_S}-Y}
   +
   \opnorm{Y-L_{G_{\widehat S(Y)}}}
   \leq 2\alpha.
\]
Lemma~\ref{lem:spectral-separation} gives
\[
   \sqrt{|S\symdiff \widehat S(Y)|}
   \leq
   \opnorm{L_{G_S}-L_{G_{\widehat S(Y)}}}
   \leq 2\alpha.
\]
Squaring yields the claim.
\end{proof}

The next corollary converts a high-probability spectral guarantee into an average reconstruction guarantee.

\begin{corollary}\label{cor:expected-reconstruction}
Let $S$ be uniformly random among all $k$-subsets of $U$, let $Y=A(G_S)$, and let $\widehat S=\widehat S(Y)$.  Suppose that, for every $S$,
\[
   \Prb\bigl[\opnorm{Y-L_{G_S}}\leq \alpha\bigr]
   \geq 1-\beta.
\]
Then
\[
   \E\bigl[|S\symdiff \widehat S|\bigr]
   \leq
   4\alpha^2+2\beta k.
\]
Moreover, since $|S|=|\widehat S|=k$,
\[
   \E\bigl[|S\setminus \widehat S|\bigr]
   =
   \E\bigl[|\widehat S\setminus S|\bigr]
   \leq
   2\alpha^2+\beta k.
\]
\end{corollary}

\begin{proof}
On the success event $\opnorm{Y-L_{G_S}}\leq \alpha$, Lemma~\ref{lem:decoding} gives $|S\symdiff \widehat S|\leq 4\alpha^2.$
On the failure event, we use the trivial bound $|S\symdiff \widehat S|\leq 2k$, because both sets have size $k$.  Since the failure probability is at most $\beta$ for every fixed $S$, it is also at most $\beta$ when $S$ is uniform.  Hence
\[
   \E\bigl[|S\symdiff \widehat S|\bigr]
   \leq
   4\alpha^2+2\beta k.
\]
Finally, since $|S|=|\widehat S|=k$, the number of false negatives equals the number of false positives:
\[
   |S\setminus \widehat S|=|\widehat S\setminus S|=\frac12 |S\symdiff \widehat S|.
\]
This yields the second claim.
\end{proof}

\subsection{Reconstruction violates Approximate-DP}

The previous section says that small spectral error reconstructs the hidden star support.  We now show that such reconstruction violates $(\eps,\delta)$-DP. We first record the two-edge version of differential privacy.

\begin{lemma}\label{lem:two-edge}
Let $\mathcal{A}$ be $(\eps,\delta)$-edge-DP.  If two graphs $G,H$ differ in at most two edge edits, then for every event $\mathcal E$,
\[
   \Prb[\mathcal{A}(G)\in \mathcal E]
   \leq
   e^{2\eps}\Prb[\mathcal{A}(H)\in \mathcal E]
   +(1+e^{\eps})\delta.
\]
\end{lemma}

\begin{proof}
Choose an intermediate graph $G'$ such that $G$ is edge-neighboring to $G'$ and $G'$ is edge-neighboring to $H$.  Applying DP twice gives
\begin{align*}
   \Prb[\mathcal{A}(G)\in \mathcal E]
   &\leq
   e^{\eps}\Prb[\mathcal{A}(G')\in \mathcal E]+\delta \\
   &\leq
   e^{\eps}\left(e^{\eps}\Prb[\mathcal{A}(H)\in \mathcal E]+\delta\right)+\delta \\
   &=
   e^{2\eps}\Prb[\mathcal{A}(H)\in \mathcal E]
   +(1+e^{\eps})\delta.
\end{align*}
\end{proof}

For a fixed vertex $i\in U$, define two distributions over private sets:
\[
   \mathcal D_i^1=\text{uniform over all }k\text{-subsets }S\subseteq U\text{ with }i\in S,
\]
\[
   \mathcal D_i^0=\text{uniform over all }k\text{-subsets }S\subseteq U\text{ with }i\notin S.
\]
The distributions $\mathcal D_i^1$ and $\mathcal D_i^0$ differ only by the membership of the single private edge $\{0,i\}$.  More precisely, they can be coupled so that the corresponding graphs differ in two edge edits: one deletes a star edge $\{0,j\}$ and adds the other star edge $\{0,i\}$.

\begin{lemma}\label{lem:coordinate-comparison}
Let $\mathcal{A}$ be $(\eps,\delta)$-edge-DP.  Fix $i\in U$ and let $\mathcal E_i$ be any event depending on the output.  Then
\[
   \Prb_{S\sim\mathcal D_i^1}[\mathcal{A}(G_S)\in\mathcal E_i]
   \leq
   e^{2\eps}
   \Prb_{S\sim\mathcal D_i^0}[\mathcal{A}(G_S)\in\mathcal E_i]
   +(1+e^{\eps})\delta.
\]
\end{lemma}

\begin{proof}
We give an explicit coupling.  Sample $T\sim\mathcal D_i^0$ uniformly, then sample $j\in T$ uniformly, and set $S=T\setminus\{j\}\cup\{i\}.$ Then $S$ is uniform over $\mathcal D_i^1$.

Indeed, for any fixed $S$ with $i\in S$, it can arise by choosing $T=S\setminus\{i\}\cup\{j\}$
for some $j\in U\setminus S$, and there are $N-k$ such choices.  Since ${N-1\choose k} = \frac{N-k}{k}{N-1\choose k-1}$, each such $S$ is produced with probability $\frac{N-k}{{N-1\choose k}\, k}
   =
   \frac{1}{{N-1\choose k-1}}$. 
This concludes that $S$ is uniform over $\mathcal D_i^1$.

For every coupled pair $(S,T)$, the graphs $G_S$ and $G_T$ differ by exactly two private star edges: remove $\{0,j\}$ and add $\{0,i\}$.  Lemma~\ref{lem:two-edge} therefore gives
\[
   \Prb[\mathcal{A}(G_S)\in \mathcal E_i]
   \leq
   e^{2\eps}\Prb[\mathcal{A}(G_T)\in \mathcal E_i]
   +(1+e^{\eps})\delta.
\]
Averaging over the coupling proves the lemma.
\end{proof}

We now choose the event $\mathcal E_i$ to be the event that the decoder claims the private edge $\{0,i\}$ is present, i.e., $\mathcal E_i=\{Y:i\in \widehat S(Y)\}.$ We define
\[
   p_i^1=\Prb[i\in \widehat S(Y)\mid i\in S],
   \qquad
   p_i^0=\Prb[i\in \widehat S(Y)\mid i\notin S],
\]
where $S$ is uniformly random among all $k$-subsets of $U$ and $Y=\mathcal{A}(G_S)$.  Lemma~\ref{lem:coordinate-comparison} gives, for each $i\in U$,
\begin{equation}\label{eq:per-coordinate-dp}
   p_i^1\leq e^{2\eps}p_i^0+(1+e^{\eps})\delta.
\end{equation}

On the other hand, reconstruction accuracy implies that $p_i^1$ is large on average and $p_i^0$ is small on average.

\begin{lemma}\label{lem:tp-fp}
Let $S$ be uniformly random among all $k$-subsets of $U$, let $Y=\mathcal{A}(G_S)$, and let $\widehat S=\widehat S(Y)$.  Suppose that, for every $S$,
\[
   \Prb\bigl[\opnorm{Y-L_{G_S}}\leq \alpha\bigr]
   \geq 1-\beta.
\]
Then
\[
   \frac1N\sum_{i\in U}p_i^1
   \geq
   1-\beta-\frac{2\alpha^2}{k},
\]
and
\[
   \frac1N\sum_{i\in U}p_i^0
   \leq
   \frac{\beta k+2\alpha^2}{N-k}.
\]
\end{lemma}

\begin{proof}
The expected number of false negatives is
\begin{align*}
   \E[|S\setminus\widehat S|]
   &=
   \sum_{i\in U}\Prb[i\in S\text{ and }i\notin\widehat S] \\
   &=
   \sum_{i\in U}\Prb[i\in S]\Prb[i\notin\widehat S\mid i\in S] \\
   &=
   \frac{k}{N}\sum_{i\in U}(1-p_i^1).
\end{align*}
By Corollary~\ref{cor:expected-reconstruction},
\[
   \E[|S\setminus\widehat S|]
   \leq 2\alpha^2+\beta k.
\]
Therefore
\[
   \frac{k}{N}\sum_{i\in U}(1-p_i^1)
   \leq 2\alpha^2+\beta k,
\]
which rearranges to
\[
   \frac1N\sum_{i\in U}p_i^1
   \geq
   1-\beta-\frac{2\alpha^2}{k}.
\]
Similarly, the expected number of false positives is
\begin{align*}
   \E[|\widehat S\setminus S|]
   &=
   \sum_{i\in U}\Prb[i\notin S\text{ and }i\in\widehat S] \\
   &=
   \sum_{i\in U}\Prb[i\notin S]\Prb[i\in\widehat S\mid i\notin S] \\
   &=
   \frac{N-k}{N}\sum_{i\in U}p_i^0.
\end{align*}
Again using Corollary~\ref{cor:expected-reconstruction},
\[
   \frac{N-k}{N}\sum_{i\in U}p_i^0
   \leq
   2\alpha^2+\beta k.
\]
which rearranges to
\[
   \frac1N\sum_{i\in U}p_i^0
   \leq
   \frac{\beta k+2\alpha^2}{N-k}.
\]
\end{proof}

\subsection{Proof of Theorem~\ref{thm:spectral-lower-bound}}
With the tools developed in the previous sections, we are now ready to prove the following lemma, which further derives the $\Omega(\sqrt{n})$ lower bound for DP spectral release.

\begin{lemma}\label{lem:lowerbound-main}
Fix $\eps>0$ and $0<\beta<1/2$.  Let $N=n-2$, let $1\leq k<N/2$, and consider the connected graph family $\F_{n,k}$ defined in Section~\ref{sec:hard_family}. Suppose $\mathcal{A}$ is an $(\eps,\delta)$-edge-DP algorithm that outputs a symmetric matrix $Y$, and suppose that, for every $G\in\F_{n,k}$,
\[
   \Prb\bigl[\opnorm{Y-L_G}\leq \alpha\bigr]
   \geq 1-\beta .
\]
Then
\begin{equation}\label{eq:master-ineq}
   1-\beta-\frac{2\alpha^2}{k}
   \leq
   e^{2\eps}\frac{\beta k+2\alpha^2}{N-k}
   +(1+e^{\eps})\delta.
\end{equation}
In particular, for every constant $\eps>0$ and every constant failure probability $\beta<1/2$, by choosing $k=\lfloor \rho n\rfloor$ for a sufficiently small constant $\rho=\rho(\eps,\beta)>0$, every $(\eps,\delta)$-edge-DP algorithm with negligible $\delta=o(1)$ must incur
\[
   \alpha=\Omega(\sqrt n)
\]
on some connected graph with $m=\Theta(n)$.
\end{lemma}

\begin{proof}
For each $i\in U$, inequality~\eqref{eq:per-coordinate-dp} gives
\[
   p_i^1\leq e^{2\eps}p_i^0+(1+e^{\eps})\delta.
\]
Averaging over $i\in U$ yields
\[
   \frac1N\sum_{i\in U}p_i^1
   \leq
   e^{2\eps}\frac1N\sum_{i\in U}p_i^0
   +(1+e^{\eps})\delta.
\]
Lemma~\ref{lem:tp-fp} lower bounds the left-hand side and upper bounds the right-hand side, giving
\[
   1-\beta-\frac{2\alpha^2}{k}
   \leq
   e^{2\eps}\frac{\beta k+2\alpha^2}{N-k}
   +(1+e^{\eps})\delta.
\]
This proves~\eqref{eq:master-ineq}. Consequently, if we assume:
\begin{equation}\label{eq:parameter-choice}
   e^{2\eps}\frac{\beta k}{N-k}\leq \frac{1-\beta}{4}
   \qquad\text{and}\qquad
   \delta \leq \frac{1}{8(1+e^\eps)}\leq  \frac{1-\beta}{4(1+e^\eps)},
\end{equation}
where $\beta \leq 1/2$, then
\[1-\beta-\frac{2\alpha^2}{k}
   \leq
   \frac{1-\beta}{4}
   +
   \frac{2e^{2\eps}\alpha^2}{N-k}
   +
   \frac{1-\beta}{4}.
   \]
Therefore,
\begin{equation}\label{eq:alpha-lower}
   \alpha^2
   \geq
   \frac{1-\beta}{4\left(1/k+e^{2\eps}/(N-k)\right)}
   =\Omega_{\eps,\beta,k/N}(n).
\end{equation}

Take $k=\lfloor\rho N\rfloor$, where $\rho>0$ is a sufficiently small constant depending only on $\eps$ and $\beta$ so that the first inequality in~\eqref{eq:parameter-choice} holds for all large $n$. Substituting $k = \lfloor\rho N\rfloor$, the term $\left(1/k + e^{2\eps}/(N-k)\right)$ becomes $\Theta(1/N) = \Theta(1/n)$. The bound \eqref{eq:alpha-lower} then gives $\alpha^2 = \Omega(n)$  Since $m=n-1+k=\Theta(n)$, this completes the proof.
\end{proof}
\begin{proof}[Proof of Theorem~\ref{thm:spectral-lower-bound}]
It is enough to show that for any constants $\eps>0$ and $\beta<1/2$, there exists $c=c(\eps,\beta)>0$ such that, for all sufficiently large $n$, no $(\eps,\delta)$-edge-DP algorithm with $\delta = \frac{1}{8(1+e^\eps)}$ can output $Y$ with
\[
\Prb[\|Y - L_G\| \leq c\sqrt{n}] \geq 1-\beta.
\]
on every connected $n$-vertex graph $G$. Choose $k=\lfloor\rho(n-2)\rfloor$ for a sufficiently small constant $\rho=\rho(\eps,\beta)>0$, and apply Lemma~\ref{lem:lowerbound-main}.  The hard graphs have $m=n-1+k=\Theta(n)$ edges and are connected.  The lemma gives $\alpha\geq c'\sqrt n$ for some constant $c'=c'(\eps,\beta,\rho)>0$. Taking $c<c'$ gives the contradiction.
\end{proof}

\section*{Acknowledgments.}
We are very grateful to Mina Dalirrooyfard and Jalaj Upadhyay for their insightful discussions on related topics prior to the start of this project.
\newpage
\bibliographystyle{alpha}
\bibliography{privacy}
\end{document}